\documentclass[format=acmsmall]{acmart}

\usepackage{booktabs} 

\usepackage{hyperref}
\usepackage{ltl}
\usepackage{mathtools}
\usepackage{graphicx}

\newenvironment{proof-sketch}{\noindent{\scshape{Sketch of Proof.}}\hspace*{1em}}{\qed\bigskip}

\usepackage{multirow}
\usepackage[utf8]{inputenc}
\usepackage{graphicx}   
\usepackage{listings}
\usepackage{makecell}
\usepackage{url}    
\urldef{\mailsa}\path|{name.surname}@polimi.it| 

\usepackage{pdflscape}

\usepackage{tikz}
\usetikzlibrary{arrows,decorations.pathmorphing,backgrounds,positioning,fit,petri}
\usetikzlibrary{automata}
\usetikzlibrary{graphs}
\usetikzlibrary{shadows}
\usetikzlibrary{shapes,arrows}

\usepackage{xcolor}

\usepackage[inline]{enumitem}

\usepackage{todonotes}

\usepackage[framemethod=TikZ]{mdframed}
\makeatletter

\newcommand*{\intNum}{\ensuremath{\mathbb{Z}}}
\newcommand*{\naturalNum}{\ensuremath{\mathbb{N}}}
\newcommand*{\realNum}{\ensuremath{\mathbb{R}}}
\newcommand*{\partSet}{\wp}

\newcommand{\transitions}[1]{\ensuremath{T_{#1}}}
\newcommand{\netid}{\ensuremath{\mathcal{N}}}
\newcommand{\apvariable}{\ensuremath{AP_v}}
\newcommand{\clock}{\ensuremath{x}}
\newcommand{\clocks}{\ensuremath{X}}
\newcommand{\guardUniverse}{\ensuremath{\Gamma(\clocks)}}

\newcommand{\taNetwork}{\ensuremath{\mathcal{N}}}
\newcommand{\clockConstraint}{\ensuremath{\Gamma (\clocks)}}

\newcommand{\propositions}{\ensuremath{AP}}
\newcommand{\taid}{\ensuremath{\mathcal{A}}}
\newcommand{\ta}{$\taid=\langle AP,$ $X$, $Act_\tau,$ $Q,$ $q_0,$ $Inv,$ $L,$ $T \rangle$}
\newcommand{\tavar}{$\taid=\langle AP,$ $X$, $Act_\tau,$ $Int,$ $Q,$ $q_0,$ $v^0_{var},$ $Inv,$ $L,$ $T\rangle$}

\newcommand{\var}{n}
\newcommand{\variable}{\ensuremath{n}}

\newcommand{\constant}{\ensuremath{d}}
\newcommand{\transid}{\ensuremath{t}}

\newcommand{\timevalue}{\ensuremath{r}}

\newcommand{\variables}{\ensuremath{\mathit{Int}}}

\newcommand{\clockconstraint}{\ensuremath{\gamma}}
\newcommand{\variableconstraint}{\ensuremath{\xi}}
\newcommand{\varassignement}{\ensuremath{\mu}}
\newcommand{\nullevent}{\ensuremath{\tau}}
\newcommand{\actions}{\ensuremath{Act_\nullevent}}
\newcommand{\action}{\ensuremath{\alpha}}
\newcommand{\autstate}{\ensuremath{q}}

\newcommand{\transtavar}{\ensuremath{q \xrightarrow{\clockconstraint, \variableconstraint, \action, \resettedclocks, \varassignement} q^\prime}}
\newcommand{\conftranstavar}[1]{\ensuremath{l[ #1 ] \xrightarrow{\clockconstraint, \variableconstraint, \action, \resettedclocks, \varassignement} l^\prime[ #1 ]}}

\newcommand{\tanumber}{\ensuremath{K}}

\newcommand{\finalformula}{\ensuremath{\Phi_\mathit{sig}}}

\newcommand{\tracesymbol}{\ensuremath{\eta}}
\newcommand{\signalsymbol}{\ensuremath{M_\tracesymbol}}
\newcommand{\trace}{$(\loc_0,$ $v_{\mathrm{var},0},$ $v_0)  \xrightarrow{e_0} 
(\loc_1,$ $v_{\mathrm{var},1},$ $v_1)  \xrightarrow{e_1}
(\loc_{2},$ $v_{\mathrm{var},2},$ $v_{2}) \xrightarrow{e_{2}} \ldots$}

\newcommand{\rest}[1]{\ensuremath{\overleftarrow{#1}}}

\newcommand{\rclosed}[1]{\mathtt{edge}^{](}[#1]}

\newcommand{\proposition}{\ensuremath{p}}

\newcommand{\timeposition}{\ensuremath{i}}

\newcommand{\autindex}{\ensuremath{k}}
\newcommand{\autindextwo}{\ensuremath{h}}

\newcommand{\trans}{\ensuremath{t}}
\newcommand{\automatonIndex}{\ensuremath{k}}
\newcommand{\automatonIndextwo}{\ensuremath{h}}
\newcommand{\automatonIndexthree}{\ensuremath{g}}
\newcommand{\numberOfTA}{\ensuremath{K}}

\usepackage{changes}

\usepackage{subfigure}
\usepackage[T1]{fontenc}

\newboolean{showcomments}
\setboolean{showcomments}{true} 
\ifthenelse{\boolean{showcomments}}
{\newcommand{\nb}[2]{
  \fcolorbox{black}{yellow}{\bfseries\sffamily\scriptsize#1}
  {\sf\small$\blacktriangleright$\textit{#2}$\blacktriangleleft$}
 }
 
}
{\newcommand{\nb}[2]{}
 
}
\definechangesauthor[name={CM},color=red]{CM}
\definechangesauthor[name={MB},color=blue]{MB}
\definechangesauthor[name={PL},color=green]{PL}

\usepackage{setspace}

\newcommand{\setclockconstraint}{\ensuremath{\Gamma (X)}}

\newtheorem{remark}{Remark}
\newtheorem{txproposition}{Proposition}
\newcommand{\first}[1]{\mathbin{\text{\rotatebox[origin=c]{90}{$\multimapdot$}}}_{#1}}

\newcommand{\tr}{\mathtt{t}}
\newcommand{\loc}{\mathtt{l}}
\newcommand{\notr}{\natural}

\newcommand{\roland}{\textsc{Mitl$_{0,\infty}$BMC}}

\allowdisplaybreaks

\DeclareRobustCommand{\textblue}{\color{blue!70!black}}

\DeclareRobustCommand{\textblack}{\color{black}}

\newif\ifbothversions
\bothversionsfalse

\newif\ifextended
\extendedfalse
\extendedtrue

\newif\ifcoloredversions
\coloredversionsfalse

\ifbothversions
	\newcommand{\extended}[1]{\textbf{extended:}\textblue{#1}\textblack} 
	\newcommand{\notextended}[1]{\textbf{short:}\textgreen{#1}\textblack} 
\else
	\ifextended
		\ifcoloredversions
			\newcommand{\extended}[1]{\textblue{#1}\textblack} 
		\else
			\newcommand{\extended}[1]{#1}
		\fi
		\newcommand{\notextended}[1]{}
	\else
		\ifcoloredversions
			\newcommand{\notextended}[1]{\textgreen{#1}\textblack} 
		\else
			\newcommand{\notextended}[1]{#1} 
		\fi
		\newcommand{\extended}[1]{}
	\fi
\fi

\newif\ifblind
\blindfalse
\blindtrue

\ifblind
\newcommand{\Zot}{\textsc{Zot}}
\newcommand{\aetwozot}{\textsc{ae2Zot}}
\newcommand{\aetwosbvzot}{\textsc{ae2SBVZot}}
\newcommand{\logic}{\textsc{CLTLoc}}
\newcommand{\CLTLoc}{\textsc{CLTLoc}}
\newcommand{\NAME}{\textsc{TACK}}
\else
\newcommand{\Zot}{Zot}

\newcommand{\logic}{CLTLoc$_v$}
\newcommand{\CLTLoc}{CLTLoc}
\newcommand{\NAME}{TACK}
\fi

\newcommand{\channelsend}{\ensuremath{!}}
\newcommand{\channelreceive}{\ensuremath{?}}
\newcommand{\broadcastsend}{\ensuremath{\#}}
\newcommand{\broadcastreceive}{\ensuremath{@}}
\newcommand{\manytomanysend}{\ensuremath{\&}}
\newcommand{\manytomanyreceive}{\ensuremath{*}}

\newcommand{\clockvaluationfunction}{\ensuremath{v}}
\newcommand{\clockvaluation}{\ensuremath{\clockvaluationfunction: \clocks \rightarrow \realNum_{\geq 0} }}
\newcommand{\variablevaluationfunction}{\ensuremath{v_{\mathrm{var}}}}
\newcommand{\variablevaluationfunctionpar}[1]{\ensuremath{v_{\mathrm{var},#1}}}
\newcommand{\clockvaluationfunctionpar}[1]{\ensuremath{v}_{#1}}
\newcommand{\edge}{\ensuremath{\mathrm{edge}^{](}}}

\newcommand{\blindonoff}[2]{\ifblind #1\else #2\fi}
\newcommand{\invariants}{\ensuremath{\Gamma(X)}}

\newcommand{\resettedclocks}{\ensuremath{\zeta}}
\newcommand{\universeOfActions}{\ensuremath{Act_\tau}}

\newcommand{\map}{\ensuremath{\rho}}
\newcommand{\eventtime}{\ensuremath{\Upsilon}}

\setlist[enumerate,1]{label=(\theenumi)}

\usepackage{pifont}
\newcommand{\cmark}{\ding{51}}
\newcommand{\xmark}{\ding{55}}

\begin{document}

\title{Verifying MITL formulae on Timed Automata considering a Continuous Time Semantics}

\author{Claudio Menghi}
\orcid{orcid}
\affiliation{
	\institution{SnT - Interdisciplinary Centre for Security, Reliability and Trust, University of Luxembourg}
	\streetaddress{H29 Avenue John F. Kennedy}
	\city{Gothenburg}
	\postcode{1855}
	\country{Luxembourg}}
\email{claudio.menghi@uni.lu}

\author{Marcello M. Bersani}
\orcid{orcid}
\affiliation{
	\institution{Dipartimento di Elettronica, Informazione e Bioingegneria, Politecnico di Milano}
	\streetaddress{via Golgi 42}
	\city{Milan}
	\postcode{20133}
	\country{Italy}}
\email{marcello.rossi@polimi.it}

\author{Matteo Rossi}
\orcid{orcid}
\affiliation{
	\institution{Dipartimento di Elettronica, Informazione e Bioingegneria, Politecnico di Milano}
	\streetaddress{via Golgi 42}
	\city{Milan}
	\postcode{20133}
	\country{Italy}}
\email{marcello.rossi@polimi.it}

\author{Pierluigi San Pietro}
\orcid{orcid}
\affiliation{
	\institution{Dipartimento di Elettronica, Informazione e Bioingegneria, Politecnico di Milano}
	\streetaddress{via Golgi 42}
	\city{Milan}
	\postcode{20133}
	\country{Italy}}
\email{pierluigi.sanpietro@polimi.it}

\begin{abstract}
Timed Automata (TA) is de facto a standard modelling formalism to represent systems when the 
interest is the analysis of their   behaviour as time progresses.
This  modelling formalism is mostly used for checking whether the behaviours of a system satisfy a set of properties  of interest.
Even if efficient model-checkers for Timed Automata exist, these tools are not easily configurable.
First, they are not designed to easily allow   adding new Timed Automata constructs, such as new synchronization mechanisms or communication procedures, but they assume a fixed set of Timed Automata constructs.
Second,  they usually do not support the full Metric Interval Temporal Logic (MITL) and rely on a precise semantics for the logic in which the property of interest is specified which cannot be easily modified and customized.
Finally,  they do not easily allow  using different solvers that may speed up verification in different contexts.

This paper presents a novel technique to perform model checking of \emph{full} Metric Interval Temporal Logic (MITL) properties on TA.
The technique relies on the translation of both the TA and the MITL formula into an intermediate Constraint LTL over clocks (CLTLoc) formula
which is verified through an available decision procedure.
The technique is flexible since the intermediate logic allows the encoding of new semantics as well as new TA constructs, by just 
adding new \CLTLoc\ formulae. 
Furthermore, our technique is not bound to a specific solver as the intermediate \CLTLoc\ formula can be verified using different 
procedures. 
 \end{abstract}

\maketitle

\section{Introduction}
\label{sec:Introduction}
Model checking is an automatic technique 
to verify whether a model of a system satisfies a property of interest.
Various formalisms 
have been proposed for representing the model and its properties, often in terms of state machines and temporal logics,
both having specific peculiarities depending on the designer's goals and tool availability.

Timed Automata~\cite{alur1994theory} (TA) are one of the most popular formalisms to describe system behavior when real time constraints are important.
Various tools are available to verify TA: 
Kronos~\cite{Yovine97}, the de facto standard tool Uppaal~\cite{larsen1997uppaal}, RED~\cite{Wang01},   \roland~tool~\cite{kindermann2013bounded}  and MCMT~\cite{CGR10} (though the latter can only perform reachability analysis of---parametric---networks of TA).

We believe that  novel model checking tools for TA should address three main challenges:  
(\texttt{C1}) providing different semantics of TA including the \emph{continuous time}; 
(\texttt{C2}) supporting high level expressive complex logics that easily allow the specification of the properties of interest, such as the \emph{full} Metric Interval Temporal Logic (MITL); 
and (\texttt{C3}) being extensible, i.e., allowing users to add new constructs easily, such as adding new synchronization mechanisms or communication procedures for TA. 
The main issues related to those challenges are summarized hereafter.

\texttt{C1}. The paradigm of time that is overwhelmingly adopted in practice is based on timed words~\cite{alur1994theory},---i.e., 
infinite words where each symbol is associated with a real-valued time-stamp--- and most tools 
(except~\cite{kindermann2013bounded}) are founded on such semantics.
The so-called \emph{signal-based semantics} is a different interpretation, where each instant of a dense temporal domain (e.g., $\realNum_{\geq 0}$) is 
associated with a state, called a signal.
Signals are more expressive than timed words (as proved in~\cite{DSouzaP07}), thus allowing a more precise representation of the system state over time. 
In particular, if a signal changes its value at an instant $t$,
it is possible to specify the value of the signal both in $t$ (the ``signal edge'') and in arbitrary small neighborhoods of $t$. This allows, for instance, 
to represent the location of an automaton both just before and immediately after an instantaneous state transition.
Despite its greater expressiveness, signal-based semantics has been so far confined mainly to theoretical investigations~\cite{alur1996benefits,BCM10,OW08,DSouzaP07} 
and seldom used in practice~\cite{kindermann2013bounded}, due the difficulty in developing a feasible decision procedure.
More precisely,  Kindermann et al.~\cite{kindermann2013bounded}  implemented a decision procedure for BMC of TA against MITL$_{0,\infty}$ which is based on the so-called ``super-dense'' time (also adopted by Uppaal).
Under a super-dense time assumption, a TA can fire more than one transition in the same (absolute) time instant; thus, two or more transitions can be fired one after the other and produce many simultaneous but distinct configuration changes such that time does not progress. 
Super-dense time is a modeling abstraction to represent systems that are much faster then the environment they operate in, so their reaction to external events has a negligible delay.
In the current work, the signal-based semantics is not ``super-dense'', i.e., at any time instant each TA is in exactly one state. 
This choice is mainly dictated by the use of~\cite{BRS15b} to translate MITL to CLTLoc, which is defined over the a more ``traditional'' dense-time; still, CLTLoc may be extended to super-dense time.

\texttt{C2}. Temporal Logics with metric of time, such MITL~\cite{alur1996benefits}, have been proposed to specify real-time properties,
but they are not fully supported by TA verification tools, that typically provide just some baseline functionalities to address reachability problems 
(safety assessment) or to perform model-checking of temporal logics without metric (e.g., LTL, CTL) or of fragments of Timed CTL~\cite{Yovine97}.
For example, Uppaal~\cite{larsen1997uppaal} supports only a limited set of reachability properties.
However, its ability of expressing that a certain condition triggers a reaction within a certain amount of time provides a clear improvement over being able only to specify that a reaction will eventually occur.
The gap of almost 20 years between the proof of the decidability of MITL~\cite{alur1996benefits} and its applicability in practice can be justified 
mainly with the practical complexity of the underlying decision procedure, hampering the development of efficient tools, until more
recent developments of new decision procedures, typically based on faster SMT-solvers. In fact, 
both~\cite{kindermann2013bounded,BRS13a} developed a decision procedure for the satisfiability of MITL, a problem which very recently was also tackled by~\cite{BGHM17}.
In particular, ~\cite{kindermann2013bounded} proposed a verification procedure for a fragment of MITL, namely MITL$_{0,\infty}$, on TA but under 
a semantics based on super-dense signals. It is however still unknown whether  MITL$_{0,\infty}$ 
is equivalent to MITL under the latter semantics.
To the best of the authors' knowledge, a verification procedure supporting full MITL over (standard) signals is still not available. 
A proof of the  language equivalence of TA and \CLTLoc\ over timed words  was given for the first time in~\cite{bersani2016logical}; however, 
the translation presented therein did not consider signals and had only the purpose of proving the equivalence, rather than being intended to be implemented in a tool. For instance, it makes use of many additional clocks
that would hinder the performance of any decision procedure. 
Those limitations fostered the definition of a new, more practical  translation, which is also radically different.
The new encoding has been devised to be as optimized and extensible as possible, rather than being intended to prove language-theoretical results.
Moreover, it also supports networks of TA, whose traces are interpreted for the evaluation of MITL formulae over atomic propositions and arithmetical formulae of the form $\variable\sim \constant$, where $\variable$ is an integer variable manipulated by the automata; to this end, the new encoding allows the representation of the signals associated with the atomic propositions  on the locations and with the integer variables elaborated by the network.
The new encoding also includes three synchronization primitives and allows the representation of the signal edges at the instant where transitions are taken.

\texttt{C3}. Even if a variety of tools supporting the analysis of TA and network of TA is available, they usually 
are not easy to tailor and extend.
\begin{enumerate*}
\item 
They only provide  a fixed set of modeling constructs that support designers in modeling the system under development, but which are not easily modifiable and customizable.
Typical examples are discrete variables (often on finite domains) as well as some communication and synchronization features among different TA.
For example, Uppaal 
provides designers with binary and broadcast synchronization primitives whereas RED  offers sending/receiving communication features via finite FIFO channels.
However, often new modeling requirements may prompt the designers to formulate 
specific communication/synchronization features, also based on data-structures such as queues, stacks, priority mechanisms.
Common model checkers do not explicitly support extending their features and constructs in the above directions,
since this could cause a significant variation of the underlying semantics. Ad hoc modifications of a tool are often possible, but they may require a deep knowledge of the tool internals, whose 
software implementation may be quite complex (depending also in the architecture and the programming language) and scantly documented. 
\item The existing model checkers are typically solver-dependent, since they rely on a strong relation between the problem domain and the solution domain---i.e., respectively, the models to be verified and the input language that is used by a  verification engine.
\item 
The cited tools explicitly support only TA, but not other timed formalisms such as, for instance, Time Petri Nets~\cite{Merlin74,FMMR12}, unless an
ad hoc front-end is developed (as for instance done by the Romeo toolkit~\cite{lime2009romeo}).
Hence, they not easily allow multi-formalism analysis~\cite{BFPR09}, e.g., to design systems with heterogeneous components which are more naturally modeled  by different formalisms.
\end{enumerate*}

\textbf{Contribution.} 
This paper describes a novel technique to model check networks of TA over properties expressed in the \emph{full} MITL  over signals, by relying on a purely logic-based approach. 
The technique is exemplified in the diagram of Fig.~\ref{fig:approch}.
It is based on the solution presented in~\cite{BRS15b} to translate both a MITL formula and a TA into an intermediate logical language, which is then encoded into the language of the underlying solver.
This intermediate level has thus a similar role to the Java Byte code for Java program execution on different architectures.
The advantages are that on the one hand, new TA constructs, logic formalisms or semantics can be dealt by defining new encodings into the intermediate language;
on the other hand, the intermediate language can independently be ``ported'' to different (possibly more efficient) solvers, by translating into the respective solver languages.

\begin{figure}[t]
\begin{center}
\begin{tikzpicture}[->,>=stealth',shorten >=1pt,auto,node distance=1cm,
thick,main node/.style={rectangle,draw,font=\sffamily\Large\bfseries}
]

\node[rectangle]          (l2) []  {\CLTLoc};
\node[]         (l0) [above left= of l2] {TA};
\node[]          (l1) [above right= of l2]  {MITL $\phi$};
\node[]          (l3) [below= of l2]  {\CLTLoc\ solver};

\path[->] (l0) edge[] node[text width=2cm,align=left,below] {Sect.\ref{sec:TA2CLTLoc}} (l2);
\path[->] (l1) edge[] node [text width=2cm,align=left]{\cite{BRS15b}} (l2);
\path[->] (l2) edge[] node [text width=2cm,align=left]{\cite{bersani2016tool}} (l3);
\end{tikzpicture}
\caption{A generic framework for checking the satisfaction of MITL formulae on TA.}
\label{fig:approch}
\end{center}
\end{figure}
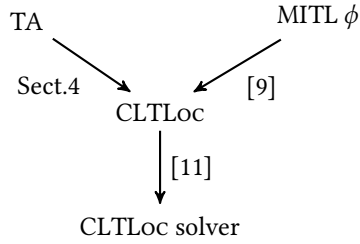

A TA and a MITL formula are translated into \blindonoff{\CLTLoc}{Constraint LTL over clocks (\CLTLoc)}, a metric temporal logic~\cite{bersani2016tool}.
\CLTLoc\ is a decidable extension of Linear Temporal Logic (LTL) including real-valued variables that behave like TA clocks. 
The satisfiability of \CLTLoc\ can be checked by using different procedures; a bounded approach based on SMT-solvers is available as part of the \Zot~formal verification tool~\cite{BPR16}.
This intermediate language easily allows for different semantics of TA such as, for instance, 
the signal edges that are generated by the TA when transitions are fired (see Sect.~\ref{sec:tasemantics}).
Moreover, different features of the TA modeling language can be introduced by simply adding or changing formulae in the \CLTLoc\ encoding.
As an example, finite queues or data-structures can be easily included as long as the new features can be expressed in terms of \CLTLoc\ formulae.
The \CLTLoc\ formula encoding  the network of TA and the MITL property is modular, in the sense that the parts that 
translate the MITL prpperty are separated from the formulae translating the TA network.
Moreover, each aspect of the semantics of the (network of) TA is isolated in a specific formula, with only few interconnecting points with the other parts.  
Therefore, each part of the resulting translation is self-reliant, thus easily allowing changes or extensions.

The technique presented in this work is implemented in a Java tool, called \blindonoff{\NAME}{\NAME\ (Timed Automata ChecKer)} (\url{https://github.com/claudiomenghi/TACK}), which is built on the \textrm{QTLSolver} (\url{https://github.com/fm-polimi/qtlsolver}), and extends the translation of~\cite{bersani2016tool} to deal with a network of TA and to add a new front-end for the specification of the network and its properties.
\NAME\ takes as input a (network of) TA, described with a syntax compatible with Uppaal, and the MITL property to be verified.
Unlike Uppaal, TA and MITL are interpreted according to the \emph{signal-based semantics}.
The \CLTLoc~ formula produced by \NAME\ is then fed to \Zot\ for automated verification.

To evaluate the benefits that ensue from the adoption of an intermediate language, this work shows how to dealt with different signal-based semantics, synchronization primitives and liveness conditions. 
Furthermore, to show the flexibility that is yielded by decoupling  the model-checking problem and the resolution technique, different solvers are employed for verifying the intermediate \CLTLoc~ encoding.
The efficiency of technique is evaluated over  
some standard benchmarks, namely the  Fischer (see e.g.~\cite{abadi1994old}),  the CSMA/CD  (see e.g.~\cite{CSMACD}) and the Token Ring (see e.g.~\cite{jain1994fddi}) protocols.

\vskip 0.1in  
The paper is structured as follows.
Section~\ref{sec:Background} presents the background and the notation used in the rest of this work.
Section~\ref{sec:tasemantics} introduces the continuous time semantics of TA.
Section~\ref{sec:TA2CLTLoc} presents the algorithm to convert a TA in a \CLTLoc\ formula.
Section~\ref{sec:checkingMITLI} describe the model checking algorithm to verify MITL formulae on automata with time.
Section~\ref{sec:evaluation} evaluates \NAME  and 
discusses the experimental results. 
Section~\ref{sec:conclusion} concludes.

\section{Background}
\label{sec:Background}
This section presents TA (enriched with integer-valued variables and  synchronization), MITL and CLTLoc.
\subsection{Timed automata}
\noindent Let \clocks\ be a finite set of \emph{clocks} with values in \realNum.
 \clockConstraint\ is the set of \emph{clock constraints}  over \clocks\ defined by the syntax 
$\gamma \coloneqq  \clock \sim c \mid \neg \gamma \mid \gamma \wedge \gamma$, where $\sim \in \{ <, =\}$, $x \in X$ and $c \in \naturalNum$.
Let $Act$ be a set of events,
 $\universeOfActions$ is the set $Act \cup \{ \tau \}$, where $\tau$ is used to indicate a null event.

\begin{definition}[Timed Automaton]
Let $AP$ be a set of atomic propositions,
 \clocks\ be a set of clocks and 
 $Act$ be  a set of events. A
\emph{Timed Automaton} is a tuple 
\ta, where:
\begin{enumerate*}
\item[]  $Q$ is a finite set of control states (also called locations);
\item[] $q_0 \in Q$ is the initial state;
\item[] $Inv: Q\rightarrow \invariants$ is an invariant assignment function;
\item[] $L:$ $Q$ $\rightarrow \wp(AP)$ is a function labeling the states in $Q$ with elements of $AP$;
\item[] $T \subseteq_{\textit{fin}} Q \times Q \times \guardUniverse \times \universeOfActions \times \partSet(\clocks)$ is a finite set of transitions.
\end{enumerate*} 
\end{definition}

A transition $t=(q,q',\gamma,\action,\resettedclocks)\in T$ is written as $q \xrightarrow{\clockconstraint, \action, \resettedclocks} q^\prime$; the notations $t^-$, $t^+$,
$t_g$,  $t_e$, $t_s$  indicate, respectively,  the source $q$, the destination $q^\prime$, the clock constraint $\clockconstraint$, the event $\action$ and the set of clocks $\resettedclocks$ to be reset when firing the transition.
Fig.~\ref{fig:TaExample} shows a simple example of TA.

\vskip 0.05in   
Let $Int$ be a finite set of \emph{integer variables} with values in $\intNum$ and $\sim \in \{ <, =\}$; 
$Assign(Int)$ 
is
the set of assignments  of the form $n := exp$, where $n \in Int$ and $exp$ is an arithmetic expression over the integer variables and elements of $\intNum$.
$\Gamma (Int)$ 
is
 the set of \emph{variable constraints $\gamma$} over $Int$ defined as 
$\gamma \coloneqq n \sim c \mid$ $n \sim n' \mid$ $\neg \gamma \mid$ $\gamma \wedge \gamma$, where  $n$ and $n'$ are integer variables  and $c \in \intNum$.

\begin{definition}[TA with Variables]
\label{def:tawithvariables}
Let $AP$ be a set of atomic propositions, \clocks\ be a set of clocks , $Act$ be a set of events  and $Int$ be a set of integer variables. A
\emph{Timed Automaton with Variables} is a tuple \tavar,
 where:
\begin{enumerate*}
\item[]  $Q$ is a finite set of control states (also called locations);
\item[] $q_0 \in Q$ is the initial state;
\item[] $\variablevaluationfunction^0: Int \rightarrow \intNum$ assigns each variable with a value in $\intNum$; 
\item[] $Inv: Q\rightarrow \Gamma(X)$ is an invariant assignment function;
\item[] $L:$ $Q$ $\rightarrow \wp(AP)$ is a function labeling the states in $Q$ with elements of $AP$;
\item[] $T \subseteq_{\textit{fin}} Q \times Q \times \Gamma(X) \times \Gamma(Int) \times Act_\tau \times \partSet(X)\times \partSet(Assign(Int)) $ is a finite set of transitions.
\end{enumerate*}
\end{definition}

\begin{figure}[t]
	\subfigure[An example of TA.]{\begin{tikzpicture}[->,>=stealth',shorten >=1pt,auto,node distance=1.1cm,
                    thick,main node/.style={circle,draw,font=\sffamily\Large\bfseries}]

\node[state,align=center, minimum size=1.4cm]          (l2) []  {$q_2$\\ $c$};
\node[state,align=center,  minimum size=1.4cm]         (l0) [above left= of l2] {$q_0$\\ $a$};
\node[state,align=center, minimum size=1.4cm]          (l1) [above right= of l2]  {$q_1$\\ $x \leq 5$};

\draw[<-]               (l0) -- node[above] {} ++(-1,0cm);
\path[->]                 (l0)    edge[]        node[text width=2cm,align=center] {
\textbf{sync}: $e_1$\\
\textbf{guard}: $x < 5$ 
} (l1);
\path[->]                (l1)    edge[]        node [text width=1.5cm,align=center]{
\textbf{sync}: $e_2$\\
} (l2);
\path[->]                (l2)    edge[]        node[text width=2.5cm,align=center]{
\textbf{sync}:  $e_3$\\
\textbf{guard}:  $x=10$\\ 
\textbf{assign}: $x$
} (l0);

\end{tikzpicture} \label{fig:TaExample}}
	\subfigure[An example of TA with Variables.]{\begin{tikzpicture}[->,>=stealth',shorten >=1pt,auto,node distance=1.1cm,
                    thick,main node/.style={circle,draw,font=\sffamily\Large\bfseries}]

\node[state,align=center, minimum size=1.4cm]          (l2) []  {$q_2$\\ $c$};
\node[state,align=center,  minimum size=1.4cm]         (l0) [above left= of l2] {$q_0$\\ $a$};
\node[state,align=center, minimum size=1.4cm]          (l1) [above right= of l2]  {$q_1$\\ $x \leq 5$};

\draw[<-]               (l0) -- node[above] {} ++(-1,0cm);
\path[->]                 (l0)    edge[]        node[text width=2.2cm,align=center] {
\textbf{sync}: $e_1$\\
\textbf{guard}: $x < 5$ \\
\textbf{assign}:  $\variable:= 2$
} (l1);
\path[->]                (l1)    edge[]        node [text width=2cm,align=center]{
\textbf{sync}:  $e_2$\\
\textbf{assign}:  $\variable:= 1$ 
} (l2);
\path[->]                (l2)    edge[]        node[text width=2.5cm,align=center]{
\textbf{sync}:  $e_3$\\
\textbf{guard}:  $x=10$\\ 
\textbf{assign}: $x$, $\variable:=0$
} (l0);

\end{tikzpicture} \label{fig:TaWithVariableExample}}
	\caption{The TA in \subref{fig:TaExample} has  three locations, $q_0$, $q_1$, $q_2$, and one clock \clock.
		The transition from $q_2$ to $q_0$ is labeled with guard $\clock=10$. 
		When the transition is taken, clock \clock\ is reset---i.e., it is set to $0$. 
		Location $q_1$ is associated with invariant $\clock \leq 5$. 
		Locations $q_0$ and $q_2$ are labeled with atomic propositions $a$ and $c$, respectively.
		The TA in \subref{fig:TaWithVariableExample} is the same as the one of \subref{fig:TaExample}, except for the presence of integer variable \variable , which is set to $0$, $1$ or $2$ depending on the transition taken.}
\end{figure}

A transition is written as \transtavar\ where \variableconstraint\ is a constraint of $\Gamma(Int)$ and \varassignement\ is {a set of assignments} from $\partSet(Assign(Int))$.
The notations $t_{d}$ and $t_u$ indicate, respectively, the variable constraint \variableconstraint\ and {the set of assignments} \varassignement\ associated with a transition $t$.
An example of TA with Variables is presented in Fig.~\ref{fig:TaWithVariableExample}.

\begin{remark}
{A set of assignments}  $\varassignement \in \partSet(Assign(Int))$ might be inconsistent, i.e., 
it may assign different values to the same variable.
For example, $\varassignement=\{x=2, x=3\}$ is inconsistent since two values are assigned to variable \clock. In this case, a transition  associated with {$\varassignement$} cannot be fired.
\end{remark}

\vskip 0.05in  
When networks of TA are considered, the event symbols labeling the transitions are used to synchronize automata. 
Every event symbol $\action \in Act$ is associated with one communication channel, which can be identified with the event symbol itself, i.e., channel $\action$.
The set of actions \actions\  is defined as $\actions= \{   \tau  \} \cup   \{ Act \times Sync \}$, where $Sync$ is a set of synchronization primitives and 
$\tau$ indicates that no synchronization primitive is associated with the transition.
In this work, $Sync$ is restricted to $\{ \channelsend, \channelreceive, \broadcastsend, \broadcastreceive, \manytomanysend, \manytomanyreceive \}$ where the symbols $!$ and $?$ indicate that a TA emits and receives an event, respectively,  $\#$ denotes a broadcast synchronization sender and $@$ denotes a broadcast synchronization receiver, \manytomanysend\ denotes a one-to-many synchronization sender and \manytomanyreceive\ denotes a one-to-many communication receiver.
Symbols 
$\action \channelsend$, $\action \channelreceive$, $\action \broadcastsend$, $\action \broadcastreceive$, $\action \manytomanysend$ and $\action \manytomanyreceive$ indicate the element 
$(\action, \channelsend)$,  $(\action, \channelreceive)$, $(\action, \broadcastsend)$, $(\action, \broadcastreceive)$, $(\action,\manytomanysend)$ and $(\action,\manytomanyreceive)$ such that 
$(\action, \channelsend)$,  $(\action, \channelreceive)$, $(\action,\broadcastsend)$, $(\action, \broadcastreceive)$,  $(\action,\manytomanysend)$ and $(\alpha,\manytomanyreceive)$ is contained in the set $\{ Act \times Sync \}$.

\begin{definition}[Network of TA]
A network $\netid$ of TA is a set $\netid=\{ \mathcal{A}_1, \ldots, \mathcal{A}_K \}$ of TA  defined over the same set of atomic propositions $AP$, actions \actions, variables $Int$ and clocks \clocks .
\end{definition}

\begin{remark}
Let
\taNetwork\ be a network of TA defined over the set of clocks
 \clocks ;
 a clock $\clock \in \clocks$ is a \emph{local clock}  of an automaton $\mathcal{A}_i \in \netid$ if 
 \clock\ is used in the invariants, guards or resets of  $\mathcal{A}_i$ and it does not exist another automaton $\mathcal{A}_j \in \netid$, with $\mathcal{A}_i \neq \mathcal{A}_j$, using \clock\ in its invariants, guards or resets.
Let $Int$ be  a set of variables  of \taNetwork ,
a variable $\variable \in Int$  is a \emph{local variable} of an automaton $\mathcal{A}_i \in \netid$ if 
\variable\ is used in the guards or assignments of  $\mathcal{A}_i$ and it does not exist another automaton $\mathcal{A}_j \in \netid$, with $\mathcal{A}_i \neq \mathcal{A}_j$, using the variable  \variable\ in its  guards or assignments.
\end{remark}

\subsection{Metric Interval Temporal Logic}
\label{sec:mitl}
An interval $I$ is a convex subset of  $\realNum_{\geq 0}$ of the form $\langle a, b \rangle$  or $\langle a, \infty)$, where
\begin{enumerate*}
\item[] $a \leq b$ are non-negative integers;
\item[]  symbol $\langle$ is either $($ or $[$;
\item[]  symbol $\rangle$ is either $)$ or $]$.
\end{enumerate*}

The syntax of (well-formed) MITL formulae is defined by the grammar
\begin{align}
\phi \coloneqq \alpha \mid \phi \LTLand \phi \mid \neg \phi \mid \phi \LTLuntil_I \phi \nonumber
\end{align}
where $\alpha$ are atomic formulae.
Since MITL is here used to specify properties of TA enriched with variables, atomic formulae $\alpha$ are either propositions of $AP$ or formulae of the form $\variable\sim \constant$, where $\variable\in \variables$, $\constant \in \intNum$ and $\sim \in \{<,=\}$. In the following, set \apvariable\  indicates the universe of the formulae of the form $\variable\sim \constant$.

The semantics of MITL is defined w.r.t. signals. 
Let $\intNum^\mathit{Int}$  be
the set of total functions from $Int$ to $\intNum$.
A \emph{signal} is a total function $M: \realNum_{\geq 0} \rightarrow \wp (AP) \times \intNum^\variables$. 
Let $M$ be a signal;
the semantics of an MITL formula is defined as follows.
\begin{align*}
&M, t \models p && \text{iff} && M(t)=(P, \variablevaluationfunction) \text{ and }  p\in P  \nonumber \\ 
&M, t \models n\sim d && \text{iff} &&   M(t)=(P,\variablevaluationfunction) \text{ and } \variablevaluationfunction(\variable)\sim \constant
\\
&M, t \models \neg p && \text{iff} && M, t \models \neg \phi \nonumber \\
&M, t \models \phi \LTLand \psi && \text{iff} && M, t \models \phi \text{ and } M, t \models \psi \nonumber \\
&M, t \models \phi \LTLuntil_I \psi && \text{iff} && \exists t^\prime > t, t^\prime - t \in I, M, t^\prime \models \psi  \nonumber  \text{ and } \forall t^{\prime \prime} \in (t, t^\prime)\ M, t^{\prime \prime} \models \phi
\end{align*}

An MITL formula $\phi$ is \emph{satisfiable} if there exists a signal $M$, such that $M,0 \models \phi$. 
In this case, $M$ is called a {\em model} of $\phi$.

\subsection{Constraint LTL over clocks}\label{sec:cltloc}
\blindonoff{\CLTLoc\ }{Constraint LTL over clocks (CLTLoc)} is a temporal logic where formulae are defined over a finite set of atomic propositions and a set of dense variables over $\realNum_{\geq 0}$ representing clocks.  
\CLTLoc\ has been recently extended by supporting expressions over a set of arithmetical variables~\cite{marconi2016towards}.
\logic{} is the intermediate language that is adopted to solve the model-checking problem of TA with MITL specifications.

\logic{} allows for two kinds of atomic formulae: over clocks and over arithmetical variables.
An atomic formula over a clock $x$ is for instance $x < 4$, whereas an atomic formula over arithmetical variables is for example $n + m < 4$, with $n,m\in\intNum$.
Similarly to TA, a clock $x$ measures the time elapsed since its last ``reset''. 
\logic{} also exploits the ``next'' $\LTLx$ modality applied to integer variables~\cite{DD07}: if $n$ is an integer variable, the term $\LTLx(n)$ represents the value of $n$ in the next 
position.

\renewcommand{\exp}{\mathit{exp}}

Let $X$ be a finite set of clocks
and 
 $Int$ be a finite set of integer variables; 
formulae of \CLTLoc\  with counters  are defined by the grammar:
\begin{align}
\phi \coloneqq p \mid x \sim c \mid \exp_1 \sim \exp_2 \mid \LTLx(n) \sim\exp
\mid \phi \wedge \phi \mid \neg \phi \mid \LTLx \phi \mid  \phi \LTLuntil \phi  \nonumber
\end{align}
where $p \in AP$,  $c \in \naturalNum$,  $x \in X$, 
$exp$, $\exp_1$ and $\exp_2$ are arithmetic expressions over the sets $Int$ and $\intNum$, $n \in Int$ and
$\sim$ is a relation in $\{ <, =\}$. 
$\LTLx$, $\LTLuntil$  are the usual ``next" and ``until" operators of LTL.
Modalities such as ``eventually'' ($\LTLf$), ``globally'' ($\LTLg$), and ``release'' ($\LTLrelease$) are defined as usual.
Symbol $\top$ (true) abbreviates $(p \lor \neg p)$, for some $p \in AP$.

The strict linear order ${(\naturalNum, <)}$ is the standard representation of positions in time.
The interpretation of \emph{clocks} is defined by means of a clock valuation $\sigma: \naturalNum \times X \rightarrow \realNum_{\geq 0}$ assigning, 
for every position $i \in \naturalNum$, a real value $\sigma(i, x)$ to each clock $x \in X$.
As in TA, a clock $x$ measures the time elapsed since the last time when $x=0$, i.e., the last ``reset" of $x$.
The semantics of time evolution is {\em strict},
namely the value of a clock must strictly increase in two adjacent time positions, unless it is reset (i.e., for all $i \in \naturalNum$, $x \in X$, 
it holds that $\sigma(i+1,x)>\sigma(i,x)$, unless $\sigma(i+1,x)=0$ holds)\footnote{As discussed in the following this assumption does not allow us to capture the super-dense semantics of TAs.}.
\notextended{}
To ensure that time strictly progresses at the same rate for every clock, $\sigma$ must satisfy the following condition: for every position $i \in \naturalNum$, there exists a ``time delay" $\delta_i >0$ such that for every clock $x \in X$:
\begin{equation}
\sigma(i+1, x) = \begin{cases} \sigma(i,x)+\delta_i & \text{progress} \\ 0 & \text{reset}\ x \end{cases} \nonumber
\end{equation}
If this is the case, then $\sigma$ is called a {\em clock assignment}. 
The initial value $\sigma (0,x)$ may be any non-negative value.
Moreover, a clock assignment is such that $\sum_{i \in \naturalNum} \delta_i = \infty$, i.e., time is always progressing.

The interpretation of \emph{variables} is defined by a mapping $\iota: \naturalNum \times Int \rightarrow \intNum$ assigning, for every position $i \in \naturalNum$, a value in $\intNum$ to each variable of set $Int$.
Let $\iota$ be a valuation and $i$ be a position;
$\exp(\iota, i)$ indicates the evaluation of  $\exp$ obtained by replacing each arithmetical variable $n \in Int$ that occurs in $\exp$ with value $\iota(i, n)$.
An interpretation of \logic\ is a triple $(\pi, \sigma, \iota)$, where 
$\pi : \naturalNum \rightarrow \wp (AP)$ is a mapping associating a set of propositions with each position $i \in \naturalNum$, 
$\sigma$ is a clock assignment and 
$\iota$ is a valuation of variables.
Let $x$ be a clock, $n$ be a variable and $c$ be a constant in $\naturalNum$, the semantic of \logic\ at a position $i \in \naturalNum$ over an interpretation $(\pi, \sigma, \iota)$ is defined as follows (standard LTL modalities are omitted):
\begin{align}
& (\pi, \sigma, \iota),i \models x \sim c && \text{iff} &&  \sigma(i, x) \sim c& \nonumber\\
& (\pi, \sigma, \iota),i \models \exp_1 \sim \exp_2 && \text{iff} &&  \exp_1(\iota, i) \sim \exp_2(\iota, i) & \nonumber\\
& (\pi, \sigma, \iota),i \models \LTLx(n) \sim exp && \text{iff} && \iota(i+1,n) \sim \exp(\iota,i) & \nonumber \\
& (\pi, \sigma, \iota),i \models a && \text{iff} && a \in \pi(i) &\nonumber
\end{align}
 
A \logic{} formula $\phi$ is \emph{satisfiable} if there exist an interpretation $(\pi, \sigma, \iota)$ such that $(\pi, \sigma, \iota), 0 \models \phi$. In this case,   $(\pi, \sigma, \iota)$ is 
called a {\em model} of $\phi$, written $(\pi, \sigma, \iota) \models \phi$.
It is easy to see that \logic{} is undecidable, as it can encode a 2-counter machine; however, in this work a decidable subset of \logic{} is adopted, where the domain of arithmetical variables is \emph{finite}.

\section{Continuous time semantics for Timed Automata} 
\label{sec:tasemantics}

The behavior of TA over time is described by means of execution traces that define the evolution of the APs, variables and clocks of the automata changing their values because transitions are taken or because time elapses.
When networks of synchronizing TA are considered, the formal definition of the semantics of (network of) TA has to deal with the following issues: 
\begin{enumerate}[label=\textit{\alph*)}]
\item how the automata progress over time by means of transitions associated with actions (liveness conditions); and
\item how the automata synchronize when transitions labeled with \channelsend,\channelreceive, \broadcastsend,  \broadcastreceive, \manytomanysend\ and \manytomanyreceive\ are fired.
\end{enumerate}

Only the general case semantics of a network of TA with variables is discussed hereafter. 
Obviously, the semantics of a network of TA without variables or of a single TA are just special cases.
Furthermore, in the rest of this paper, integer variables are restricted to finite domains.

\subsection{Preliminaries}
\label{sec:preliminaries}
Let \clocks\ be a set of clocks and $\clockconstraint \in \setclockconstraint$ be a clock constraint. 
A \emph{clock valuation} is a function \clockvaluation; the notation
$\clockvaluationfunction \models \clockconstraint$ indicates that the clock valuation \clockvaluationfunction\ satisfies \clockconstraint---i.e., by replacing $\clockvaluationfunction(\clock)$ for $\clock$ in any subformula of the form $\clock \sim \constant$ the clock constraint \clockconstraint\ evaluates to true. 
Let $\timevalue$ be an element of $\realNum$,
$\clockvaluationfunction+\timevalue$ denotes the clock valuation mapping clock $\clock$ to value $\clockvaluationfunction(\clock)+\timevalue$---i.e., $(\clockvaluationfunction+\timevalue)(\clock)=\clockvaluationfunction(\clock)+\timevalue$ for all $\clock \in \clocks$.
In the following, without loss of generality, any clock constraint \clockconstraint\ is supposed to be defined by means of a conjunction of atomic formulae of the form $\clock \sim \constant$ or  $\neg(\clock \sim \constant)$, where $\sim \in \{<,= \}$, 
and that clock constraints are convex, as non-convex ones can be reduced to the convex case.
A transition from $q$ to $q'$ labeled with a convex guard can be equivalently replaced with a set of transitions, all starting in $q$ and leading to $q'$, labeled with convex guards.

Before providing the formal definition of the transition relation for networks of TA, 
the notion of \emph{weak} satisfaction relation $\models_w$ over clock valuations and clock constraints is introduced. 
The weak satisfaction relation may be used to evaluate the invariants in the locations when a transition is fired, to allow different ways of performing an instantaneous  transition.
In particular, relation $\models_w$ is never satisfied when a clock constraint 
has the form $\clock = \constant$, where
$\clock$ is a clock and $\constant$ is a positive integer.
This is motivated by the following intuition: since time strictly progresses, a location labeled with a clock constraint of the form $x=\constant$, for some positive integer $\constant$, can never be reached in any execution of the TA. 
Hence, a constraint of the form $\clock = \constant$ cannot be weakly satisfied by any assignment.
Relation
$\models_w$ weaken the evaluation of clock constraints  
 of the form $\clock < \constant$ and $\clock > \constant$.
Those two constraints are weakly satisfied for $\clockvaluationfunction(\clock) = \constant-\epsilon$ or $\clockvaluationfunction(\clock) = \constant+\epsilon$ (where $\epsilon>0$) as in the non-weak case, but 
they are weakly satisfied also for $\clockvaluationfunction(\clock) = \constant$. i.e., as if the constraints were of the forms $\clock \leq \constant$ and $\clock \ge \constant$. 
Formally, 
a clock valuation $\clockvaluationfunction$ weakly satisfies 
a clock constraint $\clockconstraint$, written 
$\clockvaluationfunction \models_w \clockconstraint$, when the following conditions hold:

\begin{center}
\begin{tabular}{lll}
$\clockvaluationfunction \models_w \clock \sim \constant$ & iff $\ \ \ \clockvaluationfunction(\clock) \sim \constant \text{ or } \clockvaluationfunction(\clock) = \constant$ & $\sim\in\{<,>\}$ \\
$\clockvaluationfunction \models_w \neg(\clock \sim \constant)$ & iff $\ \ \ \clockvaluationfunction(\clock)\not\sim \constant$ & $\sim\in\{<,>\}$ \\
$\clockvaluationfunction \not \models_w \clock = \constant$ & for any $x \in \clocks, \constant \in \naturalNum$
\end{tabular}
\end{center}
Naturally, $\models_w$ can be extended to conjunctions of formulae $x \sim \constant$ or $\neg(x \sim \constant)$.
For instance, the formula $\clock<1 \land \neg(y<1) \land \neg(y=1)$ is both satisfied and weakly satisfied by the clock evaluation such that  $\clockvaluationfunction(\clock)=0.8$ and $\clockvaluationfunction(y)=1.2$, but it is only weakly satisfied if $\clockvaluationfunction(\clock)=1$ and $\clockvaluationfunction(y)=1$.
Notice that formula $\neg(y=\constant)$, for any non-negative integer $\constant$, always evaluates to true according to the weak satisfaction relation, because of the third condition.

A \emph{variable valuation} $\variablevaluationfunction$ is a function $\variablevaluationfunction: \variables \rightarrow \intNum$ that maps each variable in $\variables$ to an integer number; also, if $\variableconstraint \in \Gamma (\variables)$ is a variable constraint, $\variablevaluationfunction \models \variableconstraint$ indicates that valuation \variablevaluationfunction\ satisfies \variableconstraint. 
Let $\transid=\transtavar$ be a transition, $\clockvaluationfunction$ be a clock valuation  and $\variablevaluationfunction$ be a variable valuation,
\transid\ is \emph{enabled} in the valuation
when $\clockvaluationfunction$ satisfies \clockconstraint\ and $\variablevaluationfunction$ satisfies \variableconstraint.
In addition, a satisfaction relation for assignments is here introduced.
Let $\variablevaluationfunction$ and $\variablevaluationfunction^\prime$ be two variable valuations;
$(\variablevaluationfunction^\prime, \variablevaluationfunction) \models \mu$
indicates that all the assignments in $\mu$ are satisfied by means of $\variablevaluationfunction^\prime$ and $\variablevaluationfunction$.
Formally, all assignments of the form $\variable = \exp$ hold when $\variable$
is replaced with $\variablevaluationfunction^\prime(\variable)$ and every occurrence of $m \in \variables$ in $\exp$ is replaced with $\variablevaluationfunction(m)$.
Moreover, let $U(\mu)$ be the set of variables that are updated by $\mu$---that is, that appear as the left-hand side in an assignment of $\mu$---and let $U(t)$ indicate the set $U(\mu)$ given a transition $t$.

\begin{remark}
Relation $\models$ does not hold for inconsistent transitions, i.e., assigning multiple distinct values to a variable. For example, if $\varassignement=\{\variable=2, \variable=3\}$, there is no assignment to $\variable$ such that both $\variable=2$ and $\variable=3$ hold.
\end{remark}

\begin{definition}
Let \netid\  be a network of \tanumber\ TA.
A \emph{configuration} of \netid\ is a tuple $( \loc , \variablevaluationfunction, \clockvaluationfunction )$ where $\loc$ is a vector $[\autstate^1, \ldots, \autstate^\tanumber]$---
s.t.
$\autstate^1, \ldots, \autstate^\tanumber$ are locations of 
$\taid_1, \ldots, \taid_\tanumber$---\variablevaluationfunction\ (resp., \clockvaluationfunction) is a variable (resp., clock) valuation for the set
\variables\ (resp.,  \clocks )
 including all integer variables  (resp., clocks) appearing in the TA of the network.  
\end{definition}

When a network of TA is considered, 
it is possible that some 
automata in the network take a transition while the remaining others do not fire a transition and keep their state unchanged.
Firing a transition labeled with the null event \nullevent\ (i.e., a transition that does not synchronize) is however different from not taking a transition at all.
Symbol $\_$ indicates that an automaton \autindex\ does not perform any transition in \transitions{\autindex}.

The notation $\loc[\autindex]$ indicates the location of automaton $\taid_\autindex$---i.e., if $\loc[\autindex] = j$, then automaton $\taid_\autindex$ is in location $\autstate^\autindex_j$, assuming that the locations of each automaton are numbered, with $0$ indicating the initial one.

\newcommand{\ie}{\mathtt{ie}}
\newcommand{\ei}{\mathtt{ei}}

Two kinds of configuration changes may occur when an automaton in the network performs a transition from a state $q$ to $q^\prime$.
They are indicated in Def.~\ref{def:confatrans} with symbols $\ei$ {(excluded-included)} and $\ie$ {(included-excluded)}.
{Intuitively, these symbols constraints how the network behaves when a transition is fired.
The symbol $\ei$ forbids an automaton to be in state $q$ (excluded) in the instant in which the transition from $q$ to $q^\prime$ is fired,
while it forces the automaton to be in state $q^\prime$ (included).
Viceversa, the symbol $\ie$ forces an automaton to be in state $q$ (included) in the instant in which the transition from $q$ to $q^\prime$ is fired,
while it forbids the automaton to be in state $q^\prime$ (excluded).
}
Consider, for instance, a location $q$ labeled with $x<1$ and an outgoing transition.
If the automaton is in $q$, then the transition can be executed when $\clockvaluationfunction(x)<1$, in which case the corresponding configuration change can be arbitrarily marked either with $\ei$ or with $\ie$.
The transition can be executed even when $\clockvaluationfunction(x)=1$ holds, but in this case the kind of configuration change 
can only be 
$\ei$.

	\begin{definition}
		\label{def:confatrans}
		Let \netid\  be a network of \tanumber\ TA.
		Let $(\loc, \variablevaluationfunction, \clockvaluationfunction)$, $(\loc', \variablevaluationfunction', \clockvaluationfunction')$ be two configurations, let $\delta \in \realNum_{>0}$ and $\Lambda$ be a tuple of \tanumber\ symbols such that $\Lambda[k] \in \{\actions \times \{\ei, \ie\}\} \cup \{ \_  \}$
		for every $1 \le \autindex \leq \tanumber$. 
		Then, a configuration change is either a transition 
		$(\loc, \variablevaluationfunction, \clockvaluationfunction) \xrightarrow{\Lambda} (\loc', \variablevaluationfunction', \clockvaluationfunction')$ 
		or  a transition 
		$(\loc, \variablevaluationfunction, \clockvaluationfunction) \xrightarrow{\delta} (\loc', \variablevaluationfunction', \clockvaluationfunction')$ defined as follows.
		\begin{enumerate}
			\item 				\label{condone} $(\loc, \variablevaluationfunction, \clockvaluationfunction) \xrightarrow{\Lambda} (\loc', \variablevaluationfunction', \clockvaluationfunction')$  occurs if 
			\begin{enumerate}
				\item 
				\label{firing}
				for each $\Lambda[\autindex]=(\action, b)$
				there is a transition 
				$\conftranstavar{\autindex}$	 in $\mathcal{A}_{\autindex}$ such that:
				\begin{enumerate}
					\item
					\label{guards}
					$\clockvaluationfunction \models \clockconstraint$ and $\variablevaluationfunction \models \variableconstraint$, 

					\item
					\label{resets}
					$\clockvaluationfunction'(\clock)=0$ holds for all $\clock \in \resettedclocks$, 

					\item
					\label{assignments}
					$(\variablevaluationfunction',\variablevaluationfunction) \models \varassignement$, 
					\item
					\label{invariant-leftopenrightclosed}
					when $b = \ei$ then:
					\begin{itemize}
						\item $\clockvaluationfunction \models_w Inv(\loc[\autindex])$ and
						\item $\clockvaluationfunction' \models Inv(\loc'[\autindex])$
					\end{itemize}
					\item\label{invariant-leftclosedrightopen}
					when $b = \ie$:
					\begin{itemize}
						\item $\clockvaluationfunction \models Inv(\loc[\autindex])$ and 
						\item $\clockvaluationfunction' \models_w Inv(\loc'[\autindex])$ 
					\end{itemize}
				\end{enumerate}
				\item
				\label{keepstate} 
				for each $\Lambda[\autindex] = \_$ it holds that:
				\label{epsilon-transition}
				\begin{enumerate}
					\item
					\label{epsilon-keepstate} 
					$\loc'[\autindex]=\loc[\autindex]$;
					\item $\clockvaluationfunction \models Inv(\loc[\autindex])$ and $\clockvaluationfunction' \models Inv(\loc'[\autindex])$.
				\end{enumerate}
				\item
				\label{keepclocksandvariables}
				for each clock $\clock \in \clocks$ (resp., integer variable $\variable \in \variables$), if \clock\ (resp., \variable) does not appear in any \resettedclocks\
				(resp., it is not assigned by any $A$) of one of the transitions taken by $\mathcal{A}_{1}, \ldots, \mathcal{A}_{K}$, then $\clockvaluationfunction'(x) = \clockvaluationfunction(x)$ (resp., $\variablevaluationfunction'(n)=\variablevaluationfunction(n)$);
			\end{enumerate}
			\item \label{time-transition} $(\loc, \variablevaluationfunction, \clockvaluationfunction) \xrightarrow{\delta} (\loc', \variablevaluationfunction', \clockvaluationfunction')$ 
occurs if 
			$\loc'[k]=\loc[k]$,
			$\variablevaluationfunction'= \variablevaluationfunction$,
			$\clockvaluationfunction'= \clockvaluationfunction+\delta$
			and for all $1 \leq \autindex \leq K$,
			$\clockvaluationfunction' \models_w Inv(\loc[\autindex])$.
		\end{enumerate}
	\end{definition}

A configuration change $(\loc, \variablevaluationfunction, \clockvaluationfunction ) \xrightarrow{\Lambda} (\loc', \variablevaluationfunction', \clockvaluationfunction')$, for some $\Lambda \in  \{\{\actions \times \{\ei, \ie\}\} \cup \{ \_  \}\}^\numberOfTA$,
satisfying (1) is called a \emph{discrete transition}.
If it satisfies (2) then it is called a \emph{time transition}.
For convenience of notation, symbols $(\action,\ei)$ and $(\action,\ie)$, for some $\action \in \actions$, are hereinafter denoted respectively with $\action^{)[}$ and $\action^{](}$, meaning that the discrete transition performed by the $k$-th automaton, such that $\Lambda[\autindex]=(\action,\ei)$ (resp., $\Lambda[\autindex]=(\action,\ie)$), is \emph{open-closed} (resp., \emph{closed-open}).
	The use of symbols $\action^{)[}$ or  $\action^{](}$ allows the distinction of two different ways of performing a transition by means of an action $\action$. 
	The two modes are determined by the conditions in~\ref{invariant-leftopenrightclosed} and~\ref{invariant-leftclosedrightopen} and depend on the invariants of the locations involved in the transition, the clock values and the resets applied in the configuration change. 
	Location invariants and resets make it possible to constrain every symbol $\Lambda[\autindex]$,
	associated with $\mathcal{A}_k$, and hence to define how the configuration change in $\mathcal{A}_k$ is realized. 
	Cases~\ref{condone} and~\ref{time-transition}  are discussed in detail in the following.

Case~\ref{condone}. If the discrete transition is open-closed---i.e., the symbol is $\action^{)[}$---then \ref{invariant-leftopenrightclosed} holds.
The conditions of this case impose that $\clockvaluationfunction'$ satisfies the invariant of the destination location and $\clockvaluationfunction$ \emph{weakly} satisfies the invariant of the source location. 
Therefore, if the invariant of the source state is $\clock <1$, 
then the transition can be taken with $\clockvaluationfunction(x)\leq 1$. 
This is achieved through the weak  satisfaction relation that guarantees the (weak) satisfaction of the invariant $\clock <1$ with $\clockvaluationfunction(x)=1$.
Conversely, if the discrete transition is closed-open---i.e., the symbol is $\action^{](}$---then \ref{invariant-leftclosedrightopen} holds.
The conditions defined therein allow the invariant of the destination location to be weakly satisfied and the transition to be fired with $\clockvaluationfunction(x)\geq 1$, if the invariant of the destination state is $\clock>1$. 

Based on the invariants and resets, the symbol $\Lambda[\autindex]$ is either non-deterministically chosen between $\action^{)[}$ and $\action^{](}$ because both symbols are allowed, or it is determistically defined because only one is permitted, according to conditions \ref{invariant-leftopenrightclosed} and ~\ref{invariant-leftclosedrightopen}.
Figure~\ref{fig:leftclosed-open-transitions} shows two automata 
and the possible transitions.
The case of a transition on a symbol $\action^{)[}$ is exemplified in Fig.~\ref{fig:leftopenTA}.
When $v(x)\leq 1$ holds, the invariant $Inv(\loc[\autindex])$ of the first location is weakly satisfied by $v$---i.e., $\clockvaluationfunction \models_w Inv(\loc[\autindex])$ holds.
In such a case, since $Inv(\loc'[\autindex])$ is empty---hence it is trivially true---symbol $\action^{)[}$ is allowed.
In the automaton of Fig.~\ref{fig:leftclosedTA}, instead, only symbol $\action^{](}$ is allowed, because the constraint on the second location requires that $x$ is not reset.
Constraint $\clockvaluationfunction' \models_w Inv(\loc'[\autindex])$ of condition \ref{invariant-leftclosedrightopen} is satisfiable with $\clockvaluationfunction'(x)=0$. 
Conversely, $\clockvaluationfunction \models Inv(\loc[\autindex])$ of condition \ref{invariant-leftopenrightclosed} would be falsified because of the reset and, hence, $\action^{)[}$ is prevented.
In addition, $\action^{](}$ is also allowed in the automaton of Fig.~\ref{fig:leftopenTA} when $v(x) < 1$ holds, because $\clockvaluationfunction \models Inv(\loc[\autindex])$ holds in that case.

Case~\ref{time-transition}. It 
defines the time transitions by means of the weak relation $\models_w$.
Consider an open-closed transition that changes the location of an automaton currently in $q$, for some $q$, and the last time transition immediately preceding it which makes the time progress of $\delta$ time units, for some $\delta>0$.
In order to perform the open-closed transition, 
$v+\delta$
weakly satisfies $Inv(q)$, as required by \ref{invariant-leftopenrightclosed}.
The use of relation $\models$ instead of $\models_w$ in \ref{time-transition} would prevent the occurrence of some open-closed configuration changes, i.e., 
those that would be caused by an  $Inv(q)$ being weakly, but not strongly, satisfied in the time transition.
In case of closed-open transitions, $v+\delta$ (strongly) satisfies  $Inv(q)$, as required by \ref{invariant-leftclosedrightopen}.
Hence, $v+\delta$ also weakly satisfies $Inv(q)$, as in condition~\ref{time-transition}.
For instance, in Fig.~\ref{fig:leftopenTA}, if the automaton is in the location labeled with $x<1$ and $\clockvaluationfunction(x)=0.8$ then the time progress $\delta=0.2$ is permitted by \ref{time-transition} in order to perform the outgoing transition in an open-closed manner with $\clockvaluationfunction'(x)=1$. 
Moreover, if the time progress $\delta$ is such that the invariant of the current location $q$ is such that $v+\delta\models Inv(q)$ holds, then both kind of transitions are allowed.

The combination of conditions~\ref{condone} and~\ref{time-transition} describe how the configuration of a network of TA changes.

Based on the previous arguments, the symbols $\action^{)[}$ and $\action^{](}$ 
will be 
considered in Sec.~\ref{sec:checkingMITLI} 
to define the signals associated with atomic propositions  and variables when discrete transitions are taken.

The notion of trace is now introduced. 
Recall that, $\variablevaluationfunction^0: Int \rightarrow \intNum$ assigns each variable with a value in $\intNum$ (see Def~\ref{def:tawithvariables}).

\begin{definition}
\label{def:trace}
Let \netid\  be a network of \tanumber\ TA. 
A \emph{trace} is an infinite sequence \[\eta=(\loc_0, \variablevaluationfunctionpar{0}, \clockvaluationfunctionpar{0}), {e_0}, (\loc_1, \variablevaluationfunctionpar{1}, \clockvaluationfunctionpar{1}), {e_1}, (\loc_2, \variablevaluationfunctionpar{2}, \clockvaluationfunctionpar{2}), {e_2}, \ldots\] such that:
\begin{enumerate}
\item for all $i \in \naturalNum$, $e_i = \Lambda_i$ or $e_i = \delta_i$;
\item
for all $h \in \naturalNum$ it holds that $(\loc_h, \variablevaluationfunctionpar{h}, \clockvaluationfunctionpar{h})  \xrightarrow{e_h} (\loc_{h+1},\variablevaluationfunctionpar{h+1}, \clockvaluationfunctionpar{h+1})$;

\item \label{firstTransitionIsDelay}
$e_0 = \delta_0$, for some $\delta_0 \in \realNum_{>0}$;

\item
\label{invariantholdsinstate}
for all $1 \leq \autindex \leq K$, it holds that $\loc_0[\autindex]=0$, $\clockvaluationfunction_0 \models Inv(\loc_0[\autindex])$, for all $\clock \in \clocks$ it holds that $\clockvaluationfunctionpar{0}(\clock) = 0$, and for all $\variable \in \variables$ it holds that $\variablevaluationfunctionpar{0}(n) = \variablevaluationfunction^0(\variable)$.

\item
\label{noconsecLambda}
discrete transitions must be followed by time transitions; that is, if $e_h$ is a discrete transition ($e_h = \Lambda_{h}$), then $e_{h+1}$ is a time transition ($e_{h+1} = \delta_{h+1}$).
 
\end{enumerate}
The word $w(\eta)$ of a trace $\eta$
is the sequence $e_0 e_1 \dots$.
\end{definition}

With a slight abuse of notation, a trace $(\loc_0, \variablevaluationfunctionpar{0}, $ $ \clockvaluationfunctionpar{0}), {e_0}, (\loc_1, $ $\variablevaluationfunctionpar{1}, \clockvaluationfunctionpar{1}), {e_1},$
 $ (\loc_2, $ $ \variablevaluationfunctionpar{2}, \clockvaluationfunctionpar{2}), {e_2}, \ldots$ can be written as \trace.

Since by condition \ref{noconsecLambda} there cannot be two consecutive discrete transitions, and since any finite sequence of consecutive delays $\delta_h \dots \delta_{h+k}$, with $k\geq 0$, is equivalent to a single delay $\sum_{i=h}^{h+k} \delta_i$, a trace can always be rewritten into
a new one such that discrete and time transitions strictly alternate.
Moreover, by the previous property, every {time transition $\delta_h$ can always be replaced with}  a finite
sequence of $m$ pairs of time and discrete transition $\delta_{h,0}\Lambda_{h,0} \delta_{h,1}\Lambda_{h,2} \dots \delta_{h,m}$, strictly alternating, such that in $\Lambda_{h,i}[\autindex]=\_$, for all $0\leq i\leq m-1$ and $\delta_h = \sum_{i=0}^{m} \delta_{h,i}$.
This property is used in Sec.~\ref{sec:checkingMITLI} to allow the use of~\cite{BRS15b} in the resolution of the model-checking problem of TA with MITL.

To facilitate future discussions, a trace is represented with the following notation where the numbering of configurations increases only after the discrete transitions:

\begin{align}
(\loc_0, \variablevaluationfunctionpar{0}, \clockvaluationfunctionpar{0}) 
\xrightarrow{\delta_0} 
(\loc'_0, \variablevaluationfunctionpar{0}', \clockvaluationfunctionpar{0}') 
\xrightarrow{\Lambda_0}
(\loc_1, \variablevaluationfunctionpar{1}, \clockvaluationfunctionpar{1}) 
\xrightarrow{\delta_1}  \ldots  \nonumber
\end{align}

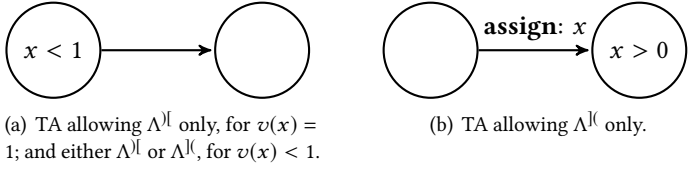
\begin{figure}[t]
	\caption{Two examples of transitions enforcing a different and unique configuration change.}
	\label{fig:leftclosed-open-transitions}
	\subfigure[TA allowing $\Lambda^{)[}$ only, for $v(x)=1$; and either $\Lambda^{)[}$ or $\Lambda^{](}$, for $v(x)<1$.]{\begin{tikzpicture}[->,>=stealth',shorten >=1pt,auto,node distance=1.5cm,
thick,main node/.style={circle,draw,font=\sffamily\Large\bfseries}]

\node[state,align=center,  minimum size=1.25cm]         (l0) [] {$x < 1$};
\node[state,align=center, minimum size=1.25cm]          (l1) [right= of l0]  {};

\path[->] (l0) edge[] node[text width=2cm,align=center] {} (l1);

\end{tikzpicture} \label{fig:leftopenTA}}
	\ \ \ \ \ \ \ \ \  
	\subfigure[TA allowing $\Lambda^{](}$ only.]{\begin{tikzpicture}[->,>=stealth',shorten >=1pt,auto,node distance=1.5cm,
                    thick,main node/.style={circle,draw,font=\sffamily\Large\bfseries}]

\node[state,align=center,  minimum size=1.25cm]         (l0) [] {};
\node[state,align=center, minimum size=1.25cm]          (l1) [right= of l0]  {$x > 0$};

\path[->] (l0) edge[] node[text width=2cm,align=center] {
\textbf{assign}: $x$ 
} (l1);

\end{tikzpicture} \label{fig:leftclosedTA}}
\end{figure}

\subsection{Liveness and synchronization}\label{sec:livenessandsynch}
Definition~\ref{def:trace} only provides weak conditions on the occurrences of discrete transitions and does not express any restriction on how TA synchronize.
In fact, beside the first three conditions  requiring that traces are sequences of configuration changes starting from a specific initial configuration, 
only Condition~\ref{noconsecLambda} expresses a restriction on how the configuration changes occur, which only prevents discrete transitions from occurring consecutively, one after the other. 
However, one is typically interested in ``live'' traces, in which some transition is eventually taken and where the effect of the synchronizing primitives is precisely defined.

\subsubsection{Liveness.}
 
Table~\ref{tab:livenesssemantic} shows the formal definition of the following four possible liveness conditions, for a generic trace of a network with $K$ timed automata.
\begin{itemize}
\item \emph{Strong (Weak) transition liveness:} at any time instant,  \emph{each} (resp., \emph{at least one}) automaton of the network eventually performs a transition.
\item \emph{Strong (Weak) guard liveness:} at any time instant, \emph{for each automaton} (resp., \emph{there exists an automaton} such that) the values of clocks and variables will eventually enable one of its transitions
\footnote{The constraint does not force the transition to be taken. Moreover, alternative definitions can be given by considering only  clock or variable guards.}.
\end{itemize}
\begin{table*}[t]
	\caption{Formal definition of different liveness properties for traces.}
	\label{tab:livenesssemantic}
	\begin{tabular}{ p{2.5cm}  p{10.5cm} }
		\toprule
		\textbf{Name} & \textbf{Formulation of the Semantics} \\
		\toprule
		\parbox[t]{3.5cm}{Strong  transition\\  liveness} & 
		For every $h \geq 0$ and $1\leq \autindex \leq K$, there exists $j > h$ such that 
		$(\loc'_j, \variablevaluationfunctionpar{j}', \clockvaluationfunctionpar{j}')  \xrightarrow{\Lambda_j} 
		(\loc_{j+1}, \variablevaluationfunctionpar{j+1}, \clockvaluationfunctionpar{j+1} )$ belongs to the trace and $\Lambda_j[\autindex] \neq \_$.\\
		\midrule
		\parbox[t]{3.5cm}{Weak transition\\  liveness} &
		For every $h \geq 0$ there exist $1\leq \autindex \leq K$ and $j > h$ such that
		$(\loc'_j, \variablevaluationfunctionpar{j}', \clockvaluationfunctionpar{j}')  \xrightarrow{\Lambda_j} 
		(\loc_{j+1}, \variablevaluationfunctionpar{j+1}, \clockvaluationfunctionpar{j+1} )$ belongs to the trace and $\Lambda_j[\autindex] \neq \_$.\\
		\midrule 
		\parbox[t]{3.5cm}{Strong guard\\  liveness} & 
		For every $h \geq 0$ and $1\leq\autindex \leq K$, there exist $j > h$ and a configuration $(\loc'_j, \variablevaluationfunctionpar{j}', \clockvaluationfunctionpar{j}')$ in the trace such that there is a transition \transtavar\, with $q = \loc'_j[k]$, for which 
		$\variablevaluationfunctionpar{j} \models \variableconstraint$ and  $\clockvaluationfunctionpar{j} \models \clockconstraint$ hold.  \\
		\midrule 
		\parbox[t]{3.5cm}{Weak guard\\  liveness} & For every $h \geq 0$ there exist $1\leq\autindex \leq K$, $j > h$, and a configuration $(\loc'_j, \variablevaluationfunctionpar{j}', \clockvaluationfunctionpar{j}')$ in the trace such that there is a transition \transtavar\, with $q = \loc'_j[k]$,
		for which
		$\variablevaluationfunctionpar{j} \models \variableconstraint$ and  $\clockvaluationfunctionpar{j} \models \clockconstraint$ hold. \\
		\bottomrule
	\end{tabular}
\end{table*}

Even if the previous conditions restrict the occurrence of transitions or the satisfiability of guards along the trace, they do not prevent, in general, the progress of time from slowing down.
This issue is well-known in the literature of timed verification and, intuitively, it is caused by so-called \emph{time convergent} traces, where the sum of all the delays $\delta_i$ associated with time transitions is bounded by some positive integer.
Therefore, the previous liveness conditions allow Zeno traces, i.e., where infinitely many actions can occur in finite time. 
Avoiding Zeno traces can be done in several ways.
For instance, one can require strong transition liveness and introduce a new TA in the network which infinitely many times along the trace resets a clock when 
the clock value  
reaches $1$.

\subsubsection{Synchronization}
Section~\ref{sec:Background} introduced qualifiers 
 \channelsend, \channelreceive, \broadcastsend, \broadcastreceive, \manytomanysend, and \manytomanyreceive\ labeling actions on the transitions with the goal of capturing different ways in which the automata of a network can synchronize.
Qualifiers \channelsend\ and \channelreceive\ describe a so-called \emph{channel-based} synchronization; 
qualifiers \broadcastsend\ and \broadcastreceive\ describe a \emph{broadcast} synchronization; and qualifiers 
\manytomanysend\ and \manytomanyreceive\ describe a \emph{one-to-many} synchronization.

Channel-based, broadcast and one-to-many synchronizations can be arbitrarily mixed in the same configuration change.
Table~\ref{tab:networksemantics} shows the formal definition of the channel-based,  broadcast and one-to-many synchronization mechanisms for a generic trace of a network of TA. 
\begin{itemize}
\item \emph{Channel-based synchronization:} 
for any discrete transition, every ``sending'' (qualifier \channelsend) action in a TA $\mathcal{A}_{\autindex}$
must be matched by exactly one corresponding ``receiving'' (qualifier \channelreceive) action in another TA $\mathcal{A}_{\autindex^\prime}$ on the same channel (e.g., $\action \channelsend$ and $\action \channelreceive$) labeling an enabled transition. 
 \item \emph{Broadcast synchronization:} for any discrete transition, every ``sending'' (qualifier  \broadcastsend) action in a TA $\mathcal{A}_{\autindex}$ is matched by all and only ``receiving'' (qualifier \broadcastreceive) actions in TA $\mathcal{A}_{\autindex^\prime}$ labeling an enabled transition. In other words, either $\mathcal{A}_{\autindex^\prime} \neq \mathcal{A}_{\autindex}$ takes a transition labeled with $\action \broadcastreceive$, or it does not exist any enabled transition for $\mathcal{A}_{\autindex^\prime}$ labeled with $\action \broadcastreceive$.
\item \emph{One-to-many synchronization:} for any discrete transition, every ``sending'' (qualifier \manytomanysend) action in $\mathcal{A}_{\autindex}$ is matched by a (non empty) set of ``receiving'' (qualifier \manytomanyreceive) actions in $\mathcal{A}_{\autindex^\prime}$ labeling an enabled transition. 
The one-to-many synchronization  is a variation of the broadcast in which when an automaton $\mathcal{A}_{\autindex}$  sends a message it is received by at least one receiver. However,  not all the automata that have a transition labeled with  $\action \manytomanyreceive$ are forced to receive the message.
\end{itemize}
Notice that, for any channel $\action$, the previous synchronizations allow only one TA to send a message on $\action$ at any time instant; on the other hand, distinct TA can send messages concurrently on separate channels.

\begin{table*}[t]
	\caption{Definition of different constraints on traces depending on synchronization primitives.}
	\label{tab:networksemantics}
	\begin{tabular}{ p{1.5cm}  p{11.5cm} }
		\toprule
		\textbf{Type} & \textbf{Formulation of the Semantics} \\
		\toprule
		\parbox[t]{2.5cm}{Channels} &  For every configuration change $(\loc'_h, \variablevaluationfunctionpar{h}',$ $\clockvaluationfunctionpar{h}')$ 
		$ \xrightarrow{\Lambda_h}  
		(\loc_{h+1}, \variablevaluationfunctionpar{h+1}, \clockvaluationfunctionpar{h+1} )$ and $1\leq \autindex \leq \numberOfTA$
		such that $\action \channelsend=\Lambda_h[\autindex]$ holds, there exists exactly one $0 < \autindex' \leq \numberOfTA$ such that $\autindex' \neq \autindex$ and $\action \channelreceive=\Lambda_h[\autindex^\prime]$ hold, and vice-versa. 
		\\
		\midrule
		\parbox[t]{2.5cm}{Broadcast} &  
		For every configuration change $(\loc'_h, \variablevaluationfunctionpar{h}',$ $\clockvaluationfunctionpar{h}')$ $ \xrightarrow{\Lambda_h} 
		(\loc_{h+1}, \variablevaluationfunctionpar{h+1}, \clockvaluationfunctionpar{h+1} )$ and $1\leq \autindex \leq K$ such that $\action  \broadcastsend=\Lambda_h[\autindex]$ holds, for every $1 \leq \autindex' \leq K$, with $\autindex' \neq \autindex$, it holds that $\action\broadcastsend\not=\Lambda_h[\autindex']$ and either $\action \broadcastreceive=\Lambda_h[\autindex^\prime]$ holds, or no transition 
		$q \xrightarrow{\clockconstraint, \variableconstraint,  \action \broadcastreceive, \resettedclocks, \varassignement} q'$ in $\mathcal{A}_{\autindex^\prime}$ 
		is such that $\variablevaluationfunctionpar{h}' \models \variableconstraint$ and $\clockvaluationfunctionpar{h}' \models \clockconstraint$ hold.
		\\
		\midrule
		\parbox[t]{2.5cm}{One-to-\\ many} &  For every configuration change $(\loc'_h, \variablevaluationfunctionpar{h}',$ $\clockvaluationfunctionpar{h}')$ $ \xrightarrow{\Lambda_h} 
		(\loc_{h+1}, \variablevaluationfunctionpar{h+1}, \clockvaluationfunctionpar{h+1} )$ and $1\leq \autindex \leq \numberOfTA$ such that $\action \manytomanysend=\Lambda_h[\autindex]$, there exists at least one $1\leq \autindex' \leq K$ such that $\autindex' \neq \autindex$ and $\action \manytomanyreceive=\Lambda_h[\autindex^\prime]$ hold.
		\\
		\bottomrule 
	\end{tabular}
\end{table*}

\section{From timed automata to \blindonoff{\logic\ }{CLTLoc$_v$}}
\label{sec:TA2CLTLoc}

This section shows that, given a network $\netid$ of TA, 
it is possible to construct a \logic\ formula 
$\Phi_{\mathcal{N}}= \varphi_{\mathcal{N}} 
\wedge \varphi_{\mathit{l}} \wedge \varphi_{\mathit{s}} \wedge \varphi_{\mathit{ef}}$
whose models represent 
the traces of $\mathcal{N}$. 
Formulae $\varphi_{\mathcal{N}}$, 
$\varphi_{\mathit{l}}$,  
$\varphi_{\mathit{s}}$ and $\varphi_{\mathit{ef}}$
encode, respectively: the behavior of the network, that is  the effect of the transitions on the configuration of the network including how clocks are reset, how variables and locations are modified and when transitions can be taken ($ \varphi_{\mathcal{N}} $); a set of constraints on the 
liveness conditions ($ \varphi_{\mathit{l}}$); the semantics of the firing of the transitions that depend on 
the synchronization modifiers that decorate their labeling events ($\varphi_{\mathit{s}}$);
and constraints on the types of edges (open-closed or closed-open) with which transitions are taken ($\varphi_{\mathit{ef}}$).
Before discussing \logic\ formulae $\varphi_{\mathcal{N}}$, 
$\varphi_{\mathit{l}}$  
$\varphi_{\mathit{s}}$ and $\varphi_{\mathit{ef}}$, formula $\phi_{\mathit{clock}}$ is first introduced to encode 
 a set of constraints on the \logic\ clocks used in $\varphi_{\mathcal{N}}$.

\newcommand{\x}{\mathtt{x}}

\subsection{Encoding constraints over clocks ($\boldsymbol{\phi_{\mathit{clock}}}$)}\label{sec:clock-encoding}
Unlike those in TA, clocks in \logic\ formulae cannot be tested and reset at the same time.
For instance, while it is possible that a transition in a TA both has guard $\clock>5$ and resets clock $\clock$, in \logic\ simultaneous test and reset 
would yield a contradiction, as testing $\clock>5$ and resetting \clock\ in the same position equals to formula $\clock>5 \land \clock=0$.
Therefore, for each clock $\clock \in \clocks$, two clocks $\clock_0$ and $\clock_1$ are introduced in $\Phi_\netid$ to represent a single clock \clock\ of the automaton.
An additional Boolean variable
$\x_v$ keeps track, in any discrete time position, of which clock between $\clock_0$ and $\clock_1$ is ``active'' ($\x_v$ being equal to 0 or 1 respectively).
Clocks $\clock_0$ and $\clock_1$ are never reset at the same time and their resets alternate. 
If $\x_v=0$ (resp., $\x_v=1$) at position \timeposition\ of the  model of $\Phi_\mathcal{N}$, then $\clock_0$ (resp., $\clock_1$) is the active clock at \timeposition\ and $\sigma(\timeposition,\clock_0)$ (resp., $\sigma(\timeposition,\clock_1)$) is the value used to evaluate the clock constraints at \timeposition.
If the reset of \clock\ has to be represented at \timeposition, clock $\clock_1$ (resp., $\clock_0$) is set to $0$ and the value $\x_v$ in position $\timeposition+1$ is set to 1 (resp., 0)---i.e., the active clock is switched.

Table~\ref{tab:clocks} shows the formulae $\phi_1$ and $\phi_2$ that are used to define $\phi_{\mathit{clock}}$.
Formula $\phi_1$ specifies that initially, the active clock is $\clock_0$. In position 0 variable $\x_v$ is equal to $0$ (indicating that $\clock_0$ is the active clock),
$\clock_0$ is also equal to $0$ and $\clock_1$ has an arbitrary value greater than zero.
Formula $\phi_2$ specifies that if $\clock_j$ is reset, it cannot be reset again before $\clock_{(j+1)\mod 2}$ is reset.
For instance, if clock $\clock_0$ is reset, then it cannot be reset again (it remains grater than zero) and it is the active clock ($\x_v=0$) as long as $\clock_1$ is different from $0$. 

\begin{table}
	\centering
	\caption{Encoding of the clocks of the automata.}
	\label{tab:clocks}
	\begin{tabular}{ l }
		\toprule
		$\phi_1 \coloneqq \bigwedge_{\clock \in \clocks}( \clock_0=0 \ \LTLand \ \clock_1>0 \ \LTLand \ \x_v=0)$ \\
		\midrule
		$\phi_2(j) \coloneqq \bigwedge_{\clock \in \clocks}
		(\clock_j=0) \LTLimplication 
		\LTLx 
		(
		(\clock_{(j+1)\mod 2}=0) 
		\LTLrelease 
		(
		(\x_v=j) 
		\wedge 
		(\clock_j> 0) ))$\\
		\bottomrule
	\end{tabular}
\end{table}

Formula $\phi_{\mathit{clock}}$ is defined as $\phi_1 \land \LTLg (\phi_2(0) \land \phi_2(1))$.
Since every clock \clock\ is represented by two variables $\clock_0$ and $\clock_1$, all clock constraints of the form $\clock \sim \constant$ in $\Gamma (\clocks)$ that appear in the automaton are translated by means of the following \logic~formula:
$$\phi_{\clock \sim \constant}:=((\clock_0 \sim \constant) \LTLand (\x_v=0)) \LTLor ((\clock_1 \sim \constant) \LTLand (\x_v=1)).$$

\begin{example}
Figure~\ref{fig:alternatingclocks} depicts a sequence of tests and resets of clock \clock\ over 
9
discrete positions.
The first row shows the sequence of operations $[\clock:=0]$, $[\clock>5,\clock:=0]$, $[\clock<1]$ and $[\clock=3,\clock:=0]$, where $[\clock>5,\clock:=0]$, for instance, means that $\clock$ is tested against 5 and it is reset simultaneously.
In the second row, all the operations on $x$ are represented by means of clocks $\clock_0$ and $\clock_1$.
Based on the value of $\x_v$ at $i$, the active clock at that position is used to realize the correspondent operation.
A 
continuous 
line 
identifies the regions where  either clock $x_0$ or $x_1$ is active.
The third row shows the constraints on $\sigma$ that are enforced by the operations on $x$.
\end{example}

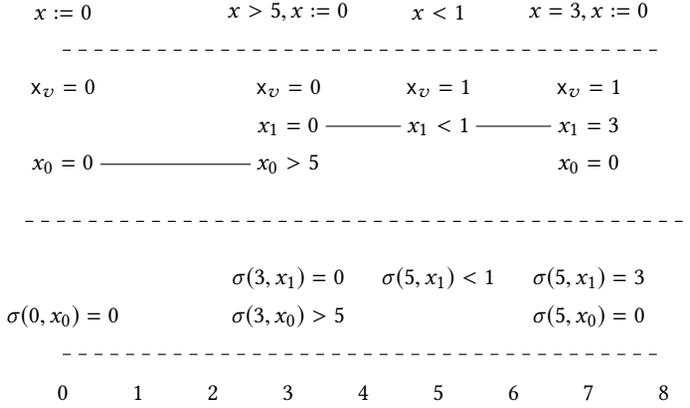
\begin{figure}
	\begin{tikzpicture}
	\draw[dashed] (0,-1.5) -- (8,-1.5);
	\draw[dashed] (0,2.5) -- (8,2.5);
	\draw (0.5,1) -- (2.5,1);
	\draw (3.5,1.5) -- (4.5,1.5);
	\draw (5.5,1.5) -- (6.5,1.5);
	\draw[dashed] (-0.5,0.25) -- (8.5,0.25);
	\draw node at (0,-2) {\small $0$};
	\draw node at (1,-2) {\small $1$};
	\draw node at (2,-2) {\small $2$};
	\draw node at (3,-2) {\small $3$};
	\draw node at (4,-2) {\small $4$};
	\draw node at (5,-2) {\small $5$};
	\draw node at (6,-2) {\small $6$};
	\draw node at (7,-2) {\small $7$};
	\draw node at (8,-2) {\small $8$};
	\draw node at (0,3) {\small $\clock:=0$};
	\draw node at (3,3) {\small $\clock>5, \clock:=0$};
	\draw node at (5,3) {\small $\clock<1$};
	\draw node at (7,3) {\small $\clock=3, \clock:=0$};	
	\draw node at (0,1) {\small $\clock_0=0$};
	\draw node at (3,1) {\small $\clock_0>5$};		
	\draw node at (3,1.5) {\small $\clock_1=0$};
	\draw node at (5,1.5) {\small $\clock_1<1$};
	\draw node at (7,1.5) {\small $\clock_1=3$};	
	\draw node at (7,1) {\small $\clock_0=0$};
	\draw node at (0,2) {\small $\x_v=0$};
	\draw node at (3,2) {\small $\x_v=0$};
	\draw node at (5,2) {\small $\x_v=1$};
	\draw node at (7,2) {\small $\x_v=1$};
	\draw node at (0,-1) {\small $\sigma(0,\clock_0)=0$};
	\draw node at (3,-1) {\small $\sigma(3,\clock_0)>5$};		
	\draw node at (3,-0.5) {\small $\sigma(3,\clock_1)=0$};
	\draw node at (5,-0.5) {\small $\sigma(5,\clock_1)<1$};
	\draw node at (7,-0.5) {\small $\sigma(5,\clock_1)=3$};	
	\draw node at (7,-1) {\small $\sigma(5,\clock_0)=0$};	
	\end{tikzpicture}
	\caption{Representation of tests and resets of clock \clock\ by means of the two copies $\clock_0$ and $\clock_1$. }
	\label{fig:alternatingclocks}
\end{figure}

\subsection{Encoding the network ($\boldsymbol{\varphi_{\mathcal{N}}}$).}\label{sec:phi_N}

Formula $\varphi_{\mathcal{N}}$  encodes both the relation $\xrightarrow{e}$ between pairs of configurations and all (and only) the conditions of  Def.~\ref{def:trace} defining a trace.
However, $\varphi_{\mathcal{N}}$ does not express any restriction on automata synchronization and it does not impose any specific liveness condition. 
The discrete positions of the \logic\ model render the configurations of the network evolving over the continuous time by means of a discrete sequence of points.
All those positions in the model represent the discrete transitions performed by the automata of $\mathcal{N}$ that modify the values of variables, clocks and locations and also the time transitions that produce the elapsing of time.
The model of the formula $\varphi_{\mathcal{N}}$ is thus a (representation of a) possible trace realized by the network $\mathcal{N}$.

A generic configuration $(\loc,\variablevaluationfunction, \clockvaluationfunctionpar{})$ of $\mathcal{N}$ is represented in the \CLTLoc\ formula by means of the values of clocks and variables and a set of auxiliary variables representing locations and transitions of the automata.
An array $\loc$ of $\numberOfTA$ integer variables in the \CLTLoc\ formula encodes the location of each automaton in the network, with $\loc[\autindex]\in \{0, \dots, |Q_\autindex|-1\}$ for every $1\leq \autindex\leq \numberOfTA$.
Given an enumeration of the elements in $Q_\autindex$, $\loc[\autindex] = i$ indicates that $\mathcal{A}_\autindex$ is in the location $q_i$ of $Q_\autindex$.
An array $\tr$ of $\numberOfTA$ integer variables encodes the transitions of each automaton in the network, with $\tr[\autindex] \in \{0, \dots, |T_\autindex|-1\} \cup \{\notr\}$, for every $1\leq \autindex\leq \numberOfTA$. 
Given an enumeration of the elements in $T_\autindex$, $\tr[\autindex] = i$ indicates the execution of transition $t_i$ of $T_\autindex$, while $\tr[\autindex] = \notr$ indicates that no transition of $T_k$ is performed (it represents the symbol $\_$ in the discrete transitions of traces).
An array $\edge$ of $\numberOfTA$ Booleans represents the kind of transition taken by each automaton $\mathcal{A}_k$.
If $\edge[\autindex]$ is true, then the configuration change of $\mathcal{A}_k$ at the current position is closed-open; otherwise, it is open-closed.
Finally, for each variable $\variable \in \variables$, a corresponding \logic\ integer variable is introduced.

The configuration change determined by a discrete transition is explicitly encoded with a formula that expresses the effect of the transition on the network configuration. 
Conversely, since in \logic\ time progress is inherent in the model, the encoding does not explicitly deal with time transitions of traces because between any pair of adjacent positions $\timeposition$ and $\timeposition+1$ the time always advances.
To facilitate understanding and future discussions, a trace is written 
as:

\begin{align}
(\loc_0, \variablevaluationfunctionpar{0}, \clockvaluationfunctionpar{0}) 
\xrightarrow{\delta_0 \Lambda_0} 
(\loc_1, \variablevaluationfunctionpar{1}, \clockvaluationfunctionpar{1}) 
\xrightarrow{\delta_1 \Lambda_1}  
(\loc_2, \variablevaluationfunctionpar{2}, \clockvaluationfunctionpar{2}) \ldots \nonumber
\end{align}

where time and discrete transitions are paired together and the notation $(\loc_i, \variablevaluationfunctionpar{i}, \clockvaluationfunctionpar{i}) 
\xrightarrow{\delta_i \Lambda_i} 
(\loc_{i+1}, \variablevaluationfunctionpar{i+1}, \clockvaluationfunctionpar{i+1})$ is simply a rewriting of

\begin{align}
(\loc_i, \variablevaluationfunctionpar{i}, \clockvaluationfunctionpar{i}) 
\xrightarrow{\delta_i} 
(\loc'_i, \variablevaluationfunctionpar{i}', \clockvaluationfunctionpar{i}') 
\xrightarrow{\Lambda_i}
(\loc_{i+1}, \variablevaluationfunctionpar{i+1}, \clockvaluationfunctionpar{i+1}). \nonumber
\end{align}

for some configuration $(\loc'_i, \variablevaluationfunctionpar{i}', \clockvaluationfunctionpar{i}')$. 
Formula $\varphi_{\mathcal{N}}$ is built by assuming that the configuration of the network does not change over the intervals of time delimited by a pair of positions of the \CLTLoc\ model, except for clocks progressing.
Hence, any pair of positions $i$ and $i+1$ of the model of $\varphi_{\mathcal{N}}$ represents (the pair of transitions) $(\loc_i, \variablevaluationfunctionpar{i}, \clockvaluationfunctionpar{i}) \xrightarrow{\delta_i \Lambda_i} (\loc_{i+1}, \variablevaluationfunctionpar{i+1}, \clockvaluationfunctionpar{i+1})$.

The atomic propositions and variables, appearing in $\varphi_{\mathcal{N}}$, are interpreted with the following meaning:
\begin{itemize}
	\item if $l[\autindex]=j$ holds at position \timeposition, then automaton $\mathcal{A}_\autindex$ is in state $q^\autindex_j$ over the interval of time that starts at \timeposition\ and ends in $\timeposition+1$.
	\item if $\tr[\autindex]=j$ holds at position \timeposition, then automaton $\mathcal{A}_\autindex$ performs transition $j$ in $\timeposition+1$.
\end{itemize}
 
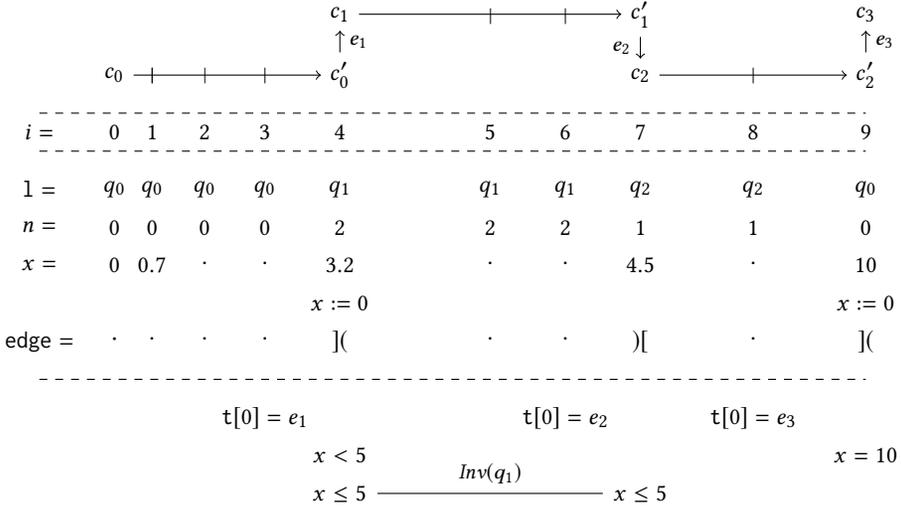
\begin{figure}
	\begin{tikzpicture}
	\draw node at (0,0) (a) {\small $c_0$};
	\draw node at (3,0) (b) {\small $c'_0$};
	\draw node at (3,0.8) (c) {\small $c_1$};
	\draw node at (7,0.8) (d) {\small $c'_1$};
	\draw node at (7,0) (e) {\small $c_2$};
	\draw node at (10,0) (f) {\small $c'_2$};
	\draw node at (10,0.8) (g) {\small $c_3$};
	\draw[->] (a) -- (b) node[midway,above]{};
	\draw[->] (b) -- (c) node[midway,right]{\footnotesize $e_1$};;
	\draw[->] (c) -- (d) node[midway,above]{};
	\draw[->] (d) -- (e) node[midway,left]{\footnotesize $e_2$};;
	\draw[->] (e) -- (f) node[midway,above]{};
	\draw[->] (f) -- (g) node[midway,right]{\footnotesize $e_3$};
	\draw node at (0.5,0)  {\tiny $|$};
	\draw node at (0.5,0)  {\tiny $|$};
	\draw node at (1.2,0)  {\tiny $|$};
	\draw node at (2,0)  {\tiny $|$};
	\draw node at (5,0.776)  {\tiny $|$};
	\draw node at (6,0.776)  {\tiny $|$};
	\draw node at (8.5,0)  {\tiny $|$};
	
	\pgfmathsetmacro{\ilocation}{-0.75}
		\draw[dashed] (-1,-0.5) -- (10,-0.5);
	\draw node at (-1,\ilocation) {\small $\timeposition=$};
	\draw node at (0,\ilocation) {\small $0$};
	\draw node at (0.5,\ilocation) {\small $1$};
	\draw node at (1.2,\ilocation) {\small $2$};
	\draw node at (2,\ilocation) {\small $3$};
	\draw node at (3,\ilocation) {\small $4$};
	\draw node at (5,\ilocation) {\small $5$};
	\draw node at (6,\ilocation) {\small $6$};
	\draw node at (7,\ilocation) {\small $7$};
	\draw node at (8.5,\ilocation) {\small $8$};
	\draw node at (10,\ilocation) {\small $9$};
	\draw[dashed] (-1,-1) -- (10,-1);
	\pgfmathsetmacro{\llocation}{-1.5}
	\pgfmathsetmacro{\dlocation}{-2}
	\pgfmathsetmacro{\xlocation}{-2.5}
	\pgfmathsetmacro{\edgelocation}{-3.5}
	\draw node at (-1,\llocation) {\small $\mathtt{l}=$};
	\draw node at (-1,\dlocation) {\small $\variable=$};
	\draw node at (-1,\xlocation) {\small $\clock=$};
	\draw node at (-1,\edgelocation) {\small $\mathtt{edge}=$};	
	\draw node at (0,\llocation) {\small $q_0$};
	\draw node at (0.5,\llocation) {\small $q_0$};
	\draw node at (1.2,\llocation) {\small $q_0$};
	\draw node at (2,\llocation) {\small $q_0$};
	\draw node at (3,\llocation) {\small $q_1$};
	\draw node at (5,\llocation) {\small $q_1$};
	\draw node at (6,\llocation) {\small $q_1$};
	\draw node at (7,\llocation) {\small $q_2$};
	\draw node at (8.5,\llocation) {\small $q_2$};
	\draw node at (10,\llocation) {\small $q_0$};
	\draw node at (0,\dlocation) {\small $0$};
	\draw node at (0.5,\dlocation) {\small $0$};
	\draw node at (1.2,\dlocation) {\small $0$};
	\draw node at (2,\dlocation) {\small $0$};
	\draw node at (3,\dlocation) {\small $2$};
	\draw node at (5,\dlocation) {\small $2$};
	\draw node at (6,\dlocation) {\small $2$};
	\draw node at (7,\dlocation) {\small $1$};
	\draw node at (8.5,\dlocation) {\small $1$};
	\draw node at (10,\dlocation) {\small $0$};
	\draw node at (0,\xlocation) {\small $0$};
	\draw node at (0.5,\xlocation) {\small $0.7$};
	\draw node at (1.2,\xlocation) {\small $\cdot$};
	\draw node at (2,\xlocation) {\small $\cdot$};
	\draw node at (3,\xlocation) {\small $3.2$};
	\draw node at (3,-3) {\small $x:=0$};
	\draw node at (5,\xlocation) {\small $\cdot$};
	\draw node at (6,\xlocation) {\small $\cdot$};
	\draw node at (7,\xlocation) {\small $4.5$};
	\draw node at (8.5,\xlocation) {\small $\cdot$};
	\draw node at (10,\xlocation) {\small $10$};
	\draw node at (10,-3) {\small $x:=0$};
	\draw node at (0,\edgelocation) {$\cdot$};
	\draw node at (0.5,\edgelocation) {\small $\cdot$};
	\draw node at (1.2,\edgelocation) {\small $\cdot$};
	\draw node at (2,\edgelocation) {\small $\cdot$};
	\draw node at (3,\edgelocation) {\small $]($};
	\draw node at (5,\edgelocation) {\small $\cdot$};
	\draw node at (6,\edgelocation) {\small $\cdot$};
	\draw node at (7,\edgelocation) {\small $)[$};
	\draw node at (8.5,\edgelocation) {\small $\cdot$};
	\draw node at (10,\edgelocation) {\small $]($};
	\draw[dashed] (-1,-4) -- (10,-4);
	\draw node at (2,-4.5) {\small $\tr[0]=e_1$};
	\draw node at (6,-4.5) {\small $\tr[0]=e_2$};
	\draw node at (8.5,-4.5) {\small $\tr[0]=e_3$};
	\draw node at (3,-5) {\small $x<5$};
	\draw node at (10,-5) {\small $x=10$};
	\draw node at (3,-5.5) {\small $x\leq 5$};
	\draw node at (7,-5.5) {\small $x\leq 5$};
	\draw[-] (3.5,-5.5) -- (6.5,-5.5) node[midway,above]{\footnotesize $\mathit{Inv}(q_1)$};
		\end{tikzpicture}
	\caption{Interpretation of atom in $\varphi_{\mathcal{N}}$.}
	\label{fig:phi-N}
\end{figure} 
 
\begin{example}\label{ex:cltloc-model-trace}
Figure~\ref{fig:phi-N} shows a trace of the automaton depicted in Fig.~\ref{fig:TaWithVariableExample}  that consists of various time transitions and three discrete transitions associated with events $e_1$, $e_2$ and $e_3$.
{
To facilitate readability, the discrete transitions such that $\Lambda[0]=\_$ are indicated with a vertical bar and the discrete transitions where at least one automaton executes (in the next position) a transition are drawn by showing the primed configurations.}
Every discrete transition corresponds to a unique position in the \logic\ model and every time transition determines the time progress between pairs of adjacent positions.
The first area below the trace shows the discrete positions \timeposition\ of the \logic\ model and the second one, for each position \timeposition, provides the values of the variables representing location $\loc$, variable \variable, clock \clock\ (a dot stands for a monotonically increasing positive value), and the value of variable $\mathtt{edge}$ (a dot represents an irrelevant assignment to $\mathtt{edge}$).
The first discrete transition labeled with $e_1$ occurs at position $4$, where the guard $\clock<5$ holds; at that moment, clock $\clock$, whose value is equal to $3.2$, is reset and the location changes from $q_0$ to $q_1$.
The second transition---associated with event $e_2$---occurs at position $7$ when $\clock=4.5$ holds, before the value of \clock\ violates the invariant $\clock \leq 5$ of location $q_1$, and produces the change of location from $q_1$ to $q_2$.
The last transition, associated with event $e_3$, occurs at position $9$ with $\clock=10$, it resets $\clock$ and changes location to $q_0$.
In the \logic\ model, discrete transitions are represented one position earlier than the position where they actually occur, namely at position 3, 6, and 8,
respectively.

\begin{remark}
As specified in Definition~\ref{def:trace}, Condition~\ref{firstTransitionIsDelay}, the first configuration change is associated with a time transition.
Thus, it is not possible to fire a discrete transition at position $i=0$.
\end{remark}

The third segment of Fig.~\ref{fig:phi-N} shows the exact positions where transitions $\tr[0] =e_1$,  $\tr[0] =e_2$ and $\tr[0] =e_3$ occur (first line), the positions where the guards are evaluated in the \logic\ model (second line) and the sequence of positions where the invariant of $q_1$ holds (third line). 
For convenience of notation, the assignment for $\mathtt{edge}$ is shown by means of symbols $]($ and $)[$.
At positions 4, 7 and 9 it is shown one among the possible assignments that are compatible with the clock values in the model.
For instance, $)[$ is also possible at position 4.
\end{example}
 
\begin{table*}
\caption{Encoding of the automaton.}
\label{tab:automaton}
\centering
\small
\begin{tabular}{c c c c c c }
\toprule
 \multicolumn{2	}{ c |}{
$\varphi_{1}\coloneqq \underset{ \autindex \in [1,K]}{\bigwedge} (\loc[\autindex]=0) $\hspace{1cm}
} 
&
 \multicolumn{2}{| r |}{
 \hspace{1cm}
$\varphi_{2}  \coloneqq  \underset{\variable \in \variables}{\bigwedge} \var=\variablevaluationfunction^0(\variable)$ 
\hspace{1cm}
}
&
\multicolumn{2}{| r }{
$\varphi_{3} \coloneqq 
\underset{ 
\begin{subarray}{l}
\autindex \in [1,K]
\end{subarray}
}{\bigwedge}  Inv(\loc[\autindex])  $} \\
\midrule
\multicolumn{6}{c}{
	$\varphi_{4} \coloneqq 
	\underset{ 
		\begin{subarray}{l}
		\autindex \in [1,K]\\
		q \in Q_\autindex
		\end{subarray}
	}{\bigwedge} \left(  (\loc[\autindex]=q \land \tr[\autindex] = \notr) \rightarrow \LTLx( Inv(q) \land r_1(Inv(q)) \right)$} 
\\
\midrule
\multicolumn{6}{ c }{
\begin{tabular}{c}
$
\varphi_{5} \coloneqq  
\underset{ 
\begin{subarray}{l}
\autindex \in [1,K],t \in T_k
\end{subarray}
}{\bigwedge}
\tr[k] = {t}
\LTLimplication
\left( \loc[k] = t^- \land  \phi_{\variableconstraint}  \wedge \LTLx (\loc[k]={t^+}   \wedge   \phi_{\clockconstraint} \wedge \phi_{\varassignement} \wedge   \phi_{\resettedclocks} \land \phi_\mathit{edge}(t^-,t^+,k) \right)$ \\ 
$\phi_\mathit{edge}(a,b,i) \coloneqq \phi_{\action^{](}}(a,b,i) \ \lor \ \phi_{\action^{)[}}(a,b,i)$ \\ \\
$\phi_{\action^{](}}(a,b,i) \coloneqq Inv(a) \land r_2(Inv_w(b)) \land \edge[i]$\\ \\
$\phi_{\action^{)[}}(a,b,i) \coloneqq Inv_w(a) \land r_2(Inv(b)) \land \neg \edge[i]$
\end{tabular}
}
\\
\midrule
\multicolumn{6}{c}{
$\varphi_{6} \coloneqq 
\underset{ 
\autindex \in [1,K], 
q,q' \in Q_\autindex \mid
q \neq q'
}{\bigwedge}  \left( ((\loc[\autindex]=q) \wedge \LTLx (\loc[\autindex]=q')) \rightarrow 
\underset{
t \in T_\autindex,
t^- = q,
t^+ = q'}{\bigvee} (\tr[k] = t) \right)$}\\
\midrule
\multicolumn{3}{ c |}{$\varphi_{7} \coloneqq \underset{\clock \in \clocks}{\bigwedge} \left( \LTLx(\clock_0 = 0 \vee \clock_1 = 0) \rightarrow 
\underset{
\begin{subarray}{c}
\autindex \in [1,\numberOfTA]\\
t \in T_\autindex \mid \clock \in t_s
\end{subarray}
}{\bigvee} \tr[k] = t \right)$}
&
\multicolumn{3}{| c}{$\varphi_8 \coloneqq \underset{\variable \in \variables}{\bigwedge} \left( (\neg (\variable =\LTLx (\variable))) \rightarrow \underset{
\begin{subarray}{c}
\autindex \in [1,\numberOfTA]\\
t \in T_\autindex \mid \variable \in U(t)
 \end{subarray}}{\bigvee} \tr[k] = t\right)$} \\
\bottomrule
\end{tabular}
\end{table*}

A network of TA is 
transformed into a \logic\ formula 
using the formulae in the Table~\ref{tab:automaton}.
In the following, the 
invariant of location $q$ is indicated with $Inv(q)$
and the weak version of $Inv(q)$, where all relations $<,>$ are replaced with $\leq,\geq$ and the equalities are replaced with false, is denoted with $Inv_w(q)$.
With slight abuse of notation, $Inv(q)$ and $Inv_w(q)$ are used in Fig.~\ref{tab:automaton} to indicate the CLTLoc formula corresponding to the invariant of $q$ and its weak version.

Before explaining the formulae of Table~\ref{tab:automaton} in details, a short description is first provided. Formulae $\varphi_{1}$, $\varphi_{2}$ and $\varphi_{3}$ specify the initial conditions that must hold in the TA.
Formula $\varphi_{4}$ specifies the behavior of a TA when a time transition is fired.
Formula $\varphi_{5}$ and formulae $\varphi_{6}$-$\varphi_{8}$ 
define, respectively, the necessary and the sufficient conditions that must hold when a discrete transition is performed {with a symbol different from $\notr$, i.e., when the corresponding symbol in the trace is not $\_$\footnote{In $\varphi_{5} \ldots \varphi_{8}$, symbols of \universeOfActions\ do not appear, as events are only used to define how the TA synchronize 
and they will be discussed in the next section.}.}
More precisely, formulae $\varphi_{6}$-$\varphi_{8}$ force the execution of (at least) a transition in the network if a reset or a change of the value of variables or locations occurs and
they prevent a variation in the values of the model that is not caused by the occurrence of a transition.
Each of the formulae is discussed in detail in the following.

\begin{itemize}
\item Formula $\varphi_{1}$ specifies that at position 0 every automaton is in its initial state;
\item Formula $\varphi_{2}$ specifies that at position 0 every variable $\var$ is assigned its initial value $\variablevaluationfunctionpar{0}(\variable)$;
\item  Formula $\varphi_{3}$ specifies that at position 0 the invariant of the initial state of each TA holds;
\item 
Formula $\varphi_{4}$ encodes the case~\ref{epsilon-transition} of Def.~\ref{def:confatrans} requiring that, for every $1\leq \autindex\leq \numberOfTA$, 
if automaton $\mathcal{A}_\autindex$ does not perform a transition
{when other TA do}, 
then the clocks
still satisfy the invariant of the current location $\loc[\autindex]$.
{Formula $Inv(q) \land r_1(Inv(q))$ guarantees that the values of the active clocks satisfy $Inv(q)$, even in the case when they are reset. 
Unlike $Inv(q)$, $r_1(Inv(q))$ does not make use of formulae $\phi_{x\sim \constant}$ to evaluate constraints $x\sim \constant$ when the clock assignments are equal to 0.
In fact, the evaluation of a constraint $x\sim \constant$ through $\phi_{x\sim \constant}$ only depends on the active value of $x$, which is always different from 0 by definition (see Sec.~\ref{sec:clock-encoding}).
To this goal, the mapping $r_1$, defined below, replaces every constraint of the form $\clock \sim \constant$ with the value of $0 \sim \constant$ if $\clock$ is reset. 
Given an atomic formula $\beta(\clock)$ of the form $\clock \sim \constant$ or  $\neg(\clock \sim \constant)$, 
where $\sim \in \{<,=,>\}$, let $\beta_{[x\leftarrow c]}$ be true or false depending on the value of the formula obtained by replacing $\clock$ in $\beta(\clock)$ with the constant $c$.
Then, for a clock constraint $\clockconstraint$, let $r_1(\clockconstraint)$ be defined as the formula obtained from $\clockconstraint$ by replacing, for all clocks $\clock$, 
each occurrence of an atomic formula $\beta(x)$ with the formula: 
}

\begin{align}
(\clock_0=0 \lor \clock_1=0) \rightarrow \beta_{[x\leftarrow 0]}.\nonumber
\end{align}

When $\clock$ is reset, $\clock_0$ or $\clock_1$ are equal to 0, and $\beta_{[x\leftarrow 0]}$ 
 establishes the value of $\beta(x)$.
Since a CLTLoc model represents the occurrence of a transition one time position before the effect,
$\varphi_{4}$ imposes that the invariant $Inv(q)$ associated with the current location $\loc[\autindex]=q$ is satisfied in the next position if no transition is taken in the current one, i.e, when $\trans[\autindex]=\notr$. 
The value of $\edge[\autindex]$ is irrelevant because no transition of $\mathcal{A}_\autindex$ is occurring.
Hence, no constraint is specified for it.

\item Formula $\varphi_{5}$ encodes the case~\ref{firing} of Def.~\ref{def:confatrans}.
Similarly to formula $\varphi_4$, the value of active clocks and their resets have to be considered carefully in the evaluation of the invariants of $\loc[\autindex]$ and $\loc'[\autindex]$.
In the definition of discrete transitions of Def.~\ref{def:confatrans}, the invariant of $\loc[\autindex]$ is always evaluated with respect to $\clockvaluationfunction$ whereas the one of $\loc'[\autindex]$ is evaluated with respect to $\clockvaluationfunction'$.
Function $r_2$ is used to encode the conditions in \ref{invariant-leftclosedrightopen} and \ref{invariant-leftopenrightclosed}, that require the satisfaction of $Inv(\loc'[\autindex])$, or possibly its weak version, with $\clockvaluationfunction'$.
Let $\gamma$, $\beta$ and $\beta_{[x\leftarrow c]}$ be formulae defined as above and let $r_2(\clockconstraint)$ be the formula where all the occurrences $x\sim \constant$ are replaced with the formula 

\begin{align}
((\clock_0>0 \land \clock_1>0) \rightarrow \phi_{\clock \sim \constant}) \land ((\clock_0=0 \lor \clock_1=0) \rightarrow (x \sim \constant)_{[x\leftarrow 0]}) \nonumber
\end{align}

and all the occurrences $\neg(x \sim \constant)$ are replaced with the formula 

\begin{align}
((\clock_0>0 \land \clock_1>0) \rightarrow \neg\phi_{\clock \sim \constant}) \land ((\clock_0=0 \lor \clock_1=0) \rightarrow \neg(x \sim \constant)_{[x\leftarrow 0]}). \nonumber
\end{align}
For instance, by means of $r_2$, the value of a constraint $x\sim \constant$ occurring in $Inv(\loc'[\autindex])$ is either $(x\sim \constant)_{[x\leftarrow 0]}$, if $x$ is reset by the transition (i.e., $\clockvaluationfunction'(x) = 0$); or $\phi_{x\sim \constant}$ if it is not. 
In the latter case, its value is determined by the active clock for $x$ that is equal to $\clockvaluationfunction'(x)$.

Let $t$ be $\transtavar \in T_\autindex$, let 
$\phi_\clockconstraint$ be the \CLTLoc\ formula expressing the guard $\clockconstraint$, 
let $\phi_{\variableconstraint}$ be the formula expressing the constraints on the integer variables,
let $\phi_\resettedclocks$ be the formula $\underset{\clock \in \resettedclocks}{\bigwedge}  (\clock_0 = 0 \vee \clock_1 = 0)$ encoding in \CLTLoc\ the effect of resets applied by $t$ on clocks in $\resettedclocks$ and 
let $\phi_{\varassignement}$ be the formula translating the assignments of the form $\variable:=\exp$ that appear in $t$ into a (semantically) equivalent \logic{} formula.

Recall that $\tr[\autindex]=t$ represents the execution of transition $t$ in the next position of time and that $t^+$, $t^-$ are the locations $q$ and $q'$ connected by $t$.
If $\tr[\autindex]=t$ holds at position $\timeposition$ (hence, transition $t$ is actually performed by $\mathcal{A}_\autindex$ at $\timeposition+1$) then:
\begin{enumerate}
\item Automaton $\mathcal{A}_\autindex$ is currently in $q$ and changes location to $q'$ in the next position of time. Hence, $\loc[\autindex]$ at $i$ and $\timeposition+1$ is, respectively, $q=t^-$ and $q'=t^+$.
\item The condition on the integer variables holds. Formula $\phi_\variableconstraint$ is 
satisfied 
at position \timeposition\ because the effect of $t$ on the integer variables is enforced at $\timeposition+1$.
\item The condition on the clocks holds. Formula $\phi_\gamma$ is satisfied by the clock assignments in $\timeposition+1$.
\item All the clock resets and variable assignments  are performed. Hence, formulae $\phi_\resettedclocks$ and $\phi_\varassignement$ hold at position $\timeposition+1$, i.e., resets and updates are performed when $t$ is actually taken.
\item The configuration change is either open-closed or closed-open. Formula $\phi_\mathrm{edge}$ encodes the cases~\ref{invariant-leftclosedrightopen} and~\ref{invariant-leftopenrightclosed}.
If the configuration change is closed-open then, according to~\ref{invariant-leftclosedrightopen}, $\clockvaluationfunction \models Inv(\loc[\autindex])$ 
and  $\clockvaluationfunction' \models_w Inv(\loc'[\autindex])$ must hold.
The first condition is ensured by $Inv(t^-)$, while the second by $r_2(Inv_w(t^+))$, where $t^- = \loc[\autindex]$ and $t^+ = \loc'[\autindex]$.
Since the transition is closed-open, then $\edge[i]$ is set to true.
If the configuration change is open-closed then, according to~\ref{invariant-leftopenrightclosed}, $\clockvaluationfunction \models_w Inv(\loc[\autindex])$ and $\clockvaluationfunction' \models Inv(\loc'[\autindex])$ must hold.
The first condition is ensured by $Inv_w(t^-)$, while the second by $r_2(Inv(t^+))$, where $t^- = \loc[\autindex]$ and $t^+ = \loc'[\autindex]$.
Since the transition is not closed-open, then $\edge[i]$ is false.

\end{enumerate}
\item Formula $\varphi_{6}$ specifies that if automaton $\mathcal{A}_\autindex$ modifies its location from $q$ to $q'$ over two adjacent positions, then a transition $t \in T_\autindex$ such that $q=t^-$ and $q' = t^+$ is taken at $i+1$---i.e., $\tr[\autindex]=t$ holds at position $i$.
\item Formula $\varphi_{7}$ specifies that if a reset of $x$ (i.e., either $x_0=0$ or $x_1=0$) occurs at $i+1$ then a transition resetting clock $x$ is performed at $i+1$---i.e., $\tr[\autindex]=t$ holds at position $i$, for some $1\leq \autindex\leq \numberOfTA$ and $t\in T_\autindex$ such that $x \in t_s$ (where $t_s$ is the set of clocks that transition $t$ resets).
\item Formula $\varphi_{8}$ specifies that if the value of variable $\variable$ in $i+1$ is not equivalent to the one in $i$ then a transition modifying $\variable$ is performed at $i+1$---i.e., $\tr[\autindex]=t$ holds at position $i$, for some $1\leq \autindex\leq \numberOfTA$ and $t\in T_\autindex$ such that $\variable \in U(t)$ (where $U(t)$ is the set of integer variables that transition $t$ updates by means of the assignments in $\varassignement$).
\end{itemize}

Formula $\varphi_{\netid}$ encoding the network is then defined in Formula~\eqref{f:network}. 
\begin{align}
\label{f:network}
\varphi_{\netid}=\phi_{\mathit{clock}} \land \varphi_{1} \land \varphi_{2} \land  \varphi_{3} \land \LTLg( \underset{4 \leq i \leq 8}{\bigwedge} \varphi_{i})
\end{align}

\begin{remark}
\label{rem:simedges}
The proposed encoding allows the simultaneous execution,  by different automata, of transitions with edges of type  $\action^{)[}$ and $\action^{](}$.
\end{remark}

\newcommand{\n}{\mathtt{n}}

The correctness of the \CLTLoc\ encoding is demonstrated by showing a correspondence between the traces of a network $\netid$ {(Def~\ref{def:trace})} and the models $(\pi,\sigma,\iota)$ of {the \CLTLoc\ formula}  $\varphi_{\netid}$.
Without loss of generality, assume that the set of clocks $\clocks$ of $\netid$ is not empty (if $\clocks = \emptyset$, one could always add a clock that is never reset, nor tested, and the behavior of the network would not change).

At the core of the proof there is a mapping, $\map$, between traces of TA and \logic{} models.
First of all, every trace $\eta$ of TA can be given the form
$$
(\loc_0, \variablevaluationfunctionpar{0}, \clockvaluationfunctionpar{0}) 
\xrightarrow{\delta_0 \Lambda_0} 
(\loc_1, \variablevaluationfunctionpar{1}, \clockvaluationfunctionpar{1}) 
\xrightarrow{\delta_1 \Lambda_1}  
(\loc_2, \variablevaluationfunctionpar{2}, \clockvaluationfunctionpar{2}) \ldots
$$
since any pair of consecutive time transitions can be seen as a pair of time transitions separated by a discrete transition in which the action is $\_$ (i.e., nothing happens in between).
Then, given a trace $\eta$, a \logic{} model $(\pi,\sigma,\iota)$ that belongs to $\map(\eta)$ can be built as follows.

For every position $h \in \naturalNum$, the function $\iota$ assigns each \logic{} variable $\loc[\autindex]$  to location $\loc_{h}[\autindex]$ (that is, $\iota(h, \loc[\autindex]) = \loc_{h}[\autindex]$).
The value of variables $n \in \variables$ is defined as $\iota(h,\variable)=\variablevaluationfunctionpar{h}(\variable)$.
Function $\pi$ assigns values to the atomic propositions of $\loc_h$, for every index $h \in \naturalNum$: $p\in \pi(h)$ if, and only if, $p\in  L(\loc_h[\autindex])$, for some $\autindex$.
The clock valuation $\sigma$  specifies the assignments to both active and inactive clocks.
Recall that $(\loc_h, \variablevaluationfunctionpar{h}, \clockvaluationfunctionpar{h}) 
\xrightarrow{\delta_h \Lambda_h} 
(\loc_{h+1}, \variablevaluationfunctionpar{h+1}, \clockvaluationfunctionpar{h+1})$ is a shortcut for
$$
(\loc_h, \variablevaluationfunctionpar{h}, \clockvaluationfunctionpar{h}) 
\xrightarrow{\delta_h} 
(\loc'_h, \variablevaluationfunctionpar{h}', \clockvaluationfunctionpar{h}') 
\xrightarrow{\Lambda_h}
(\loc_{h+1}, \variablevaluationfunctionpar{h+1}, \clockvaluationfunctionpar{h+1})
$$
where the time transition $(\loc_h, \variablevaluationfunctionpar{h}, \clockvaluationfunctionpar{h}) 
\xrightarrow{\delta_h} 
(\loc'_h, \variablevaluationfunctionpar{h}', \clockvaluationfunctionpar{h}') $ only updates clocks, while locations and integer variables are unchanged.
For convenience of writing, let $x_a$ and $x_i$ be, respectively, the \emph{active} and the \emph{inactive} clocks associated with $x$ in a given position.
Initially, the active clock is $x_0$ (i.e., $x_a$ is $x_0$), and its value is $0$; that is, $\sigma(0,x_0) = 0$ and $\iota(0,\x_v)=0$ (the value of $x_1$---i.e., $x_i$---is arbitrary, hence it can be any positive value).
For all $h \in \naturalNum$, $x \in \clocks$, define $\sigma(h+1,x_a) = \sigma(h,x_a) + \delta_h$; and also $\sigma(h+1,x_i) = \sigma(h,x_i) + \delta_h$ unless $x$ is reset in position $h+1$ of the trace, in which case $\sigma(h+1,x_i) = 0$; clock $x_i$ becomes the active clock from $h+1$ (excluded), and the value of {$\iota(h+2,\x_v)=(\iota(h+1,\x_v)+1) \mod 2$.

The value of predicate $\tr[\autindex]$ at position $h$ is defined based on configurations $(\loc'_h, \variablevaluationfunctionpar{h}', \clockvaluationfunctionpar{h}')$,
$(\loc_{h+1}, \variablevaluationfunctionpar{h+1}, \clockvaluationfunctionpar{h+1})$ and on symbol $\Lambda_h[\autindex]$.
In particular, define $\iota(h,\tr[\autindex])=\notr$ when $\Lambda_h[\autindex]=\_$.
Instead, define $\iota(h,\tr[\autindex])=t$ when there is $t=q \xrightarrow{\clockconstraint, \variableconstraint, \action, \resettedclocks, \varassignement} q'$
such that $\loc'_h[\autindex]=q$, $\loc_{h+1}[\autindex]=q'$, $\Lambda_h[\autindex]=\action$, and $\clockvaluationfunction_h'$, $\clockvaluationfunction_{h+1}$, $\variablevaluationfunctionpar{h}'$ $\variablevaluationfunctionpar{h+1}$ are compatible with $\clockconstraint, \variableconstraint, \resettedclocks, \varassignement$
according to the semantics of Def.~\ref{def:confatrans}.
Finally, the value of $\edge[\autindex]$ is set according to $\Lambda_h[\autindex]$: if $\Lambda_h[\autindex] = (\action, b)$ for some action $\alpha$, then $\edge[\autindex] \in \pi(h+1)$ if, and only if, $b = \ie$; otherwise, if $\Lambda_h[\autindex] = \_$, the value of $\edge[\autindex]$ in $h+1$ is arbitrary (it is also arbitrary in the origin).

Notice that, given a trace $\eta$, its mapping $\map(\eta)$ contains more than one \logic{} model (for example, the value of $x_i$ in the origin is arbitrary).
The inverse mapping $\map^{-1}$, instead, defines, for each \logic{} model $(\pi,\sigma,\iota)$, a unique trace $\eta = \map^{-1}((\pi,\sigma,\iota))$;
indeed, the presence of at least a clock $x$ in $\netid$ entails that at each position $h+1$ at least one of the 
two corresponding \logic{} clocks $x_0, x_1$ is not reset, which in turn uniquely identifies the delay $\delta_{h}$.

The rest of this section sketches the proof for the following result.

\begin{theorem}
\label{theorem:mapping}
Let $\netid$ be a network of TA and $\varphi_\netid$ be its corresponding \logic{} formula. 

For every trace $\eta$ of $\netid$, every \CLTLoc\ model $(\pi, \sigma, \iota)$ such that $(\pi, \sigma, \iota) \in \map(\eta)$ is a model of $\varphi_\netid$.

Conversely, for all \CLTLoc\ models $(\pi, \sigma, \iota)$ of $\varphi_\netid$, $\map^{-1}((\pi, \sigma, \iota))$ is a trace of $\netid$.
\end{theorem}

\begin{proof}
\emph{From traces to models.}\label{sec:traces2models}

Let $\eta =
(\loc_0, \variablevaluationfunctionpar{0}, \clockvaluationfunctionpar{0}) 
\xrightarrow{\delta_0 \Lambda_0} 
(\loc_1, \variablevaluationfunctionpar{1}, \clockvaluationfunctionpar{1}) 
\xrightarrow{\delta_1 \Lambda_1}  
(\loc_2, \variablevaluationfunctionpar{2}, \clockvaluationfunctionpar{2}) \ldots
$
be a trace of $\netid$, and $(\pi, \sigma, \iota) \in \map(\eta)$.
Formulae $\varphi_1$, $\varphi_2$ and $\varphi_3$ of $\varphi_\netid$ are  satisfied since they trivially hold at position $0$.

The following arguments show that, at every position $h \in \naturalNum$, \logic\ formulae $\varphi_{4}, \varphi_{5}$ are satisfied by $(\pi,\sigma,\iota)$ (i.e., $(\pi, \sigma, \iota), h \models \varphi_{4}$ and $(\pi, \sigma, \iota), h \models \varphi_{5}$ both hold).
Different cases are considered, depending on the nature of $\Lambda_h[\autindex]$ (where $1\leq \autindex\leq \numberOfTA$).
\begin{enumerate}
	\item Case $\Lambda_h[\autindex]=\_$. In this case, the  conditions  \ref{epsilon-transition} of Definition~\ref{def:confatrans} hold in the trace:
	\begin{itemize}
		\item $\loc_h[\autindex]=\loc_{h+1}[\autindex]$;
		\item $\clockvaluationfunction_h' \models Inv(\loc_h[\autindex])$ and $\clockvaluationfunction_{h+1} \models Inv(\loc_{h+1}[\autindex])$.
	\end{itemize}
	The antecedent of $\varphi_4$ holds at position $h$ and formula $Inv(\loc[\autindex]) \land r_1(Inv(\loc[\autindex]))$ holds at $h+1$---hence satisfying the entailment---because the second condition of \ref{epsilon-transition} holds in the trace.
	Indeed, all clock constraints $\beta$ in $Inv(\loc[\autindex])$ of the form $x\sim\constant$ correspond to formula $\phi_\beta$ in $\varphi_{4}$, evaluated with the values of the active clocks at position $h+1$.
	Hence, $\clockvaluationfunction_h'(x)$ satisfies $\beta$ if, and only if, $\phi_\beta$ holds at position $h+1$, since by construction $\sigma(h+1,x_a)=\clockvaluationfunction_h'(x)$.
	Moreover, in case of reset of $x$ at $h+1$ in the trace (in which case, by construction, $\sigma(h+1,x_{i})=\clockvaluationfunction_{h+1}(x)=0$ holds), every clock constraint $\beta$ of the form $x\sim\constant$ or $\neg (x\sim\constant)$ corresponds to \logic{} formula $r_1(\beta)$, which in turn reduces to the constant $\beta_{[x\leftarrow 0]}$, and which equals to the evaluation of $\beta$ in $\clockvaluationfunction_{h+1}(x)$.\\
		If $\Lambda_h[\autindex]=\_$, then by construction $\iota(h, \tr[\autindex])=\notr$ and $\varphi_5$ is vacuously satisfied.

	\item Case $\Lambda_h[\autindex]= \action^{](}$. Let $t$ be the transition such that $\iota(h, \tr[\autindex]) = t$.
	By Def.~\ref{def:confatrans} condition~\ref{firing} , the constraints  hold in the trace (notice that $\variablevaluationfunctionpar{h} = \variablevaluationfunctionpar{h}'$):
	\begin{itemize}		
			\item $\clockvaluationfunctionpar{h}' \models \clockconstraint$ and $\variablevaluationfunctionpar{h}' \models \variableconstraint$, 
			\item $\clockvaluationfunctionpar{h+1}(\clock)=0$ for all $\clock \in \resettedclocks$, 
			\item $(\variablevaluationfunctionpar{h+1},\variablevaluationfunctionpar{h}') \models \varassignement$,
			\item $\clockvaluationfunctionpar{h}' \models Inv(\loc_h[\autindex])$ and 
			\item $\clockvaluationfunctionpar{h+1} \models_w Inv(\loc_{h+1}[\autindex])$.
	\end{itemize}
	Thus, formula $\varphi_4$ vacuously holds in $h$.\\
	Formula 	$\varphi_5$ holds since the antecedent of $\varphi_5$ holds for the case $\tr[\autindex]=t$ and
	the consequent  is satisfied as follows. 
	By construction, it holds that $\iota(h, \loc[\autindex]) = \loc_{h}[\autindex]$ and $\iota(h+1, \loc[\autindex]) =  \loc_{h+1}[\autindex]$.
	Formula $\phi_\clockconstraint$ holds at position $h+1$ with $\sigma(h+1,x_a)=\clockvaluationfunction_h'(x)$, for every clock $x$; also, $\phi_\variableconstraint$ holds at position $h$ with $\iota(h,\variable) = \variablevaluationfunctionpar{h}(\variable)=\variablevaluationfunctionpar{h}'(\variable)$ for every variable $n$.
	Formulae $\phi_\resettedclocks$ and $\phi_\varassignement$ hold at $h+1$, since $\sigma(h+1,x_i)=\clockvaluationfunctionpar{h+1}(x)=0$ for all clock $x \in \resettedclocks$, and $\iota(h+1,\variable) = \variablevaluationfunctionpar{h+1}(\variable)$ for all variable $n \in \variables$.
	The first condition on the invariants of $\loc[\autindex]$ is evaluated as follows.
	Each clock constraint $\beta$ in $Inv(t^-)$
	corresponds to \logic{} formula $\phi_\beta$, evaluated with the values of the active clocks at position $h+1$ of $\sigma$. 
	Hence, since $\sigma(h+1,x_a)=\clockvaluationfunction_h'(x)$, $\phi_\beta$ holds at position $h+1$ if, and only if, $\clockvaluationfunction_h'(x)$ satisfies $\beta$.
	The formula $r_2(Inv_w(t^+))$ must hold at $h+1$ to guarantee the enforcement of the last condition in the model.
	By definition of $r_2$, each clock constraint $\beta$ in $Inv_w(t^+)$ of the form $x \sim \constant$ or $\neg(x \sim \constant)$ corresponds to formula $\phi_\beta$, evaluated with the values in $\sigma$ of the active clocks at position $h+1$, if the clock is not reset; otherwise, $\beta$ reduces to $\beta_{[x\leftarrow 0]}$, whose value is that of $\beta$ when valuation $\clockvaluationfunction_{h+1}$ is considered.
	By construction, $\edge[k]  \in \pi(h+1)$ holds, so $\phi_\mathit{edge}(t^-,t^+,k)$ also holds in $h+1$.

	\item Case $\Lambda_h[\autindex]= \action^{)[}$. The proof is similar to the previous one. The only differences are the conditions on the invariants, that are $Inv_w(t^-)$ and $r_2(Inv(t^+))$.
	However, the same arguments of the previous case hold.
\end{enumerate}
Formulae $\varphi_6$
and $\varphi_8$ are trivially satisfied when the location $\loc[\autindex]$ does not change in $\eta$---hence it does not change in $(\pi, \sigma, \iota)$---and for all variables $n$ that have the same value in $h$ and $h+1$; similarly, $\varphi_7$ is satisfied when a clock $x$ is not reset in $h+1$ (hence, both $x_a$ and $x_i$ are not 0).
Otherwise, if between positions $h$ and $h+1$ in the trace the location changes, or a variable is updated, or a clock is reset, then there must be a transition taken in the TA, hence by construction in $(\pi, \sigma, \iota)$ both the antecedents and the consequences of formulae $\varphi_6 - \varphi_{8}$ trivially hold.

\emph{From models to traces.}\label{sec:models2traces}
Let $(\pi, \sigma, \iota)$ be a model of $\varphi_\netid$.
The proof shows that $\eta = \map^{-1}((\pi, \sigma, \iota)) =
(\loc_0, \variablevaluationfunctionpar{0}, \clockvaluationfunctionpar{0}) 
$ $\xrightarrow{\delta_0 \Lambda_0} 
(\loc_1, \variablevaluationfunctionpar{1}, \clockvaluationfunctionpar{1}) 
\xrightarrow{\delta_1 \Lambda_1}  
(\loc_2, \variablevaluationfunctionpar{2}, \clockvaluationfunctionpar{2}) \ldots$ is a trace of $\netid$ according to Def.~\ref{def:trace} (and the related Def.~\ref{def:confatrans}).

Since formulae $\varphi_{1}, \varphi_{2}, \varphi_{3}$  hold in at position 0 of $(\pi, \sigma, \iota)$, 
the values of $\loc_0, \variablevaluationfunctionpar{0}, \clockvaluationfunctionpar{0}$ defined by mapping $\map^{-1}$  constitute an initial configuration of $\netid$ according to Def.~\ref{def:trace}.

The following arguments show that, at each position $h \in \naturalNum$ of $\eta$, for all $\autindex$ (with $1\leq \autindex\leq \numberOfTA$) the conditions of Def.~\ref{def:confatrans} hold.
Separate cases are considered, depending on the value of $\iota(h, \tr[\autindex])$.
	
	\begin{enumerate}
		\item
		$\iota(h, \tr[\autindex]) =\notr$. Then, $\loc[\autindex]$ is the same in $h$ and $h+1$ and $\varphi_4$ ensures that $Inv(\loc[\autindex])\land r_1(Inv(\loc[\autindex]))$ holds at $h+1$, which in turn entails that $\clockvaluationfunction'_{h} \models Inv(\loc[\autindex])$  and $\clockvaluationfunction_{h+1} \models Inv(\loc[\autindex])$ both hold.
		Hence, condition \ref{keepstate} holds.

		\item
		$\iota(h, \tr[\autindex]) = t \neq \notr$ and $\edge[\autindex] \in \pi(h+1)$, where $t=q \xrightarrow{\clockconstraint, \variableconstraint, \action, \resettedclocks, \varassignement} q'$ is a transition of $\mathcal{A}_k$.
		Since $\varphi_5$ holds in $h$, then the following conditions hold: 
		\begin{itemize}
			\item $\iota(h, \loc[\autindex])=q$ and $\iota(h+1, \loc[\autindex])=q'$,
			\item $\phi_\variableconstraint$ holds in $h$ and
			\item $\phi_\clockconstraint$, $\phi_\varassignement$, $\phi_\resettedclocks$ and $\phi_\mathrm{edge}(q,q',\autindex)$
			hold in $h+1$.
		\end{itemize}
		All conditions in~\ref{firing}, and in particular those in \ref{invariant-leftclosedrightopen} corresponding to the case $\action^{](}$, for some action $\action$, hold as follows:
		\begin{enumerate}
			\item $\clockvaluationfunctionpar{h}' \models \clockconstraint$ and $\variablevaluationfunctionpar{h}' \models \variableconstraint$ (condition \ref{guards}) are guaranteed by $\phi_\clockconstraint$ holding in $h+1$ and $\phi_\variableconstraint$ holding in $h$, respectively (notice that in $h+1$ variables are updated by $\varassignement$ and clocks are possibly reset, but $\clockconstraint$ is evaluated through the active clocks).
			\item $\clockvaluationfunctionpar{h+1}(\clock)=0$ (condition \ref{resets}) holds, for all $\clock \in \resettedclocks$, since $\phi_\resettedclocks$ holds in $h+1$.
			\item $(\variablevaluationfunctionpar{h+1},\variablevaluationfunctionpar{h}) \models \varassignement$ (condition \ref{assignments}) holds, since $\phi_\varassignement$ holds in $h+1$.
			\item $\clockvaluationfunctionpar{h}' \models Inv(\loc'_h[\autindex])$ holds because, by $\phi_\mathrm{edge}(q,q',\autindex)$, $Inv(q)$ holds in $h+1$, and $\map^{-1}$ defines that $\loc'_h[\autindex]$ is $q$; then, the first condition of \ref{invariant-leftclosedrightopen} holds.
			\item $\clockvaluationfunctionpar{h+1} \models_w Inv(\loc_{h+1}[\autindex])$ holds because, again by $\phi_\mathrm{edge}(q,q',\autindex)$, $r_2(Inv_w(q'))$ holds in $h+1$, and $\map^{-1}$ defines that $\loc_{h+1}[\autindex]$ is $q'$; then, the second condition of \ref{invariant-leftclosedrightopen} holds.
		\end{enumerate}
		
		\item
		$\iota(h, \tr[\autindex]) = t \neq \notr$ and $\edge[\autindex] \notin \pi(h+1)$.
		The proof, which now focuses on condition \ref{invariant-leftopenrightclosed}, is similar to the previous case---with the only differences being the conditions on the invariants, which are now $Inv_w(q)$ and $r_2(Inv(q'))$---and it is omitted for simplicity.
	\end{enumerate}
	Condition \ref{keepclocksandvariables} follows from $\varphi_7$ and $\varphi_8$ holding at $h$, which impose the occurrence of a transition that resets a clock $x$ or modifies the value of a variable $\variable$ if $x$ is reset or $\variable$ is updated at position $h+1$.
	Finally, condition \ref{time-transition} holds because, as mentioned above, formula $\varphi_{4}$ imposes that $Inv(\loc[\autindex])$ holds at $h+1$, hence $\clockvaluationfunction'_{h} \models Inv(\loc[\autindex])$ also holds, which in turn implies that $\clockvaluationfunction'_{h} \models_w Inv(\loc[\autindex])$ holds (notice that a model for $\varphi_{\netid}$ cannot reach a location that includes constraints of the form $x \sim \constant$ in its invariant, since time is strictly monotonic, and the residence time in the location cannot be null).
	In addition, formula $\varphi_{5}$ defines that $Inv(\loc[\autindex])$ or $Inv_w(\loc[\autindex])$ holds at $h+1$, and both entail that $\clockvaluationfunction'_{h} \models_w Inv(\loc[\autindex])$ holds.
\end{proof}

\subsection{Encoding liveness, synchronization and edge constraints ($\boldsymbol{\varphi_{\mathit{l}}}$, $\boldsymbol{\varphi_{\mathit{s}}}$ and $\boldsymbol{\varphi_{\mathit{ef}}}$)}
\label{subsec:livandsynch}
As seen in Section~\ref{sec:livenessandsynch}, different liveness conditions and synchronization mechanisms for networks of TA can be considered. 
This section describes how the liveness conditions and synchronization mechanisms presented in Section~\ref{sec:livenessandsynch} can be encoded in \CLTLoc.
Several liveness conditions could be required for a network of TA, so Formula $\varphi_\mathit{l}$ captures a conjunction of the following conditions, each one encoded by a \CLTLoc\ formula in Table~\ref{tab:livenessCLTLOoc} (if no liveness condition is required, $\varphi_\mathit{l}$ reduces to true).

\begin{table}[t]
\centering
\caption{Formulae encoding the different liveness conditions.}
\label{tab:livenessCLTLOoc}
\begin{tabular}{ c c }
\toprule
\textbf{Liveness} & \textbf{Property} \\
\toprule
 \makecell{\emph{Strong} \emph{transition}}  &
$ \underset{\autindex \in [1,K]}{\bigwedge} \LTLg \left( \LTLf \left( \underset{\trans \in T_\autindex}{\bigvee} \tr[k]=t \right) \right)$\\
\midrule
\makecell{\emph{Weak} \emph{transition}} &
$  \LTLg \left( \LTLf \left( \underset{\autindex \in [1,K], \trans \in T_\autindex}{\bigvee} \tr[k]=\trans \right) \right)$\\
\midrule
\makecell{\emph{Strong guard}} &
$ \LTLg \left( \underset{\autindex \in [1,K]}{\bigwedge} \left(\underset{q \in Q_\autindex}{\bigwedge} \loc[\autindex]=q \rightarrow \LTLf \left( \bigvee_{t \in T_\autindex, t^- = q} \phi_{\trans_{\gamma}} \wedge \phi_{\trans_{\gamma_{var}}} \right) \right)\right)$ \\
\midrule
\makecell{\emph{Weak guard}} &
$\LTLg \left( \underset{\autindex \in [1,K]}{\bigvee} \left(\underset{q \in Q_\autindex}{\bigwedge}  \loc[\autindex]=q \rightarrow \LTLf \left( \bigvee_{t \in T_\autindex, t^- = q} \phi_{\trans_{\gamma}} \wedge \phi_{\trans_{\gamma_{var}}} \right) \right)\right)$\\
\bottomrule
\end{tabular}
\end{table}

\begin{itemize}
	\item \emph{Strong  transition liveness:} at any time instant, at least one transition in \emph{every}  automaton is eventually fired.
	\item \emph{Weak transition liveness:} at any time instant, at least one transition in \emph{at least one} automaton is eventually fired.
	\item \emph{Strong guard liveness:} at any time instant, \emph{every} automaton eventually reaches a state that has an outgoing transition whose guard holds.
	\item \emph{Weak guard liveness:} at any time instant,  \emph{at least one} automaton eventually reaches a state that has an outgoing transition whose guard holds.
\end{itemize}

\newcommand{\othersynconexcept}{\mathtt{sync}\text{-}\mathtt{on}\text{-}\mathtt{but}}
\newcommand{\sameedge}{\mathtt{same}\text{-}\mathtt{edge}}

\newcommand{\existsinsyncon}{\mathtt{sync}\text{-}\mathtt{on}}

Synchronization is encoded by relying on the following abbreviations, where $ \action \in \actions$, $\automatonIndex,\automatonIndextwo \in [1, \numberOfTA]$ and $S$ is a set of indices in $[1,\numberOfTA]$.
Recall that, given a transition $t$, $t_e \in \actions$ represents the symbols that labels $t$.

\begin{align}
&  \phi_{\existsinsyncon}(\automatonIndex,\action)  \coloneqq  \underset{
	t \in T_{\automatonIndex}\mid
	\trans_e=  \action}{\bigvee}
(t[ \automatonIndex ]= t)  \label{f:existsyncon} \\
&\phi_{\othersynconexcept}(S,\action) 
 \coloneqq  
\underset{	\automatonIndexthree \in \{i \mid i \in [1,\numberOfTA]\}\setminus S}{\bigvee} \phi_{\existsinsyncon}(\automatonIndexthree,\action) \label{f:existsynconnoautomaton}
\\
&\phi_{\sameedge}(\automatonIndex,\automatonIndextwo) 				    \coloneqq  \LTLx (\rclosed{\automatonIndex}\leftrightarrow\rclosed{\automatonIndextwo}) 	 \label{f:sameedge}		\\		\nonumber 
\end{align}

\noindent Formula $\phi_{\existsinsyncon}(\automatonIndex,\action)$  specifies that  a transition $t$ of $\mathcal{A}_\automatonIndex$ labeled with the action \action\ is fired.
Formula $\phi_{\othersynconexcept}(S,\action)$ specifies that a transition $t$ 
labeled with the action \action , and belonging to a
TA whose index does not belong to set $S$, 
is fired.
Finally, $\phi_{\sameedge}(\automatonIndex,\automatonIndextwo)$ specifies that the transitions taken by $\mathcal{A}_\automatonIndex$ and $\mathcal{A}_\automatonIndextwo$ have the same edge structure, i.e., either they are both open-closed or both closed-open.

\begin{table*}[t]
\centering
\caption{Formulae encoding different types of synchronizations.}
\label{tab:synchronizationCLTLocFormulae}
\small
\begin{tabular}{ p{1.2cm}  p{12cm} }
\toprule
\textbf{Name} & \textbf{Property} \\
\toprule
\multirow{2}{*}{\shortstack[l]{Channel}} 
&
$\upsilon_1 \coloneqq 
\underset{
\begin{subarray}{c}
\autindex \in [1, \numberOfTA] \\
 t \in T_\autindex \mid \trans_e = \action !
\end{subarray}
}{\bigwedge} 
 \left( t[\autindex]=t \LTLimplication 
 \underset{
\begin{subarray}{c} 
\autindextwo \in [1, \numberOfTA],   \autindextwo \neq \autindex  \end{subarray}
  }{\bigvee} 
 \left(
	\begin{gathered}
	 \phi_{\existsinsyncon}(\autindextwo,\action?) 
	  \wedge 
	 \neg\phi_{\othersynconexcept}(\{\autindex,\autindextwo\},\action?) \\
	 \wedge \\
	 \phi_{\sameedge}(\autindex,\autindextwo)
	  \end{gathered}  
\right)
 \right)$\\
& 
\\
&
    $\upsilon_2 \coloneqq 
\underset{
\begin{subarray}{c} 
\autindex \in [1, \numberOfTA]  \\
  t \in T_\autindex \mid \trans_e =\action ?
  \end{subarray}
  }{\bigwedge}
\left(
 t[\autindex]=t \LTLimplication 
 \underset{
\begin{subarray}{c}  
	\autindextwo \in [1, \numberOfTA], \autindextwo \neq \autindex
 \end{subarray} 
 }{\bigvee} 
\left( 
	 \phi_{\existsinsyncon}(\autindextwo,\action!) 
  \wedge  
 \neg\phi_{\othersynconexcept}(\{\autindex,\autindextwo\},\action!)
\right)
 \right)
$\\
\midrule 
\multirow{2}{*}{\shortstack[l]{Broadcast} 
 }& 
$
	\upsilon_1 \coloneqq 
		\begin{gathered} 
	\underset{
		\begin{subarray}{c}
		\autindex \in [1, \numberOfTA]  \\
		t \in T_\autindex \mid
		\trans_e  = \action \#
		\end{subarray}
	}{\bigwedge}  
	\left( \tr[\autindex]=t \rightarrow
	\left(	
	\neg\phi_{\othersynconexcept}(\{\autindex\},\alpha \#)	
		\right)
			\right)	
				\\ \wedge \\
	\underset{
		\begin{subarray}{c}
		\autindex \in [1, \numberOfTA]  \\
		t \in T_\autindex \mid
		\trans_e  = \action \#
		\end{subarray}
	}{\bigwedge}  
	\left( \tr[\autindex]=t \rightarrow
	\left(	
	\underset{
		\begin{subarray}{c}
		\autindextwo \in [1, \numberOfTA] \\
		\autindextwo \neq \autindex
		\end{subarray}}{\bigwedge}  
	\left(
	\begin{gathered}
	\phi_{\existsinsyncon}(\autindextwo,\action @) \land \phi_{\sameedge}(\autindex,\autindextwo)
	\\ \vee \\
	\left(
	\begin{gathered}
	\underset{
		\begin{subarray}{l}
		t'  \in T_{k'}\mid
		t'_e=\alpha @
		\end{subarray}}{\bigwedge}  (\LTLx(\neg \phi_{t'_\gamma}) \LTLor \neg \phi_{t'_{\gamma_{var}}} \lor l[\autindextwo]\not=t'^-)\\
	\end{gathered}\right)
	\end{gathered}
	\right) 
	\right)
	\right)
	\end{gathered}
	$ 	\vspace{0.5cm}
   \\
   &  $
\upsilon_2 \coloneqq 
\underset{
\begin{subarray}{c}
\autindex  \in [1, \numberOfTA]  \\
 t \in T_\autindex \mid \trans_e =\alpha @
  \end{subarray}}{\bigwedge}
\left( t[\autindex]= t \LTLimplication 
\phi_{\othersynconexcept}(\{k\},\alpha \#)
  \right) 
  $
   \\
   \midrule
   \parbox[t]{2.5cm}{One-to-\\many} &
$
	\upsilon_1 \coloneqq 
	\underset{
		\begin{subarray}{c}
		\autindex \in [1, \numberOfTA]  \\
		t \in T_\autindex \mid
		\trans_e  = \action \&
		\end{subarray}
	}{\bigwedge}  
	\left( \tr[\autindex]=t \rightarrow
	\left(	
	\neg\phi_{\othersynconexcept}(\{\autindex\},\action \&)	
	\wedge
	\phi_{\othersynconexcept}(\{\autindex\},\action *)
	\right)
	\right)
			$ 	   
   \\  
   &  $
\upsilon_2 \coloneqq 
\underset{
\begin{subarray}{c}
\autindex  \in [1, \numberOfTA] , \\
 \trans \in T_\autindex \mid \trans_e =\action *
  \end{subarray}}{\bigwedge}
\left( t[\autindex]= t \LTLimplication 
\phi_{\othersynconexcept}(\{k\},\action \&)  
  \right) 
  $\\
   \bottomrule
\end{tabular}
\end{table*}

The abbreviations in Formulae~\eqref{f:existsyncon},~\eqref{f:existsynconnoautomaton} and~\eqref{f:sameedge} are used in Table~\ref{tab:synchronizationCLTLocFormulae}  to encode the channel-based, the broadcast and the one-to-many synchronizations.
The intermediate \CLTLoc\ formula $\phi_{\mathit{sync\_type}}$ (where $\mathit{sync\_type}$ is $\mathit{channel}$, $\mathit{broadcast}$ or $\textit{one-to-many}$) is $\phi_{\mathit{sync\_type}} \coloneqq \LTLg (\upsilon_1 \wedge \upsilon_2)$, where $\upsilon_1$ and $\upsilon_2$ depend on the selected type of synchronization.
Then, if multiple synchronizations are considered, $\phi_{\mathit{s}}$ is the conjunction of the corresponding $\phi_{\mathit{sync\_type}}$, one for each type of synchronization.
Note that the syntax of MITL adopted in this work does not allow event symbols to appear in the formulae, being the language limited to constrain the values of the variables in \variables\ and the atomic propositions in \propositions\ over the time. 
For this reason, symbols of \universeOfActions\ do not have a corresponding \logic\ representation, yet they are used to instantiate the formulae encoding the synchronization among the TA. 

\begin{itemize}
\item \emph{Channel-based synchronization.}
Formula $\upsilon_1$ specifies that any sending event $\action \channelsend$ in a transition $t$ of an automaton $\autindex$ must be matched by exactly one corresponding receiving event $\action \channelreceive$ of a transition $t^\prime$ of  another automaton $\autindextwo$.
This is specified by stating that there exists one of the automata  with index $\autindextwo$ that syncs on $\action \channelreceive$ and all the others (with index different than $\autindex$ and $\autindextwo$) do not sync on $\action \channelreceive$. Furthermore, the shape of the edges of the transitions of the automata that sync on action $\action$ must correspond.

Formula  $\upsilon_2$ specifies that any receiving event $\action \channelreceive$ must be matched by exactly one corresponding sending event $\action \channelsend$ in one of the other automata.

\item \emph{Broadcast synchronization.} 
Formula $\upsilon_1$ first specifies that if an automaton $\autindex$ broadcasts on a channel $\action$, no other automaton broadcasts on that channel.
Then, it specifies that if an automaton $\autindex$ broadcasts on a channel $\action$, all the other automata $\autindextwo$ either sync and receive on that channel, and also match the shape of the transition, or they do not sync.
If the automaton does not sync, either it is in a state that has no outgoing transition labeled with  $\action \broadcastreceive$, or the guards of the outgoing transitions labeled with $\action \broadcastreceive$ are not satisfied (i.e., no transitions are enabled).

Formula  $\upsilon_2$ specifies that any receiving event $\action \broadcastreceive$ must be matched by exactly one corresponding sending event $\action \broadcastsend$ in one of the other automata.

\item \emph{One-to-many.}
Formula $\upsilon_1$ specifies that if an event $\action \manytomanysend$ is sent, no other automaton sends the same event, and
 at least one automaton receives the event $\action \manytomanyreceive$. 
Formula $\upsilon_2$ specifies that if an event $\action \manytomanyreceive$ is received by an automaton, some automaton has sent event $\action \manytomanysend$.
\end{itemize}

As mentioned in Remark~\ref{rem:simedges}, the semantics---and corresponding encoding---of networks of TA allows for transitions with different types of edged to be taken at the same time.
As depicted in Fig.~\ref{fig:phi-N}, it also allows the same automaton to take transitions with different edges over time.
However, one might desire to restrict this behavior (for example for synchronization reasons), and only allow transitions to be taken with a certain type of edge.
These restrictions (if any), are captured by formula  $\varphi_{\mathit{ef}}$.
Table~\ref{tab:edges} shows some examples of restrictions and corresponding \logic{} constraints.
In particular, the ``closed-open'' (resp., open-closed) restriction states that, when a transition is taken, it must with some symbol $\action^{](}$ (resp., $\action^{)[}$).
The ``unrestricted'' case obviously means that no constraint is introduced, hence it corresponds to true.
Other possibilities could be envisaged, but these are the most relevant for our purposes.

The following theorem extends Theorem \ref{theorem:mapping} to liveness conditions and synchronization mechanisms.

\begin{theorem}
\label{theorem:mapping2}
Let $\netid$ be a network of TA,
$l$ be a (set of) liveness conditions,
$s$ be the (set of) synchronization mechanisms used in network $\netid$, and $\mathit{ef}$ the set of restrictions on edges.
Let  $\Phi_\netid$ 
be the \logic{} formula  $ \varphi_{\mathcal{N}} 
\wedge \varphi_{\mathit{l}} \wedge \varphi_{\mathit{s}} \land \varphi_{\mathit{ef}}$.

For every trace $\eta$ of $\netid$ that satisfies the selected liveness conditions $l$, synchronization mechanisms $s$, and edge restrictions $edge$,
any \CLTLoc\ model $(\pi, \sigma, \iota)$ such that $(\pi, \sigma, \iota) \in \map(\eta)$ is a model of $\Phi_\netid$.

Conversely, for each \CLTLoc\ model $(\pi, \sigma, \iota)$ of $\Phi_\netid$, $\map^{-1}((\pi, \sigma, \iota))$ is a trace of $\netid$ that satisfies the selected liveness conditions $l$, synchronization mechanisms $s$, and edge restrictions $ef$.
\end{theorem}

The proof is omitted for reasons of brevity, as it is rather standard. Indeed, it is a straightforward extension of the proof of Theorem~\ref{theorem:mapping}, and follows from the fact that each 
CLTLoc formula listed in Table~\ref{tab:livenessCLTLOoc} (resp., Table~\ref{tab:synchronizationCLTLocFormulae}) encodes the corresponding semantics described in Table~\ref{tab:livenesssemantic} (resp., Table~\ref{tab:networksemantics}); in addition, the formulae of Table~\ref{tab:edges} capture the corresponding restrictions on edges.

\begin{table*}[t]
\centering
\caption{Formulae encoding different types of edges.}
\label{tab:edges}
\small
\begin{tabular}{ c  c }
\toprule
\textbf{Name} & \textbf{Property} \\
\toprule
closed-open & $\LTLg  \underset{\automatonIndex [1,K]}{\bigwedge}  \rclosed{\automatonIndex}$ \\
open-closed & $\LTLg  \underset{\automatonIndex [1,K]}{\bigwedge}  \neg \rclosed{\automatonIndex}$  \\
unrestricted & $\top$  \\ 
   \bottomrule
\end{tabular}
\end{table*}

\section{Checking the satisfaction of MITL formulae over TA}
\label{sec:checkingMITLI}
Traces encode executions of TA by means of denumerable sequences of time and discrete transitions.
However, the evolution of a network of TA is continuous, hence it is more naturally represented by means of \emph{signals} (see Section \ref{sec:mitl}).
{Timed words, instead of signals, are commonly adopted to represent the semantics of TA: although they are expressive enough in many cases, they cannot describe 
the values at the edge of signals---i.e., in correspondence of configuration changes.
For instance, in a temporal logic such as MITL one can indeed state properties whose value is affected by signal edges; e.g., the set 
of signals of symbol \proposition{} that change value only over intervals that are left-closed/right-open can be specified with the MITL formula 
{$(\proposition \Rightarrow \proposition \LTLuntil_{(0,+\infty)} \top) \land (\neg \proposition \Rightarrow \neg \proposition \LTLuntil_{(0,+\infty)} \top)$}.
Therefore, model checking a TA against such an expressive language requires modeling edges as well.
}

Traces are tightly bound to signals: 
intuitively, given a trace \tracesymbol, the projection over the real line of the values of its integer variables and atomic propositions  determines a signal \signalsymbol.
To be able to consistently associate signals with traces of a TA, however, we need impose the following restriction on traces.

\begin{definition}
\label{rem:tracerestriction}
Let $\netid$ be a network of TA. A trace $\eta$ of $\netid$ is \emph{edge-consistent} if, for any configuration change $(\loc'_h, \variablevaluationfunctionpar{h}', \clockvaluationfunctionpar{h}') 
\xrightarrow{\Lambda_h}
(\loc_{h+1}, \variablevaluationfunctionpar{h+1}, \clockvaluationfunctionpar{h+1})$
there are two transitions $ \loc'_h[ k ] \xrightarrow{\clockconstraint, \variableconstraint, \action, \resettedclocks, \varassignement} \loc_{h+1}^\prime[ k ]$ and $\loc'_h[ \bar{k} ] \xrightarrow{\bar{\clockconstraint}, \bar{\variableconstraint}, \bar{\action}, \bar{\resettedclocks}, \bar{\varassignement}} \loc_{h+1}^\prime[ \bar{k} ]$, of two distinct TA $k,\bar{k}$, which both set the value of variable $\variable$ (in a compatible manner), then the edge of the transitions is the same; that is, either they are $\action^{](}$ and $\bar{\action}^{](}$, or they are $\action^{)[}$ and $\bar{\action}^{)[}$.
\end{definition}

In the rest of this section, only traces that are edge-consistent are considered.

Let \tracesymbol\ be an edge-consistent trace 
{$
(\loc_0, \variablevaluationfunctionpar{0}, \clockvaluationfunctionpar{0}) 
\xrightarrow{\delta_0} 
(\loc'_0, \variablevaluationfunctionpar{0}', \clockvaluationfunctionpar{0}') 
\xrightarrow{\Lambda_0}
(\loc_1, \variablevaluationfunctionpar{1}, \clockvaluationfunctionpar{1}) 
\xrightarrow{\delta_1}  \ldots
$; we indicate by $\eventtime(e)$ the ``time'' of a symbol $e$ (where $e$ can be either $\delta$ or $\Lambda$), defined as follows:}
\begin{itemize}
	\item $\eventtime(\delta_0) = 0$;
	\item $\eventtime(\Lambda_h) = \eventtime(\delta_{h}) + \delta_{h}$ {for all $h\geq 0$};
	\item $\eventtime(\delta_h) = \eventtime(\Lambda_{h-1})$ {for all $h > 0$}.
\end{itemize}

Recall that, given a trace $\eta$, its associated word $w(\eta)$ is the sequence $\delta_0 \Lambda_0  \delta_1 \Lambda_1 \dots$; also, given a set of assignments $\mu$, $U(\mu)$ is the set of variables updated by $\mu$.
Let $(\loc,$ $\variablevaluationfunction,$ $\clockvaluationfunction)$ be a configuration; we denote as  $c(\loc, \variablevaluationfunction, \clockvaluationfunction)$ 
the pair $(\cup_{1\leq k \leq K} L(\loc_h[k]), \variablevaluationfunctionpar{h}) \in \partSet(AP) \times \intNum^\variables$
of the atomic propositions and variable assignments that hold in the configuration $(\loc,$ $\variablevaluationfunction,$ $\clockvaluationfunction)$.

\begin{definition}\label{def:signal}
{Let \tracesymbol\ be an edge-consistent trace of a network $\netid$ of TA.
}
	The \emph{signal} \signalsymbol\ associated with \tracesymbol\ is the function $\signalsymbol:   \realNum_{\geq 0} \rightarrow \partSet(AP) \times \intNum^\variables $ such that:
	\begin{enumerate}
		\item\label{signal-time-0} $\signalsymbol(0)=c(\loc_0, \variablevaluationfunctionpar{0}, \clockvaluationfunctionpar{0})$;
		\item\label{signal-time-t} for all $\delta_h$ in $w(\tracesymbol)$, for all $r \in \realNum_{\geq 0}$ such that $\eventtime(\delta_h) < r < \eventtime(\delta_h)+\delta_h$
		then $\signalsymbol(r)=c(\loc_h, \variablevaluationfunctionpar{h}, \clockvaluationfunctionpar{h})$;
		\item\label{signal-discrete-t} for all $\Lambda_h$ in $w(\tracesymbol)$, $\signalsymbol(\eventtime(\Lambda_h))=(A, \variablevaluationfunction) \in \partSet(AP) \times \intNum^\variables$ where, for all $p \in AP$ and $n \in Int$:
		\begin{enumerate}
			\item \label{signal-discrete-t-atoms} $\proposition \in A$ if, for some $\action \in \actions$ and for some $1\leq \autindex \leq \numberOfTA$: 
			\begin{itemize}
				\item $\proposition \in L(\loc_h[k])$ and $\Lambda_h[\autindex] \in \{\_, \action^{](}\}$ holds, or
				\item $\proposition \in L(\loc_{h+1}[k])$ and $\Lambda_h[\autindex] = \action^{)[}$ holds
			\end{itemize}
			\item \label{signal-discrete-t-var-I} $\variablevaluationfunction(\variable) = \variablevaluationfunctionpar{h}(\variable)$ if one of the following conditions holds:
			\begin{itemize}
				\item there is no transition $\loc'_{h}[\autindex] \xrightarrow{\clockconstraint, \variableconstraint, \action, \resettedclocks, \varassignement} \loc_{h+1}[\autindex]$ compatible with the configuration change and such that $\variable\in U(\mu)$; 
				\item 
				 there is $1\leq \autindex\leq \numberOfTA$ and a transition $\loc'_{h}[\autindex] \xrightarrow{\clockconstraint, \variableconstraint, \action, \resettedclocks, \varassignement} \loc_{h+1}[\autindex]$---compatible with the configuration change---such that $\Lambda_h[\autindex] = \action^{](}$ and $\variable\in U(\mu)$.
			\end{itemize} 
			\item \label{signal-discrete-t-var-II} $\variablevaluationfunction(\variable) = \variablevaluationfunctionpar{h+1}(\variable)$ if there is $1\leq \autindex\leq
			\numberOfTA$ and a transition $\loc'_{h}[\autindex] \xrightarrow{\clockconstraint, \variableconstraint, \action, \resettedclocks, \varassignement} \loc_{h+1}[\autindex]$---compatible with the configuration change---such that $\Lambda_h[\autindex] = \action^{)[}$ and $\variable\in U(\mu)$ hold.

		\end{enumerate}
		
	\end{enumerate}
\end{definition}

Condition~\ref{signal-time-t} defines the correspondence between the time transitions in a trace \tracesymbol\ and the values of the signal within the left-open/right-open intervals of \signalsymbol.
In particular, any trace $\tracesymbol$ defines an infinite set of intervals $I_h$ of the form $( \eventtime(\delta_{h}),$ $\eventtime(\delta_{h})+\delta_h )$, for all $h\geq 0$.
The value of signal $M_\tracesymbol$ in every interval $I_h$ is determined
by
the propositions and variable assignments $c(\loc_h, \variablevaluationfunctionpar{h}, \clockvaluationfunctionpar{h})$ that hold in $I_h$.

Condition \ref{signal-discrete-t} handles the case of a discrete transition in $h$ and defines the value of \signalsymbol\ at time $\eventtime(e_{h})$ when a configuration change occurs (i.e., when $e_h = \Lambda_h$ holds). 
Conditions~\ref{signal-discrete-t-atoms},~\ref{signal-discrete-t-var-I} and~\ref{signal-discrete-t-var-II} define, respectively, the atomic propositions  and the value of the integer variables based on the transition $\loc'_{h}[\autindex] \xrightarrow{\clockconstraint, \variableconstraint, \action, \resettedclocks, \varassignement} \loc_{h+1}[\autindex]$ performed at time $\eventtime(\Lambda_{h})$ by automaton $\mathcal{A}_\autindex$, for every $1\leq \autindex\leq \numberOfTA$.
Condition~\ref{signal-discrete-t-atoms} specifies that $M_\tracesymbol(\eventtime(\Lambda_h))$ includes the atomic propositions  in $L(\loc_{h}[\autindex])$ if the discrete transition performed by $\mathcal{A}_\autindex$ is closed-open (i.e., $\Lambda_h[\autindex]=\action^{](}$ holds), or if no transition is taken;
otherwise, if the discrete transition is open-closed (i.e., $\Lambda_h[\autindex]=\action^{)[}$ holds), $M_\tracesymbol(\eventtime(\Lambda_h))$ includes the atomic propositions  in $L(\loc_{h+1}[k])$.
Conditions~\ref{signal-discrete-t-var-I} and~\ref{signal-discrete-t-var-II} define the value $\variablevaluationfunction(\variable)$ of variable $\variable$ at time $\eventtime(\Lambda_{h})$.
The value of $\variablevaluationfunction(\variable)$ is the same as $\variablevaluationfunctionpar{h}(\variable)$ if there is an automaton $\mathcal{A}_\autindex$ that performs a closed-open discrete transition that modifies $\variable$, or if none of the automata updates $\variable$ (condition~\ref{signal-discrete-t-var-I}).
Conversely (condition~\ref{signal-discrete-t-var-II}), the value of $\variable$ becomes $\variablevaluationfunctionpar{h+1}(\variable)$ at time $\eventtime(\Lambda_{h})$ if there is an automaton $\mathcal{A}_\autindex$ that performs a open-closed discrete transition that modifies $\variable$.

Consider a trace $\tracesymbol$ and its associated signal $M_\tracesymbol$.
$M_\tracesymbol$ is \emph{left-closed} when, for all $r \in \realNum_{\geq 0}$, if $\signalsymbol(r)=(A, \variablevaluationfunction)$ for some $A$, $\variablevaluationfunction$, then there is $\varepsilon \in \realNum_{>0}$ such that, for all $r < r' < r + \varepsilon$ it also holds $\signalsymbol(r')=(A, \variablevaluationfunction)$.
Dually, $M_\tracesymbol$ is \emph{right-closed} when, for all $r \in \realNum_{> 0}$, if $\signalsymbol(r)=(A, \variablevaluationfunction)$ for some $A$, $\variablevaluationfunction$, then there is $\varepsilon \in \realNum_{>0}$ such that, for all $r - \varepsilon < r' < r$ it also holds $\signalsymbol(r')=(A, \variablevaluationfunction)$.
$M_\tracesymbol$ is \emph{unrestricted} if there are no constraints on the value of the signal in the neighborhood of each time instant.

The shape of a signal $M_\tracesymbol$ is determined by the transitions taken by the automata of the network.
Consider, for instance, a variable \variable\ with value $2$, and a transition $t = \transtavar$ that, when it is taken, assigns value $1$ to \variable.
There are different possibilities concerning the value of variable \variable{} in the instant when $t$ is taken: if it must be $2$ (i.e., it is not yet assigned by the transition), then the corresponding signal cannot be left-closed; if it must be $1$ (i.e., it is already assigned by the transition), then the signal cannot be right-closed.
There is an obvious relation (captured by Table~\ref{tab:signalsemantics}) between a signal $M_\tracesymbol$ being right-closed, left-closed, or unrestricted, and the edges of the transitions taken in $\tracesymbol$.
Indeed, when all edges are open-closed (i.e., $\action^{)[}$), the signal is left-closed; when they are all closed-open (i.e.,  $\action^{](}$) the signal is right-closed; when they can be both, the signal is unrestricted.
Leaving signals unrestricted---which means that, when transitions are fired, the choice of whether variables are already assigned their new values or still retain their old ones is non-deterministic---is a common approach in literature that has been used in some seminal works on TA \cite{alur1996benefits}.

Then, by imposing constraints on the types of edges that the transitions of a network $\netid$ of TA can have, one can restrict the set of corresponding signals to contain only left-closed or only right-closed ones.

\begin{table}
	\caption{Definition of different types of signals based on the symbols occurring in the corresponding traces.}
	\label{tab:signalsemantics}
	\begin{tabular}{c c}
		\toprule
		\textbf{Signal} & \textbf{Discrete transitions}\\
		\midrule
		right-closed & $\alpha^{](}$ \\
		left-closed  & $\alpha^{)[}$ \\
		unrestricted & $\alpha^{](}$ or $\alpha^{)[}$ \\
		\bottomrule
	\end{tabular} 
\end{table}

\begin{definition}
\label{rem:traceset}
Let $\netid$ be a network of TA,
$l$ be a set of liveness conditions selected from Table~\ref{tab:livenesssemantic},
$s$ be the synchronization primitives (among those of Table~\ref{tab:networksemantics}) used in $\netid$,
and $\mathit{ef}$ a restriction on the types of edges for the transitions taken by $\netid$ (such as those of Table~\ref{tab:signalsemantics}).

$\mathcal{T}(\netid, l, s, \mathit{ef})$ is the set of edge-consistent traces that satisfy the liveness conditions $l$, the semantics of synchronizations $s$, and the restriction on edges $\mathit{ef}$.

In addition, $\mathcal{S}(\netid, l, s, \mathit{ef})$ is the corresponding set of signals.
\end{definition}

\begin{example}
Figure~\ref{fig:differentsemantics} shows the same trace of Fig.~\ref{fig:phi-N}
and the values of $\loc$ and $\variable$ that contribute to the definition of the signal deriving from the trace.
At the bottom, three different types of signals---as specified on the right-hand side of the figure---are drawn according to the selected restrictions on the edges.
The last signal is based on the assignment of variable $\mathtt{edge}$ in Fig.~\ref{fig:phi-N}, whereas the first two are derived by restricting the kind of discrete transitions of the traces.
Similarly to Fig.~\ref{fig:phi-N}, the value of $\mathtt{edge}$ is explicitly written by means of $]($ or $)[$.
In correspondence to the events $e_1$, $e_2$ and $e_3$ its value defines the edge of the current interval, while it is left unspecified in the other positions to avoid cluttering the figure.
\end{example}	

\begin{figure}
	\begin{tikzpicture}
	\draw node at (0,0) (a) {\small $c_0$};
	\draw node at (3,0) (b) {\small $c'_0$};
	\draw node at (3,0.8) (c) {\small $c_1$};
	\draw node at (7,0.8) (d) {\small $c'_1$};
	\draw node at (7,0) (e) {\small $c_2$};
	\draw node at (9,0) (f) {\small $c'_2$};
	\draw node at (9,0.8) (g) {\small $c_3$};
	\draw[->] (a) -- (b) node[midway,above]{};
	\draw[->] (b) -- (c) node[midway,right]{\footnotesize $e_1$};;
	\draw[->] (c) -- (d) node[midway,above]{};
	\draw[->] (d) -- (e) node[midway,left]{\footnotesize $e_2$};;
	\draw[->] (e) -- (f) node[midway,above]{};
	\draw[->] (f) -- (g) node[midway,right]{\footnotesize $e_3$};
	\draw node at (0.5,0)  {\tiny $|$};
	\draw node at (0.5,0)  {\tiny $|$};
	\draw node at (1.7,0)  {\tiny $|$};
	\draw node at (2,0)  {\tiny $|$};
	\draw node at (5,0.776)  {\tiny $|$};
	\draw node at (6,0.776)  {\tiny $|$};
	\draw node at (8.5,0)  {\tiny $|$};
	
	\pgfmathsetmacro{\ilocation}{-0.75}
	\draw[dashed] (-1,-0.5) -- (10,-0.5);
	\pgfmathsetmacro{\llocation}{-1}
	\pgfmathsetmacro{\dlocation}{-1.5}
	\pgfmathsetmacro{\xlocation}{-2.0}
	\pgfmathsetmacro{\edgelocation}{-2.0}
	\draw node at (-1,\llocation) {\small $\mathtt{l}=$};
	\draw node at (-1,\dlocation) {\small $\variable=$};
	\draw node at (0,\llocation) {\small $q_0$};
	\draw node at (0.5,\llocation) {\small $q_0$};
	\draw node at (1.7,\llocation) {\small $q_0$};
	\draw node at (2,\llocation) {\small $q_0$};
	\draw node at (3,\llocation) {\small $q_1$};
	\draw node at (5,\llocation) {\small $q_1$};
	\draw node at (6,\llocation) {\small $q_1$};
	\draw node at (7,\llocation) {\small $q_2$};
	\draw node at (8.5,\llocation) {\small $q_2$};
	\draw node at (9,\llocation) {\small $q_0$};
	\draw node at (0,\dlocation) {\small $0$};
	\draw node at (0.5,\dlocation) {\small $0$};
	\draw node at (1.7,\dlocation) {\small $0$};
	\draw node at (2,\dlocation) {\small $0$};
	\draw node at (3,\dlocation) {\small $2$};
	\draw node at (5,\dlocation) {\small $2$};
	\draw node at (6,\dlocation) {\small $2$};
	\draw node at (7,\dlocation) {\small $1$};
	\draw node at (8.5,\dlocation) {\small $1$};
	\draw node at (9,\dlocation) {\small $0$};
	\draw[dashed] (-1,-2) -- (10,-2);
	\draw node at (-1,-3.5) {\small $\mathtt{edge}=$};	
	\draw node at (0,-3)  {$[$};
	\draw node at (2.98,-3)  {$]$};
	\draw[-] (3,-3) -- (0,-3) node[midway,above]{\footnotesize $c_0$};
	\draw node at (3.02,-3)  {$($};
	\draw node at (6.98,-3)  {$]$};
	\draw[-] (7,-3) -- (3,-3) node[midway,above]{\footnotesize $c_1$};
	\draw node at (7.02,-3)  {$($};
	\draw node at (8.98,-3)  {$]$};	
	\draw[-] (9,-3) -- (7,-3) node[midway,above]{\footnotesize $c_2$};			
	\draw node[rotate=90,text width=1cm,align=center] at (10,-3.3) {\small Right\\ Closed};	
	\draw node at (0,-3.5) {$\cdot$};
	\draw node at (0.5,-3.5) {$\cdot$};
	\draw node at (1.7,-3.5) { $\cdot$};
	\draw node at (2,-3.5) {$\cdot$};
	\draw node at (3,-3.5) {$]($};
	\draw node at (5,-3.5) {$\cdot$};
	\draw node at (6,-3.5) {$\cdot$};
	\draw node at (7,-3.5) { $]($};
	\draw node at (8.5,-3.5) {$\cdot$};
	\draw node at (9,-3.5) { $]($};	
	
	\draw node at (-1,-5.5) {\small $\mathtt{edge}=$};	
	\draw node at (0,-5)  {$[$};
	\draw node at (2.98,-5)  {$)$};
	\draw[-] (3,-5) -- (0,-5) node[midway,above]{\footnotesize $c_0$};
	\draw node at (3.02,-5)  {$[$};
	\draw node at (6.98,-5)  {$)$};
	\draw[-] (7,-5) -- (3,-5) node[midway,above]{\footnotesize $c_1$};
	\draw node at (7.02,-5)  {$[$};
	\draw node at (8.98,-5)  {$)$};	
	\draw[-] (9,-5) -- (7,-5) node[midway,above]{\footnotesize $c_2$};			
	\draw node[rotate=90,text width=1cm,align=center] at (10,-5.3) {\small Left\\ Closed};	
	\draw node at (0,-5.5) {$\cdot$};
	\draw node at (0.5,-5.5) {$\cdot$};
	\draw node at (1.7,-5.5) {$\cdot$};
	\draw node at (2,-5.5) {$\cdot$};
	\draw node at (3,-5.5) {$)[$};
	\draw node at (5,-5.5) {$\cdot$};
	\draw node at (6,-5.5) {$\cdot$};
	\draw node at (7,-5.5) { $)[$};
	\draw node at (8.5,-5.5) {$\cdot$};
	\draw node at (9,-5.5) { $)[$};

	\draw node at (-1,-7.5) {\small $\mathtt{edge}=$};	
	\draw node at (0,-7)  {$[$};
	\draw node at (2.98,-7)  {$]$};
	\draw[-] (3,-7) -- (0,-7) node[midway,above]{\footnotesize $c_0$};
	\draw node at (3.02,-7)  {$($};
	\draw node at (6.98,-7)  {$)$};
	\draw[-] (7,-7) -- (3,-7) node[midway,above]{\footnotesize $c_1$};
	\draw node at (7.02,-7)  {$[$};
	\draw node at (8.98,-7)  {$]$};	
	\draw[-] (9,-7) -- (7,-7) node[midway,above]{\footnotesize $c_2$};			
	\draw node[rotate=90,text width=1cm,align=center] at (10,-7.3) {\small Unrestricted};	
	\draw node at (0,-7.5) {$\cdot$};
	\draw node at (0.5,-7.5) {$\cdot$};
	\draw node at (1.7,-7.5) {$\cdot$};
	\draw node at (2,-7.5) {$\cdot$};
	\draw node at (3,-7.5) {$]($};
	\draw node at (5,-7.5) {$\cdot$};
	\draw node at (6,-7.5) {$\cdot$};
	\draw node at (7,-7.5) { $)[$};
	\draw node at (8.5,-7.5) {$\cdot$};
	\draw node at (9,-7.5) { $]($};		
	\end{tikzpicture}
	\caption{
Illustration of the right-closed, left-open and unrestricted semantics on the trace presented in Fig.~\ref{fig:phi-N} generated by the TA
of Fig.~\ref{fig:TaWithVariableExample}.}
	\label{fig:differentsemantics}
\end{figure}
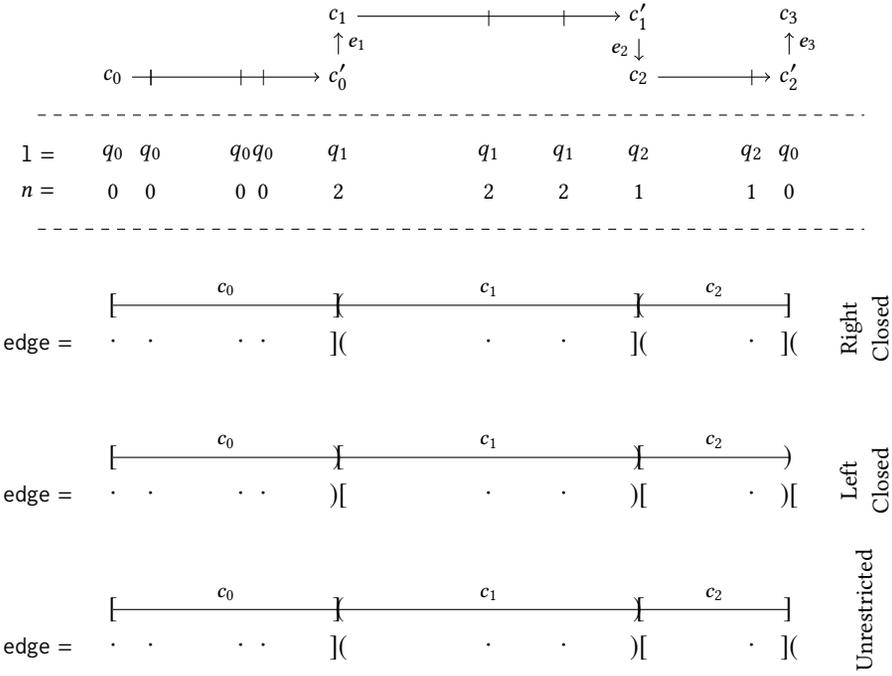

\subsection{Verification problem of networks of TA with respect to MITL formulae}
The properties of networks of TA are encoded by formulae that predicate over the values of the variables of set $Int$ and
over the atomic propositions which are labeling locations.
This section defines when a network $\netid$ of TA satisfies a MITL property $\psi$ and states the verification problem of networks of TA against a MITL formula.
Both definitions are based on the idea of selecting a (possibly proper) subset of traces of $\netid$, with respect to some selection criterion $T$.

Given a network $\netid$ of TA, a selection criterion $T$ for the traces of $\netid$ identifies a subset of traces of $\netid$.
An example of selection criterion $T$ could be ``the set of traces that correspond to right-closed signals''.
A trivial selection criterion simply identifies the set of all traces of $\netid$.
With a slight abuse of notation, in the following, $T$ indicates both the selection criterion and the set of traces that it identifies.

\begin{definition}[Satisfiability of MITL formulae over networks of TA]\label{def:MITLsatTA}
Let $\mathcal{N}$ be a network of TA, $T$ a selection criterion, and $\psi$ a MITL formula.
$\mathcal{N}$ satisfies $\psi$ restricted to  $T$ (written $\netid \models_T \psi$)
if every signal $M_\eta \in T$ is such that $M_\eta,0 \models \psi$ holds.
\end{definition}

\begin{definition}[Verification problem]\label{def:mc-problem}
Let $\mathcal{N}$ be a network of TA, $T$ be a selection criterion, and $\psi$ be a MITL formula.
The verification problem for the network of TA $\netid$ restricted to $T$ against a MITL formula $\psi$ consists in determining whether $\netid \models_T \psi$ holds.
\end{definition}

In the rest of this paper, the adopted selection criteria restrict the traces of interest to those that satisfy some liveness conditions $l$, the semantics of the synchronization primitives $s$ appearing in $\netid$, and some restriction $\mathit{ef}$ on the edges of the transitions taken.
Such a selection criterion is denoted as $T = \langle l, s, \mathit{ef} \rangle$.
In this case, the set of selected traces is $\mathcal{T}(\netid, l, s, \mathit{ef})$, and the corresponding signals are $\mathcal{S}(\netid, l, s, \mathit{ef})$.

In the following, the verification problem of Def.~\ref{def:mc-problem} is reduced to the problem of checking the satisfiability of \logic{} formula 
$\finalformula \land \Phi_{\neg\psi}$,
where $ \finalformula$ and $\Phi_{\neg\psi}$ are computed as specified in Sections~\ref{modelsignals} and ~\ref{MITLsignals}, respectively.
Section~\ref{signalcorrectness} shows the correctness of the proposed procedure.

\begin{figure}
\begin{tikzpicture}
	\draw node at (0,0) (a) {\small $c_0$};
	\draw node at (3,0) (b) {\small $c'_0$};
	\draw node at (3,0.8) (c) {\small $c_1$};
	\draw node at (7,0.8) (d) {\small $c'_1$};
	\draw node at (7,0) (e) {\small $c_2$};
	\draw node at (9,0) (f) {\small $c'_2$};
	\draw node at (9,0.8) (g) {\small $c_3$};
	\draw[->] (a) -- (b) node[midway,above]{};
	\draw[->] (b) -- (c) node[midway,right]{\footnotesize $e_1$};;
	\draw[->] (c) -- (d) node[midway,above]{};
	\draw[->] (d) -- (e) node[midway,left]{\footnotesize $e_2$};;
	\draw[->] (e) -- (f) node[midway,above]{};
	\draw[->] (f) -- (g) node[midway,right]{\footnotesize $e_3$};
	\draw node at (0.5,0)  {\tiny $|$};
	\draw node at (0.5,0)  {\tiny $|$};
	\draw node at (1.7,0)  {\tiny $|$};
	\draw node at (2,0)  {\tiny $|$};
	\draw node at (5,0.776)  {\tiny $|$};
	\draw node at (6,0.776)  {\tiny $|$};
	\draw node at (8.5,0)  {\tiny $|$};
	
	\pgfmathsetmacro{\ilocation}{-0.75}
	\draw[dashed] (-1,-0.5) -- (10,-0.5);
	\pgfmathsetmacro{\llocation}{-1}
	\pgfmathsetmacro{\dlocation}{-1.5}
	\pgfmathsetmacro{\xlocation}{-2.0}
	\pgfmathsetmacro{\edgelocation}{-2.0}
	\draw node at (-1,\llocation) {\small $\mathtt{l}=$};
	\draw node at (-1,\dlocation) {\small $\variable=$};
	\draw node at (0,\llocation) {\small $q_0$};
	\draw node at (0.5,\llocation) {\small $q_0$};
	\draw node at (1.7,\llocation) {\small $q_0$};
	\draw node at (2,\llocation) {\small $q_0$};
	\draw node at (3,\llocation) {\small $q_1$};
	\draw node at (5,\llocation) {\small $q_1$};
	\draw node at (6,\llocation) {\small $q_1$};
	\draw node at (7,\llocation) {\small $q_2$};
	\draw node at (8.5,\llocation) {\small $q_2$};
	\draw node at (9,\llocation) {\small $q_0$};
	\draw node at (0,\dlocation) {\small $0$};
	\draw node at (0.5,\dlocation) {\small $0$};
	\draw node at (1.7,\dlocation) {\small $0$};
	\draw node at (2,\dlocation) {\small $0$};
	\draw node at (3,\dlocation) {\small $2$};
	\draw node at (5,\dlocation) {\small $2$};
	\draw node at (6,\dlocation) {\small $2$};
	\draw node at (7,\dlocation) {\small $1$};
	\draw node at (8.5,\dlocation) {\small $1$};
	\draw node at (9,\dlocation) {\small $0$};
\draw node at (0,-2.5) {\small $\first{a}$};
\draw node at (0.5,-2.5) {\small $\first{a}$};
\draw node at (1.7,-2.5) {\small $\first{a}$};
\draw node at (2,-2.5) {\small $\first{a}$};
\draw node at (3,-2.5) {\small $\first{a}$};
\draw node at (0,-3.0) {\small $\rest{a}$};
\draw node at (0.5,-3.0) {\small $\rest{a}$};
\draw node at (1.7,-3.0) {\small $\rest{a}$};
\draw node at (2,-3.0) {\small $\rest{a}$};
\draw node at (9,-3.0) {\small $\rest{a}$};
\draw node at (8.5,-3.5) {\small $\first{c}$};
\draw node at (9,-3.5) {\small $\first{c}$};
\draw node at (7,-4) {\small $\rest{c}$};
\draw node at (8.5,-4) {\small $\rest{c}$};
\draw node at (5,-4.5) {\small $\first{\variable=2}$};
\draw node at (6,-4.5) {\small $\first{\variable=2}$};
\draw node at (3,-5) {\small $\rest{\variable=2}$};
\draw node at (5,-5) {\small $\rest{\variable=2}$};
\draw node at (6,-5) {\small $\rest{\variable=2}$};
\draw[thick]  (0,-6) -- (3,-6);
\draw[thick]  (3,-6.5) -- (9,-6.5);
\draw[thick]  (9,-6) -- (9.5,-6);
\draw[thick,dashed]  (9.5,-6) -- (10,-6);
\draw[fill] (0,-6) circle (0.05);
\draw[fill] (0.5,-6) circle (0.05);
\draw[fill] (1.7,-6) circle (0.05);
\draw[fill] (2,-6) circle (0.05);
\draw[fill] (3,-6) circle (0.05);
\draw[] (3,-6.5) circle (0.05);
\draw[fill] (5,-6.5) circle (0.05);
\draw[fill] (6,-6.5) circle (0.05);
\draw[fill] (7,-6.5) circle (0.05);
\draw[fill] (8.5,-6.5) circle (0.05);
\draw[fill] (9,-6.5) circle (0.05);
\draw[] (9,-6) circle (0.05);
\draw[thick]  (0,-7.5) -- (7,-7.5);
\draw[thick]  (7,-7) -- (9,-7);
\draw[thick]  (9,-7.5) -- (9.5,-7.5);
\draw[thick,dashed]  (9.5,-7.5) -- (10,-7.5);
\draw[fill] (0,-7.5) circle (0.05);
\draw[fill] (0.5,-7.5) circle (0.05);
\draw[fill] (1.7,-7.5) circle (0.05);
\draw[fill] (2,-7.5) circle (0.05);
\draw[fill] (3,-7.5) circle (0.05);
\draw[fill] (5,-7.5) circle (0.05);
\draw[fill] (6,-7.5) circle (0.05);
\draw[] (7,-7) circle (0.05);
\draw[fill] (7,-7.5) circle (0.05);
\draw[fill] (8.5,-7) circle (0.05);
\draw[fill] (9,-7) circle (0.05);
\draw[] (9,-7.5) circle (0.05);
\draw[thick]  (0,-8.5) -- (3,-8.5);
\draw[thick]  (3,-8) -- (7,-8);
\draw[thick]  (7,-8.5) -- (9.5,-8.5);
\draw[thick,dashed]  (9.5,-8.5) -- (10,-8.5);
\draw[fill] (0,-8.5) circle (0.05);
\draw[fill] (0.5,-8.5) circle (0.05);
\draw[fill] (1.7,-8.5) circle (0.05);
\draw[fill] (2,-8.5) circle (0.05);
\draw[] (3,-8) circle (0.05);
\draw[fill] (3,-8.5) circle (0.05);
\draw[fill] (5,-8) circle (0.05);
\draw[fill] (6,-8) circle (0.05);
\draw[fill] (7,-8) circle (0.05);
\draw[] (7,-8.5) circle (0.05);
\draw[fill] (8.5,-8.5) circle (0.05);
\draw[fill] (9,-8.5) circle (0.05);
\draw[dashed] (-1,-2.0) -- (10,-2.0);
\draw[dashed] (-1,-5.5) -- (10,-5.5);
\draw node at (0,-6.25) {\small $a$};
\draw node at (0.5,-6.25) {\small $a$};
\draw node at (1.7,-6.25) {\small $a$};
\draw node at (2,-6.25) {\small $a$};
\draw node at (3,-6.25) {\small $a$};
\draw node at (5,-6.25) {\small $\neg a$};
\draw node at (6,-6.25) {\small $\neg a$};
\draw node at (7,-6.25) {\small $\neg a$};
\draw node at (8.5,-6.25) {\small $\neg a$};
\draw node at (9,-6.25) {\small $\neg a$};
\draw node at (0,-7.25) {\small $\neg c$};
\draw node at (0.5,-7.25) {\small $\neg c$};
\draw node at (1.7,-7.25) {\small $\neg c$};
\draw node at (2,-7.25) {\small $\neg c$};
\draw node at (3,-7.25) {\small $\neg c$};
\draw node at (5,-7.25) {\small $\neg c$};
\draw node at (6,-7.25) {\small $\neg c$};
\draw node at (7,-7.25) {\small $\neg c$};
\draw node at (8.5,-7.25) {\small $c$};
\draw node at (9,-7.25) {\small $c$};
\draw node at (1.5,-8.25) {\footnotesize $\neg \variable=2$};
\draw node at (5,-8.25) {\footnotesize $\variable=2$};
\draw node at (8.5,-8.25) {\footnotesize $\neg \variable=2$};

\draw node at (-1,-9.5) {\small $\LTLf_{(0,1)}(c)$ };
\draw[thick]  (0,-9.5) -- (6,-9.5);
\draw[thick]  (6,-9) -- (9,-9);
\draw[thick]  (9,-9.5) -- (9.5,-9.5);
\draw[thick,dashed]  (9.5,-9.5) -- (10,-9.5);
\draw[fill] (0,-9.5) circle (0.05);
\draw[fill] (0.5,-9.5) circle (0.05);
\draw[fill] (1.7,-9.5) circle (0.05);
\draw[fill] (2,-9.5) circle (0.05);
\draw[fill] (3,-9.5) circle (0.05);
\draw[fill] (5,-9.5) circle (0.05);
\draw[fill] (6,-9.5) circle (0.05);
\draw[] (6,-9) circle (0.05);
\draw[fill] (7,-9) circle (0.05);
\draw[fill] (8.5,-9) circle (0.05);
\draw[fill] (9,-9.5) circle (0.05);
\draw[] (9,-9) circle (0.05);

\end{tikzpicture}
\caption{Relationship between a trace and the MITL signal derived by Formula $\varphi_\mathit{sig}$. }
\label{fig:traceandsignal}
\end{figure}
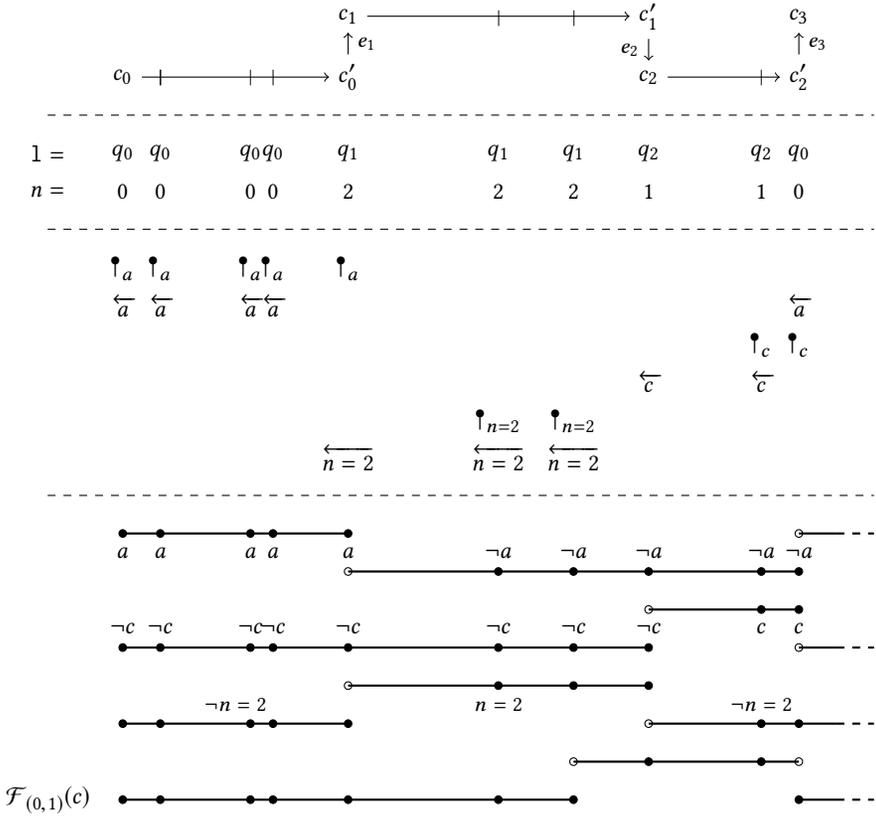

\subsection{\CLTLoc\ encoding of MITL signals}
\label{MITLsignals}
Bersani et al.~\cite{BRS15b} showed how to build a \CLTLoc\ formula $\Phi_{\psi}$ from a MITL formula $\psi$ such that the set $M_\psi$ of signals that are models of $\psi$ (i.e., $M_\psi = \{M | M, 0 \models \psi\}$)
is represented by the set of models of $\Phi_{\psi}$---hence, the satisfiability of $\psi$ is reduced to the satisfiability of $\Phi_{\psi}$.
Mapping a continuous-time signal $M$ to a denumerable sequence of elements is done by partitioning $\realNum_{\geq 0}$ into infinitely many bounded intervals, each one representing a portion of $M$ in which the values of propositions and integer variables do not change (except possibly in the endpoints).
In particular, let $I$ be an interval of the form $(a,b)$, with $a<b$, and $I_0, I_1, \dots$ be a denumerable set of adjacent intervals (i.e., $a_{i+1}=b_i$ holds for all $i>0$) covering $\realNum_{\geq 0}$---i.e., such that $\bigcup_{i\geq 0}(I_i \cup \{a_i\}) = \realNum_{\geq 0}$ holds, with $a_0=0$.
Every position $i$ in a \logic~model of $\Phi_\psi$ represents the ``configuration of $\psi$''---i.e., the value of all its subformulae---in interval $I_i$ and at instant $a_i$, according to the semantics of MITL.

\begin{remark}
As already remarked {in} Sec.~\ref{sec:tasemantics}, every sequence of consecutive time transitions in a trace can be replaced by an equivalent sequence of alternating time and discrete transitions, such that the total amount of elapsed time is the same as the original time transition and in all introduced discrete transitions every automaton does not perform any configuration change.
Every position of the models of $ \finalformula \land \Phi_{\neg\psi}$ represents a time instant where either the configuration of $\netid$, or the configuration of $\neg\psi$, changes, or both possibly change at the same time.
Therefore, in case $\neg\psi$ changes configuration at position $h$---i.e., one of its subformulae changes value---but $\netid$ does not, then $h$ in the trace of $\netid$ corresponds to a discrete transition $\Lambda_h$ such that $\Lambda_i[\autindex]=\_$, for all $1\leq \autindex\leq \numberOfTA$.
Figure~\ref{fig:traceandsignal},
which will be discussed more in depth in Example~\ref{ex:traceandsignal}, exemplifies this situation by showing a formula that changes value while the automaton does not take any transition.
\end{remark}

\subsection{\CLTLoc\ encoding of network signals}
\label{modelsignals}
Let $\netid$ be a network of TA,
$l$ be a set of liveness conditions  selected from Table~\ref{tab:livenesssemantic},
$s$ be the semantics of the synchronization primitives appearing in $\netid$ (selected from Table~\ref{tab:networksemantics}), and $\mathit{ef}$ be a restriction on the edges of the transitions taken by $\netid$;
this section defines the \logic\ formula $\finalformula$ representing the set of signals in $\mathcal{S}(\netid, l, s, \mathit{ef})$.
By Def.~\ref{def:signal}, every trace $\eta$ is associated with a signal $M_\eta$ that can be decomposed into the initial value $M_\eta(0)$, an infinite set of intervals $(\eventtime(\delta_0),\eventtime(\delta_1)),(\eventtime(\delta_1), \eventtime(\delta_2)),\ldots $,
defined by the time transitions $\delta_h$, and a set of time instants $\eventtime(\Lambda_h)$ corresponding to discrete transitions with symbol $\Lambda_h$, where $h\geq 0$.

According to \cite{BRS15b}, suitable \CLTLoc\ atoms can be used to represent the signal defined by the atomic propositions labeling the locations of automata and the arithmetical formulae occurring in a MITL formula.
In the following, $AF$ indicates the universe of propositions  of the form $n\sim d$, where $n$ is an integer variable and $d$ is a constant.
For every $\beta \in AP\cup AF$, the value of $\beta$ in the intervals $(\eventtime(\delta_h), \eventtime(\delta_{h+1}))$ is represented by proposition $\rest{\beta}$, called \emph{rest} of $\beta$; 
similarly, the value of $\beta$ in time instants $\eventtime(\Lambda_h)$ is represented by a proposition $\first{\beta}$, called \emph{first} of $\beta$.

Formula $\finalformula$ is built by combining formula $\Phi_\netid$ defined in Theorem~\ref{theorem:mapping2},
representing the traces $\eta$ of $\netid$, and a formula that constrains the propositions $\first{\beta}$ and $\rest{\beta}$, so that the signal $M_\eta$ is correctly defined and all the conditions in Def.~\ref{def:signal} are satisfied.
Even though a signal $M_\eta$ specifies a valuation $v_\mathit{var}$ in every time instant, {and it defines} the exact assignment for every variable $n \in Int$, formula $\finalformula$ 
only represents the signal of the formulae $n\sim d$ that appear in the MITL formula $\psi$, because the {truth of $\psi$} is determined only by the value of its subformulae.
The value of $\first{n\sim d}$ and $\rest{n\sim d}$, however, is defined in every time position $\timeposition$ of the model of $\finalformula$ by the value $\iota(i,n)$.

Formula $\finalformula$ is defined in~\eqref{eq:phi-N-signals}.
It is composed by two parts, 
formula $\Phi_{\mathcal{N}}$
and formula $\varphi_\mathit{sig}$, which maps traces to signals as defined in~\eqref{eq:phisign}.

\begin{table*}[t]
\centering
\caption{Formulae encoding the relation between the network of TA and the signal at the initial time instant and within the intervals.}
\label{tab:bindingCLTLoc}
\begin{tabular}{c  c  c  c }
\toprule
\multicolumn{2}{ c |}{ $\chi_1 \coloneqq \underset{
		\begin{subarray}{l} 
		\autindex \in [1,\numberOfTA],\\
		\proposition \in AP
		\end{subarray}}{\bigwedge}  \left( \first{\proposition} \leftrightarrow   
	\underset{
		\proposition \in L(q_{0,k})
	}{\bigvee} \loc[k]= q_{0,k} \right)$}  &
\multicolumn{2}{| c }{$\chi_2 \coloneqq \underset{(n \sim d) \in AF}{\bigwedge} \left( \first{n \sim d} \leftrightarrow n \sim d \right)$} \\
\midrule
\multicolumn{2}{c |}{$\chi_3 \coloneqq \LTLg \underset{\proposition \in AP}{\bigwedge} \left( \rest{\proposition} \leftrightarrow  \underset{
		\begin{subarray}{l}
		\autindex \in [1,\numberOfTA], 
		q \in Q_k,
		\proposition \in L(q)
		\end{subarray}
	}{\bigvee} \loc[\autindex]=q \right)$} & 
\multicolumn{2}{| c }{$\chi_4   \coloneqq  \LTLg \underset{(n \sim d) \in AF}{\bigwedge} (\rest{n \sim d} \leftrightarrow n \sim d )$}
\\
\midrule
\multicolumn{4}{c }{$\chi_5 \coloneqq 
	\LTLg \underset{\proposition \in AP}{\bigwedge} 
	\left( \LTLx (\first{p}) \leftrightarrow
\left(
\begin{gathered}
\underset{
		\begin{subarray}{l}
		\autindex \in [1,\numberOfTA], 
		q \in Q_k,
		\proposition \in L(q) 
		\end{subarray}
	}{\bigvee} \loc[\autindex]=q \land \left( \tr[\autindex] = \notr \vee (\tr[\autindex] \neq \notr \wedge  \rclosed{k}) \right)
\\ \lor \\
\underset{
		\begin{subarray}{l}
		\autindex \in [1,\numberOfTA], 
		q \in Q_k,
		\proposition \in L(q) 
		\end{subarray}
	}{\bigvee} 		\LTLx (\loc[\autindex]=q	)   \wedge  \tr[\autindex] \neq \notr \wedge  \neg \rclosed{k}\\
\end{gathered}	
	\right) \right)
	$
}\\		
\midrule
\multicolumn{4}{ c }{$\chi_6 \coloneqq 
	\LTLg \underset{(n\sim d) \in AF}{\bigwedge} \left( \LTLx (\first{n\sim d})  \leftrightarrow
 \left( 
 \begin{gathered}
 {n \sim d} \land
 \left(
 \begin{gathered}
	\neg \underset{
				\begin{subarray}{l}
         		\autindex \in [1,\numberOfTA],
				t \in T_k, n  \in U(t)
				\end{subarray}
			}{\bigvee} 	  
			 \tr[\autindex] = t
\\ \lor \\
 \underset{
		\begin{subarray}{l}
		\autindex \in [1,\numberOfTA], 
		t \in T_k, n  \in U(t)
		\end{subarray}
	}{\bigvee} 	  \tr[\autindex] = t \wedge  \rclosed{k}
  \end{gathered}
  \right)
\\ \lor \\
 \LTLx ( {n \sim d} )
 \land
 \underset{
		\begin{subarray}{l}
		\autindex \in [1,\numberOfTA], 
		t \in T_k, n  \in U(t)
		\end{subarray}
	}{\bigvee} 	  \tr[\autindex] = t \wedge \neg \rclosed{k} \\ 
\end{gathered}	
	\right)
	\right)
	$
}\\		
\bottomrule
\end{tabular}
\end{table*}

\begin{align}\label{eq:phi-N-signals}
\finalformula := 
 \mathrlap{\overbrace{\phantom{( \varphi_\mathcal{N}   \land \varphi_\mathit{l}  \land \varphi_\mathit{s}\land \varphi_{\mathit{ef}} )}}^{\Phi_{\mathcal{N}}}}
\varphi_\mathcal{N}  \land \varphi_\mathit{l} \land  \varphi_\mathit{s}  \land \varphi_{\mathit{ef}} \land  \varphi_\mathit{sig}  
\end{align}
Recall that, as defined in Formula~\eqref{f:network}, $\varphi_\mathcal{N}$ encodes the network of TA in \logic{}; in addition, $\varphi_\mathit{l}$, $\varphi_\mathit{s}$  and $\varphi_{\mathit{ef}}$ impose  the liveness conditions, synchronization mechanisms and restrictions on edges
 by means of the formulae in Tables~\ref{tab:livenessCLTLOoc},~\ref{tab:synchronizationCLTLocFormulae} and~\ref{tab:edges}.
Finally, formula $\varphi_\mathit{sig}$ is defined as follows, through the formulae of Table~\ref{tab:bindingCLTLoc}.
\begin{align}
\label{eq:phisign}
\varphi_\mathit{sig} :=  \bigwedge_{i \in [1,6]} \chi_i
\end{align}
Formulae $\chi_1$--$\chi_6$ create a mapping between the values of the atomic propositions and of the variables and the corresponding signals.
More precisely, formulae $\chi_3$ and $\chi_4$ (resp., $\chi_1$ and $\chi_2$) bind the values of the atomic propositions and of the variables to the corresponding signal within each time interval $(\eventtime(\delta_h), \eventtime(\delta_{h+1}))$ (resp., in the origin).
Formulae $\chi_5$ and $\chi_6$, instead, bind the  values of the atomic propositions and of the variables to the corresponding signals at the boundaries of the intervals---i.e., at time instants $\eventtime(\Lambda_h)$.

\begin{lemma}
\label{lemma:signalmapping}
Let $\netid$ be a network of TA,
$l$ be a set of liveness conditions  selected from Table~\ref{tab:livenesssemantic},
$s$ be the semantics of the synchronization primitives appearing in $\netid$ (selected from Table~\ref{tab:networksemantics}),
$\mathit{ef}$ be a restriction on the edges of the transitions taken by $\netid$ (from Table~\ref{tab:signalsemantics}),
and $\finalformula$ be the corresponding CLTLoc formula~\eqref{eq:phi-N-signals}.

For every edge-consistent trace $\tracesymbol$ of $\netid$ that also belongs to $\mathcal{T}(\netid, l, s, \mathit{ef})$,
and whose associated signal is \signalsymbol, there exists a model $(\pi,\sigma,\iota)$ of $\finalformula $ such that: 
\begin{enumerate}
\item
\label{cond:rest}
for every time instant $r \in \realNum_{\geq 0}$ such that $\eventtime(\delta_h) < r < \eventtime(\delta_{h+1})$, for some $h \in \naturalNum$, if $M_{\tracesymbol}(r)=(P,\variablevaluationfunction)$,
the following conditions hold:
\begin{align}
p \in P \ \ &\text{ iff } \ \ (\pi,\sigma,\iota),h \models \rest{p}\label{eq:cltloctosignal-ap} \\
v_{var}(n)\sim d \ \ &\text{ iff } \ \ (\pi,\sigma,\iota),h \models \rest{n\sim d}\label{eq:cltloctosignal-af}
\end{align}
\item
\label{cond:first}
for every time instant $r \in \realNum_{\geq 0}$ such that $r = \eventtime(\Lambda_h)$, for some $h \in \naturalNum$, if $M_{\tracesymbol}(r)=(P,\variablevaluationfunction)$,
the following conditions hold:
\begin{align}
p \in P \ \ &\text{ iff } \ \ (\pi,\sigma,\iota),h+1 \models \first{p}\label{eq:cltloctosignal-ap-first} \\
\variablevaluationfunction(n)\sim d \ \ &\text{ iff } \ \ (\pi,\sigma,\iota),h+1 \models \first{n\sim d}\label{eq:cltloctosignal-af-first}
\end{align}
\end{enumerate}
Conversely, for every model $(\pi,\sigma,\iota)$ of $\finalformula$, there exists an edge-consistent trace $\tracesymbol$ of $\netid$ that belongs to set $\mathcal{T}(\netid, l, s, \mathit{ef})$, with associated signal \signalsymbol, for which conditions \ref{cond:rest} and \ref{cond:first} hold.
\end{lemma}

\begin{proof-sketch}
Let $\tracesymbol$ be an edge-consistent trace of $\netid$ that also belongs to set $\mathcal{T}(\netid, l, s, \mathit{ef})$.
By Theorem \ref{theorem:mapping2}, every $(\pi,\sigma,\iota)$ such that $(\pi,\sigma,\iota) \in \map(\tracesymbol)$ is a model of $\Phi_\netid$.
Formula $\Phi_\netid$ does not constrain propositions $\first{\beta}$ and $\rest{\beta}$, with $\beta \in AP\cup AF$.
Hence, it has to be proven that, if $(\pi,\sigma,\iota)$, and $M_\tracesymbol$ are also such that conditions \eqref{eq:cltloctosignal-ap}-\eqref{eq:cltloctosignal-af-first} hold, 
then $(\pi,\sigma,\iota)$ is a model of $\finalformula$.
Since $(\pi,\sigma,\iota)$ is a model for $\Phi_\netid$ (Thm.~\ref{theorem:mapping}), it is enough to show that $(\pi,\sigma,\iota)$ is a model also for $\varphi_\mathit{sig}$.
By definition, $\tracesymbol$ meets the conditions of Def.~\ref{def:signal}.

It is straightforward to show that subformulae $\chi_1$ and $\chi_2$ of $\varphi_\mathit{sig}$ hold because of condition \ref{signal-time-0} of Def.~\ref{def:signal}, since they state that in the origin of the signal predicates $\first{\beta}$ and $\rest{\beta}$, for $\beta \in AP\cup AF$, correspond to the initial configuration of $\netid$.
Similarly, subformulae $\chi_3$ and $\chi_4$ hold because of condition \ref{signal-time-t}, since they capture the fact that, in each interval $(\eventtime(\delta_{h}), \eventtime(\delta_{h+1}))$, the predicates that hold are those of position $h$ of $(\pi,\sigma,\iota)$, which derives, by mapping $\map$, from configuration $(\loc_h, \variablevaluationfunctionpar{h}, \clockvaluationfunctionpar{h})$.

Consider now formula $\chi_5$.
The first disjunct of the right-hand side states that the label $\proposition$ holds at the beginning 
of an interval $I_{h+1}$---i.e., at time instant $\eventtime(\Lambda_{h})$, for $h \geq 0$---if it held in the previous interval $I_h$ f
or an automaton $\mathcal{A}_{\autindex}$, and either $\mathcal{A}_{\autindex}$ does not take any transition (i.e., $\Lambda_{h}[\autindex] = \_$) or, if takes one transition, it does so with an $]($ edge (i.e., $\Lambda_{h}[\autindex] = \action^{](}\}$); this corresponds to the first bullet of condition \ref{signal-discrete-t-atoms} of Def.~\ref{def:signal}.
The second disjunct, instead, states that $\proposition$ holds at the beginning of an interval $I_{h+1}$ if there is an automaton $\mathcal{A}_{\autindex}$ that takes a transition, and it does so with an $)[$ edge (i.e., $\Lambda_{h}[\autindex] = \action^{)[}\}$), which corresponds to the second bullet of condition \ref{signal-discrete-t-atoms} of Def.~\ref{def:signal}.
Similarly, the first disjunct of the right-hand side of formula $\chi_6$ captures condition \ref{signal-discrete-t-var-I} of Def.~\ref{def:signal} (each disjunct in the subformula corresponds to one of the bullets of condition \ref{signal-discrete-t-var-I}), while the second disjunct captures condition \ref{signal-discrete-t-var-II}.

The second part of the statement is proven by showing that, given a model $(\pi,\sigma,\iota)$ of $\finalformula$, the corresponding trace $\tracesymbol = \map^{-1}((\pi,\sigma,\iota))$ is such that conditions \eqref{eq:cltloctosignal-ap}-\eqref{eq:cltloctosignal-af-first} hold for signal $M_\tracesymbol$.
This can be done using similar arguments as those presented in the first part of the proof, and is omitted for brevity.
\end{proof-sketch}

\subsection{Model-checking of networks of TA with respect to MITL formulae}
\label{signalcorrectness}

Lemma \ref{lemma:signalmapping} establishes a correspondence between signals derived from traces of network $\netid$ and models of formula $\finalformula$.
The models of formula $\finalformula$ include predicates of the type $\first{\beta}$ and $\rest{\beta}$, which act as a ``bridge'' with the encoding of MITL formulae that predicate over $\beta$.
The next example shows how this allows us to match MITL constraints with signals derived from networks of TA.

\begin{example}
\label{ex:traceandsignal}
Figure~\ref{fig:traceandsignal} shows the relation between the trace of Example \ref{ex:cltloc-model-trace} depicted in Fig.~\ref{fig:phi-N} and the MITL signals referring to the atomic propositions  $a$, $c$ and the subformulae $\variable = 2$ and  $\LTLf_{(0,1)} (c)$. 
It shows the assignments to $\loc$ and $\variable$, and  the atoms representing the signals of $a$, $c$ and $\variable = 2$ in correspondence to the positions where their value is true.
For instance, at position 4, $\first{a}$ and $\rest{\variable=2}$ hold, whereas $\first{c}$, $\rest{c}$, $\rest{a}$ and $\first{\variable=2}$ are false.
The  bottom part of the figure shows the signals that are built according to the value of \logic~atoms $\first{\beta}$ and $\rest{\beta}$.
The signal of each proposition is drawn on two levels: the top one represents the value true and the bottom one represents the value false.
In every position, the value of the proposition is specified by a filled circle that defines the value in the exact time instant corresponding to the position and every line between adjacent positions represents the value of the proposition in the corresponding interval. 
An empty circle at the beginning or at the end of an interval indicates that the value of the formula in the interval does not extend also to the infimum or to the supremum of the interval, respectively.
At position $6$
formula $\LTLf_{(0,1)} (c)$ changes value, because one time unit later (which corresponds to position $7$ in this example) formula $c$ becomes true; the automaton, instead, at position $6$ has the same configuration---if clock assignments are not considered---as at position $5$.

\end{example}

The next proposition shows how, given a network $\netid$ of TA, a selection criterion $T$, and a MITL property $\psi$, the problem of checking whether $\netid \models_T \psi$ holds can be reduced to that of determining the satisfiability of \logic{} formula $\finalformula \land \Phi_{\neg\psi}$.

\begin{txproposition}
\label{prop:mcsolving}
Let $\mathcal{N}$ be a network of TA, $\psi$ be a MITL formula, and $T = \langle l, s, \mathit{ef} \rangle$---where
$l$ is a set of liveness conditions  selected from Table~\ref{tab:livenesssemantic},
$s$ is the semantics of the synchronizations primitives appearing in $\netid$ (selected from Table~\ref{tab:networksemantics}),
and $\mathit{ef}$ is a restriction on the edges of the transitions taken by $\netid$ (from Table~\ref{tab:signalsemantics}).
Also, let $\finalformula$ and $\Phi_{\psi}$ be the \logic\ formulae defined in Section~\ref{modelsignals} and in Section~\ref{MITLsignals}, respectively. Then,
$\netid \models_T \psi$ holds if, and only if, $\finalformula \land \Phi_{\neg\psi}$ does not have any models.
\end{txproposition}

\begin{proof-sketch}
By Lemma \ref{lemma:signalmapping} and by the results of \cite{BRS15b}, $\finalformula \land \Phi_{\neg\psi}$ admits a model if, and only if, there is $(\pi, \sigma, \iota)$ that corresponds to a signal $M_\tracesymbol$ that satisfies MITL formula $\neg \psi$, and such that trace $\tracesymbol$ belongs to $\mathcal{T}(\netid, l, s, \mathit{ef})$.
That is, $\finalformula \land \Phi_{\neg\psi}$ does not have any models if, and only if, $\mathcal{S}(\netid, l, s, \mathit{ef}) \cap M_{\neg\psi} = \emptyset$ holds.
This, in turn, is equivalent to saying that $\mathcal{S}(\netid, l, s, \mathit{ef}) \subseteq M_{\psi}$ holds, which corresponds to Def.~\ref{def:MITLsatTA}.
\end{proof-sketch}

Since there are automated tools for checking the satisfiability of \logic{} formulae \cite{BRS15b, BPR16}, Proposition~\ref{prop:mcsolving} establishes an effective technique to solve the verification problem of networks of TA with respect to MITL formulae: given a network $\netid$ of TA, a selection criterion $T = \langle l, s, \mathit{ef} \rangle$, and a MITL formula $\psi$, it is enough to build \logic{} formulae $\finalformula$ and $\Phi_{\neg\psi}$, then check the satisfiability of formula $\finalformula \land \Phi_{\neg\psi}$.

\section{Evaluation}
\label{sec:evaluation}
\newcommand{\livewaitone}{ \mathbf{live\text{-}one} }
\newcommand{\livewaittwo}{ \mathbf{live\text{-}two} }
\newcommand{\livewaitthree}{ \mathbf{live\text{-}three} }
\newcommand{\livewaitfour}{ \mathbf{live\text{-}four} }
\newcommand{\livewaitfive}{ \mathbf{live\text{-}five} }
\newcommand{\livewaitsix}{ \mathbf{live\text{-}six} }

\newcommand{\livetoken}{ \mathbf{live\text{-}token} }

\newcommand{\livecsmacd}{ \mathbf{live\text{-}csma} }

\newcommand{\livewaitall}{ \mathbf{live\text{-}all} }
\newcommand{\mlivewaitone}{ \mathbf{mlive\text{-}one} }
\newcommand{\mlivewaitall}{ \mathbf{mlive\text{-}w\text{-}all} }
\newcommand{\nolivecs}{ \mathbf{nolive\text{-}cs} }

\newcommand{\simplenolivecs}{ \mathbf{s\text{-}nolive\text{-}cs} }
\newcommand{\mutex}{ \mathbf{mutex} }

The procedure proposed in Section \ref{signalcorrectness} has been implemented in \NAME\ (Timed Automata ChecKer)\footnote{The tool is available at~\url{http://github.com/claudiomenghi/TACK}.}.
\NAME\ is a Java 8 application that takes as input a model expressed using the Uppaal input format 
 and a property expressed in MITL.
The model and the property are converted in a \logic\ formula as specified in Sect.~\ref{sec:checkingMITLI}.
The satisfiability of the \logic\ formula is verified using the \Zot\ formal verification tool~\cite{BPR16}.

To evaluate \NAME , a direct comparison with existing tools, i.e., Uppaal and \roland~\cite{kindermann2013bounded}, was not performed as such comparison would not be meaningful, for several  reasons.

(i)  Neither  Uppaal, nor \roland\ fully support MITL.
Uppaal supports a restricted subset of the TCTL logic, which allows the specification only of properties in the form:
$\forall \LTLg(e)$ (``for all executions $e$ globally holds");
$\forall \LTLf(e)$ (``for all executions $e$ eventually holds"); 
$\exists \LTLg(e)$ (``there exists an executions in which $e$ globally holds");
$\exists \LTLf(e)$ (``there exists an executions in which $e$ eventually holds")
and, finally, the so called ``leads-to'' formula which is encoded as $\forall\LTLg (e \Rightarrow \forall\LTLf (e'))$ (``in every execution it is always true that the occurrence of $e$ always makes $e'$ hold''), where $e$ and $e'$ are state formulae (i.e., expressions over state variables or automata locations).
\roland, instead, considers the  fragment MITL$_{0,\infty}$, but it does not provide any information on how the encoding can be extended to cover full MITL.
In fact, \roland\ adopts a super-dense semantics for time, requiring a suitable change of MITL semantics.
Indeed, some useful properties~\cite{HR04}\cite{MNP06} that hold for the standard MITL semantics, e.g., proving that the fragment MITL$_{0,\infty}$ has the same expressive power as full MITL, are not valid anymore over super-dense time.
Therefore, extending the logical language used in \roland\ to full MITL appears to be far from straightforward.

(ii) As mentioned above, both  Uppaal and \roland~\cite{kindermann2013bounded} adopt the super dense semantics of time.
This allows a TA to fire consecutive transitions without requiring time to progress.
Uppaal introduces the syntactic notion of ``committed locations'' to prevent time from progressing---i.e., when an automaton is in a committed location, only action transitions can be fired and time cannot advance.
The underlying notion of time adopted in this work is based on the CLTLoc semantics, which relies on the strict progress of time between adjacent positions and does not enable such modeling facility.

(iii) 
\roland\ is mainly a proof-of-concept tool that has not been further supported since 2013 and
does not support a direct implementation of synchronization events ($\alpha!$ and $\alpha?$).
The lack of suitable documentation (such as a user manual) does not allow a clear understanding of the potential offered by the tool.

To summarize, since the capabilities of Uppaal and \roland\  are significantly different from those provided by \NAME\ (especially in terms of the logic used to express the property), a direct comparison of \NAME\ with these tools is not significant.

The following features are here evaluated: 
(i) the efficiency  of \NAME\ in verifying MITL properties of TA;
(ii) how \NAME\ enables the introduction of the synchronization constructs and semantic constraints presented in Section~\ref{sec:tasemantics}. 
The ease of performing verification has been estimated
through a bounded model checking technique that relies on 
two different  solvers
available in the \Zot\ formal verification tool~\cite{BPR16}. 
Both solvers check the satisfiability of \logic\ formulae, but they are based on different techniques.
They rely on SMT (Satisfiability Modulo Theories) solvers (Microsoft Z3~\cite{de2008z3} in our case),
as the satisfiability problem of \logic\ has been tackled so far by reducing it to an SMT instance.
The first solver, \aetwozot\ ~\cite{bersani2016tool}, reduces the satisfiability problem of \logic\ formulae to that of a fragment of the first-order logic over real difference arithmetic; the second, \aetwosbvzot ~\cite{BPR16}, instead uses a Bit-Vector encoding.
The ability of \NAME\ to consider different features of TA is evaluated by selecting benchmarks that exploit different constructs such as, for example, different synchronization primitives.

Three different benchmarks are used in the tests: the Fischer mutual exclusion protocol \cite{abadi1994old}, the CSMA/CD protocol \cite{CSMACD} and the Token Ring protocol \cite{jain1994fddi}. All selected benchmarks have also been classically implemented in the Uppaal model checker~\cite{UPPAALWEB}. 
\NAME\ ran using the  \aetwozot\ and \aetwosbvzot\ solvers, version $4.7.1$ of Z3, on a machine equipped with an 
Intel(R) Core(TM) i7-4770 CPU (3.40GHz) with 8 cores, 16GB of RAM and Debian Linux (version 8.8). 
To test the scalability of the approach,
various configuration of the protocols were tested, each one determined by a different number $n$ of involved agents.
In particular, variable $n$ indicates: for the Fischer benchmark, the number of participants; 
for the CSMA/CD protocol, the number of competing stations;
for the Token Ring protocol, the number of processes.
In the tests, $n$ spans from $2$ to $10$.
Since the \logic{} solvers we used  relied on a bounded model-checking approach, we considered a bound $k$ spanning from $10$ to $30$, with increments of $5$.
For each combination of values of $n$ and $k$ we considered
a timeout of $2$ hours.

\vspace{0.2cm}
\emph{Fischer benchmark}.  
This benchmark describes  a mutual exclusion algorithm in which $n$ participants try to  enter a critical section.
Before trying to enter the critical section, a participant first checks if another one is in the critical section.
If this is not the case, it writes its (unique) identifier in a shared variable.
After waiting a certain amount of time, it checks again the shared variable.
If its identifier is still in the shared variable, it proceeds to the critical section. 
Otherwise,  it goes back to start since another process had simultaneously checked whether the critical section was empty and set the shared variable.
The synchronization among the participants is obtained through shared clocks and no synchronization on the transitions is present.

The following six properties (a subset of them was also considered in~\cite{kindermann2013bounded}) were verified.

	\begin{align}
	&\livewaitone &&:=&& \LTLg_{[0,\infty)} \left(p_1.req \LTLimplication \LTLf_{[0,\infty)} p_1.wait \right) &\nonumber\\
	&\livewaittwo &&:=&&  \LTLg_{[0,\infty)} \left(p_1.req   \LTLimplication \LTLf_{[0,3]} p_1.wait \right) & \nonumber\\
	&\livewaitthree &&:=& &\LTLg_{[0,\infty)} \left(p_1.req   \LTLimplication \LTLf_{(0,3)} p1.cs \right)	& \nonumber\\
	&\livewaitfour &&:=& &\LTLg_{[0,\infty)} \left(p_1.req \LTLimplication \LTLf_{(0,3)}  p_1.wait \right)&\nonumber\\
	&\livewaitfive &&:=& &\LTLg_{[0,\infty)} \left(p_1.req   \LTLimplication \LTLf_{[0,3]} p1.cs \right)	& \nonumber\\
	&\livewaitsix &&:=& &  \LTLg_{[0,\infty)} \left( \neg \left( \bigvee_{i=1:n-1} \left(p_i.cs \wedge \left(\bigvee_{j=i+1:n}\right) p_j.cs \right) \right)		\right)	& \nonumber
\end{align}

Properties $\livewaitone$ and $\livewaitsix$ are not metric, as $\LTLf_{[0,\infty)}$ and $\LTLg_{[0,\infty)}$ are equivalent to the LTL ``eventually'' and ``globally'' modalities.
Properties $\livewaittwo$ (resp., $\livewaitthree$)  and $\livewaitfour$ (resp., $\livewaitfive$) differ with respect to the interval specified in the  $\LTLf$ operator.

\begin{table}[htbp]
\footnotesize
\setstretch{0.6}
\center
\caption {Time (s) required to check the properties of the Fischer benchmark.
The symbol \cmark\ indicates that the property is satisfied, i.e., the \logic\	formula is unsatisfiable.
The symbol \xmark\ indicates that the property is not satisfied, i.e., the \logic\	formula is satisfiable.
}
\label{fig:scalabilityresultsFisher}
\begin{tabular}{ 
r 
r  r  r  
r  r  r  
r  r  r 
r  r  
}
\toprule
  \multicolumn{2}{c}{} &   \multicolumn{9}{c}{\NAME\ ae2zot}  \\  
  \midrule
  \multicolumn{2}{c}{}  &    \multicolumn{9}{c}{n}  \\
  \midrule
\multicolumn{1}{c}{}  & k &    \multicolumn{1}{c}{\textbf{2}} & \multicolumn{1}{c}{\textbf{3}} & \multicolumn{1}{c}{\textbf{4}} & \multicolumn{1}{c}{\textbf{5}} & \multicolumn{1}{c}{\textbf{6}} & \multicolumn{1}{c}{\textbf{7}} & \multicolumn{1}{c}{\textbf{8}} & \multicolumn{1}{c}{\textbf{9}} & \multicolumn{1}{c}{\textbf{10}}  \\
\toprule
\multirow{5}{*}{\rotatebox[origin=c]{90}{$\livewaitone$}} & 10  & 
 0.9 \cmark  &   1.2 \cmark &  0.9 \cmark  &  1.0 \cmark  &  1.3 \cmark  &  3.3 \cmark &  1.9 \cmark &  2.2 \cmark &  2.6 \cmark     \\
    \multicolumn{1}{c}{} &  15 &  0.8 \cmark &  3.1 \cmark &  4.9 \cmark & 	  1.7 \cmark &  4.8 \cmark &  2.6 \cmark &   4.7 \cmark &   14.1 \cmark &  6.0 \cmark  \\
    \multicolumn{1}{c}{} & 20  &  1.0 \cmark &  1.2 \cmark &  1.5 \cmark & 	  2.5 \cmark &  4.3 \cmark &  5.1 \cmark &   5.8 \cmark &  65.3 \cmark &  12.8 \cmark \\
    \multicolumn{1}{c}{} & 25 &  1.0 \cmark &  1.9 \cmark &  2.7 \cmark & 4.3 \cmark &  5.4 \cmark &  10.1 \cmark &   15.3 \cmark &  18.6 \cmark &  24.0 \cmark \\
     \multicolumn{1}{c}{} &  30 &  1.3 \cmark &  37.5 \cmark &  3.1 \cmark & 	  75.8 \cmark &  11.2 \cmark &  15.9 \cmark  &   20.4 \cmark &  30.1 \cmark &  3331.1  \cmark \\
  \cline{1-11}
   \hline
  \multirow{5}{*}{\rotatebox[origin=c]{90}{$\livewaittwo$}}  & 10 &    1.7  \cmark &   2.4 \cmark &   2.5  \cmark &  2.8 \cmark  &   4.6 \cmark &   4.6 \cmark &   5.5 \cmark  &   5.0 \cmark &   6.2  \cmark    \\
    \multicolumn{1}{c}{} & 15 &  7.9 \cmark &  18.3 \cmark &  28.0 \cmark  &  40.0 \cmark  &  70.0 \cmark &  131.2 \cmark &  211.9 \cmark &  350.6 \cmark &   417.0 \cmark \\
    \multicolumn{1}{c}{} & 20 &  42.0 \cmark &  96.4 \cmark &  139.9 \cmark &  514.7 \cmark &  897.0 \cmark &  1395.4 \cmark &  5837.1 \cmark &  - & -  \\
    \multicolumn{1}{c}{} & 25 &  208.8 \cmark  &  859.9 \cmark  &  813.1 \cmark   &  2026.9 \cmark  &  6770.0 \cmark   & - & - & - & -  \\
     \multicolumn{1}{c}{} & 30 &  392.6 \cmark &  793.5 \cmark &  2678.2 \cmark &  4193.3 \cmark & -  & - & - & - & -  \\
    \cline{1-11}
      \multirow{5}{*}{\rotatebox[origin=c]{90}{$\livewaitthree$}}  & 10 & 
       1.0 \xmark &  1.0 \xmark &  1.4 \xmark &  1.9\xmark  &  3.5 \xmark &   15.7 \cmark &   15.8 \cmark  &   19.2 \cmark  &   21.5 \cmark      \\
    \multicolumn{1}{c}{} &15  &   1.2 \xmark &   1.4 \xmark &   2.3 \xmark  &   5.3 \xmark &   5.9 \xmark &   4.3 \xmark &   13.2 \xmark &   27.6 \xmark &   47.3 \xmark  \\
    \multicolumn{1}{c}{} & 20 &  1.3  \xmark &   1.7 \xmark &   3.7 \xmark &   3.6 \xmark &   8.2 \xmark &   15.3 \xmark &   19.4 \xmark &  72.7  \xmark  &   46.1 \xmark \\
    \multicolumn{1}{c}{} & 25 &   2.0\xmark &   2.9 \xmark  &   4.1 \xmark &   9.6 \xmark &   13.7 \xmark &   41.8 \xmark &   41.8 \xmark &  97.6  \xmark &    97.7 \xmark \\
     \multicolumn{1}{c}{} & 30 &   2.2 \xmark &   3.0\xmark   &   7.4 \xmark &  12.4  \xmark  &   24.4 \xmark  &   41.0 \xmark &  68.3  \xmark &   35.3  \xmark &   342.2 \xmark  \\
    \cline{1-11}
      \multirow{5}{*}{\rotatebox[origin=c]{90}{$\livewaitfour$}} & 10 &   1.3 \cmark &  1.8 \cmark &  2.1 \cmark &   2.5 \cmark &  2.9 \cmark &  3.4 \cmark &  3.9 \cmark &  5.5 \cmark  &  5.8 \cmark      \\
    \multicolumn{1}{c}{} & 15  &  4.7 \cmark &  7.2 \cmark &  14.2 \cmark &  22.6 \cmark &  46.0 \cmark &  93.3 \cmark &  146.8 \cmark & 239.2 \cmark &  462.3 \cmark \\
    \multicolumn{1}{c}{} &  20 &  10.3 \cmark &  24.3 \cmark &  54.4 \cmark &  98.5 \cmark &  199.8 \cmark &  797.2 \cmark &  2526.5 \cmark & - & - \\
    \multicolumn{1}{c}{} & 25 &  48.5 \cmark &   112.9 \cmark &  190.9 \cmark &  779.3 \cmark &  1799.0 \cmark &  5437.316 \cmark & - & - & - \\
     \multicolumn{1}{c}{} & 30 &  152.9 \cmark &  255.1 \cmark &  675.5 \cmark & - & - & - & - & -  & -  \\
    \cline{1-11}
      \multirow{5}{*}{\rotatebox[origin=c]{90}{$\livewaitfive$}}  & 
      10  &  1.1 \xmark &  1.7 \xmark  &  1.9 \xmark &  4.1 \xmark &  7.7 \cmark &  27.5 \cmark &  31.0 \cmark &  39.7 \cmark &   47.6 \cmark   \\
    \multicolumn{1}{c}{} & 15 & 1.5 \xmark  &  1.8 \xmark &  3.8 \xmark &  9.1 \xmark &  30.5 \xmark &  73.2 \xmark &  42.2 \xmark &  211.7 \xmark &  544.6  \xmark \\
    \multicolumn{1}{c}{} & 20 &  2.8 \xmark &  3.2 \xmark &  6.5 \xmark &  16.9 \xmark &  50.7 \xmark  &   37.8 \xmark &  373.3 \xmark &  414.3 \xmark &  194.9 \xmark \\
    \multicolumn{1}{c}{} & 25 &  3.4 \xmark & 8.4  \xmark &  14.6 \xmark &  27.0 \xmark &  28.2 \xmark &  54.5 \xmark & 703.0 \xmark &  661.7 \xmark &  1853.5 \xmark   \\
     \multicolumn{1}{c}{} & 30 & 5.4 \xmark  &  7.5 \xmark & 15.5 \xmark   & 22.5  \xmark &  48.5 \xmark &  88.0 \xmark &  257.3 \xmark &  361.4 \xmark &   2756.2 \xmark \\
    \cline{1-11}
      \multirow{5}{*}{\rotatebox[origin=c]{90}{$\livewaitsix$}}  & 10  &     
       0.9 \cmark
      	& 
      	1.7 \cmark 
      		& 
      		 3.0 \cmark
      	  	& 
      	  	 4.0 \cmark
      	  		 & 
      	  		  8.6 \cmark
      	  		 & 
      	  		 16.7 \cmark 
      	  		 & 
      	  		  23.0 \cmark
      	  		 	& 
      	  		 	31.6 \cmark 
      	  		 	&   54.3 \cmark  \\
      \multicolumn{1}{c}{} & 15  &  
       1.2 \cmark
      	& 
      	 3.6 \cmark &
      		10.9 \cmark
      		& 	 
      		 20.5 \cmark &
      			61.8 \cmark	&
      			149.7 \cmark &
      			303.2 \cmark 		&
      		  646.5 \cmark &
      			1016.7 \cmark \\
    \multicolumn{1}{c}{} & 20 &  
     2.2 \cmark	&  
     11.6 \cmark &
     46.9 \cmark	& 
     140.5 \cmark	 & 
    	598.6 \cmark 	&	
    	963.3 \cmark &	
     2210.2 \cmark  &
     3800.0 \cmark	 &
     5849.1 \cmark  \\
    \multicolumn{1}{c}{} & 25  & 
     2.4 \cmark	&	
    23.7 \cmark 
     & 
     108.7 \cmark  &	
      588.3 \cmark &	
     	1484.2 \cmark & 
      4548.5 \cmark	 &	
     -	 &
     - &
     -  \\ 
     \multicolumn{1}{c}{} & 30 & 
      4.7 \cmark &
     	32.3 \cmark &
      217.5 \cmark	&
     	1451.9 \cmark &
     	3357.3 \cmark &
     - &
     -	 &
     -	 &
     - 	  \\
    \cline{1-11}
  \cline{1-11}
  \toprule
  \multicolumn{2}{c}{}  &   \multicolumn{9}{c}{\NAME\ ae2sbvzot}  \\  
    \midrule
  \multicolumn{2}{c}{}  &    \multicolumn{9}{c}{n}  \\
  \midrule
\multicolumn{1}{c}{}  & k &    \multicolumn{1}{c}{\textbf{2}} & \multicolumn{1}{c}{\textbf{3}} & \multicolumn{1}{c}{\textbf{4}} & \multicolumn{1}{c}{\textbf{5}} & \multicolumn{1}{c}{\textbf{6}} & \multicolumn{1}{c}{\textbf{7}} & \multicolumn{1}{c}{\textbf{8}} & \multicolumn{1}{c}{\textbf{9}} & \multicolumn{1}{c}{\textbf{10}}  \\
   \toprule
 \multirow{5}{*}{\rotatebox[origin=c]{90}{$\livewaitone$}}  & 10 &  0.7 \cmark   &  0.7 \cmark &  0.7 \cmark &  0.9 \cmark &  0.8 \cmark &  0.9 \cmark &  2.3 \cmark &  1.0 \cmark &  1.3 \cmark \\
    \multicolumn{1}{c}{} & 15 &  0.7 \cmark &  0.8 \cmark &  0.9 \cmark &  1.1 \cmark &  0.9 \cmark &  1.1 \cmark &  11.8 \cmark & 1.8  \cmark &  1.8 \cmark  \\
    \multicolumn{1}{c}{} & 20 &  0.7 \cmark &  0.8 \cmark  &  1.0 \cmark  & 1.4  \cmark  & 1.5  \cmark &  2.1 \cmark  & 2.0  \cmark  & 2.1  \cmark &  3.4 \cmark \\
    \multicolumn{1}{c}{} & 25 &  0.7 \cmark &  1.2 \cmark  &  11.2 \cmark &  79.9 \cmark &  1.9 \cmark &  2.6 \cmark & 4.5 \cmark  & 3.2  \cmark &  787.9 \cmark   \\
     \multicolumn{1}{c}{} & 30 &  0.9 \cmark &  1.2 \cmark &  33.8 \cmark  &  1.6 \cmark &  2.0 \cmark &  4.3 \cmark &  2.1 \cmark & 4.7 \cmark  & 4.8  \cmark   \\
  \cline{1-11}
   \hline
  \multirow{5}{*}{\rotatebox[origin=c]{90}{$\livewaittwo$}}  & 10 & 1.1  \cmark &  1.3  \cmark &  1.7  \cmark &   1.4 \cmark &  1.6  \cmark  &  1.5  \cmark  &   1.7 \cmark  &  1.9 \cmark  &   2.2  \cmark   \\
    \multicolumn{1}{c}{} & 15 &   3.3 \cmark &  3.7  \cmark  &  4.6  \cmark &  6.4 \cmark &   16.5 \cmark &   18.7 \cmark  &  28.6  \cmark &   44.8 \cmark &   87.2 \cmark  \\
    \multicolumn{1}{c}{} &  20 &   5.5 \cmark &  13.7  \cmark &  18.9  \cmark &  80.6  \cmark &   54.2 \cmark &   197.5 \cmark &   127.5 \cmark &   60.2 \cmark & 
    1417.3   \cmark \\
    \multicolumn{1}{c}{} & 25 &  160.5  \cmark &   28.8 \cmark &  51.5  \cmark  &   93.5 \cmark  &   33.1 \cmark &   963.4 \cmark &  2135.7  \cmark &   7200.0 \cmark &    7200.0 \cmark \\
     \multicolumn{1}{c}{} & 30 &  69.4  \cmark &  24.1  \cmark  &  136.9  \cmark  &   50.7 \cmark &  449.7  \cmark &  659.1  \cmark & 1750.6$^\ast$\cmark &  1927.6$^\ast$\cmark  & -   \\
    \cline{1-11}
      \multirow{5}{*}{\rotatebox[origin=c]{90}{$\livewaitthree$}}  & 10 &  0.8 \xmark &  1.0 \xmark  &  1.1 \xmark  &   2.5 \xmark  &   4.3 \xmark  & 7.6 \cmark  & 15.0 \cmark &  14.9 \cmark  &   16.9 \cmark    \\
    \multicolumn{1}{c}{} & 15 & 1.3 \xmark &  1.1 \xmark &  1.8 \xmark &  2.0 \xmark & 13.4$^\ast$\xmark &  20.9 \xmark &  29.4 \xmark  &  17.3 \xmark  &  23.4 \xmark  \\
    \multicolumn{1}{c}{} & 20 & 1.0 \xmark &   3.3 \xmark  &  1.8 \xmark  &  3.5 \xmark &  6.1 \xmark  &  4.4 \xmark &  6.7 \xmark &  51.5 \xmark &  87.1 \xmark \\
    \multicolumn{1}{c}{} & 25 & 1.1 \xmark &  1.7 \xmark &  3.8 \xmark &  2.7  \xmark &  11.7 \xmark &  23.7 \xmark  &   520.6 \xmark  &   388.6 \xmark  &  241.5$^\ast$ \xmark \\
     \multicolumn{1}{c}{} & 30 & 1.3 \xmark &  2.1 \xmark &   3.7 \xmark &  16.8 \xmark  & 10.2  \xmark  &  49.1 \xmark &   126.6 \xmark &   112.9 \xmark  & 181.2  \xmark   \\
    \cline{1-11}
      \multirow{5}{*}{\rotatebox[origin=c]{90}{$\livewaitfour$}}  & 10 &  0.9 \cmark &  1.0 \cmark &   1.1 \cmark &  1.2 \cmark &  1.3 \cmark &  1.5 \cmark  & 1.5  \cmark &  2.0 \cmark &  1.8 \cmark    \\
    \multicolumn{1}{c}{} & 15 &  2.1 \cmark &  4.0 \cmark  & 3.1  \cmark & 8.3  \cmark & 8.8 \cmark  &  15.8 \cmark  & 35.4  \cmark & 45.2  \cmark &  82.3 \cmark  \\
    \multicolumn{1}{c}{} & 20 & 4.0  \cmark &  6.5 \cmark  & 9.0  \cmark &  30.2 \cmark & 35.8 \cmark  &  69.6 \cmark  & 19.8  \cmark &  410.7 \cmark & 1269.1  \cmark \\
    \multicolumn{1}{c}{} & 25 &  10.9 \cmark &  16.6 \cmark &  22.0 \cmark  & 50.8  \cmark &  357.2 \cmark  &  404.5 \cmark &  2618.5 \cmark  & 329.4  \cmark  &  335.3$^\ast$\cmark   \\
     \multicolumn{1}{c}{} & 30 & 16.1  \cmark &  53.0 \cmark & 322.0 \cmark  & 310.3  \cmark &  137.1 \cmark &  4530.2 \cmark  & 80.3$^\ast$\cmark & 1611.7  \cmark &  598.9$^\ast$\cmark   \\
    \cline{1-11}
      \multirow{5}{*}{\rotatebox[origin=c]{90}{$\livewaitfive$}}  & 10 &  0.8 \xmark &  1.2 \xmark  &  2.9 \xmark & 1.3$^\ast$\cmark   &  3.1$^\ast$\cmark   &  17.9 \cmark  &  16.5 \cmark  & 17.8  \cmark  & 22.8  \cmark   \\
    \multicolumn{1}{c}{} & 15  &  0.9 \xmark &  1.3 \xmark  &  2.6 \xmark  &  2.3 \xmark & 7.5$^\ast$\xmark  &  19.3 \xmark &  25.9 \xmark  &  26.6 \xmark &  26.6$^\ast$\xmark \\
    \multicolumn{1}{c}{} & 20 &  2.6 \xmark &  1.7 \xmark  & 5.2  \xmark  &  6.7 \xmark &   4.7 \xmark & 16.7$^\ast$\xmark  & 36.2 \xmark  & 179.5$^\ast$\xmark  & 36.8  \xmark\\
    \multicolumn{1}{c}{} & 25 &  1.7 \xmark &  4.0 \xmark &  7.0 \xmark &  6.1 \xmark  & 12.1  \xmark &  7.1 \xmark &  493.8 \xmark & 378.7$^\ast$\xmark &  308.5 \xmark   \\
     \multicolumn{1}{c}{} & 30 & 1.5 \xmark  &  3.8 \xmark &  4.7 \xmark &  12.2 \xmark &  24.1 \xmark &  36.3 \xmark &  53.8 \xmark  & 57.3  \xmark  & -   \\
    \cline{1-11}
      \multirow{5}{*}{\rotatebox[origin=c]{90}{$\livewaitsix$}}  & 10 &   0.8 \cmark &  0.9 \cmark &  1.4 \cmark &   2.5 \cmark  & 4.9  \cmark & 3.3  \cmark  & 8.7  \cmark &  12.0 \cmark & 19.0  \cmark      \\
    \multicolumn{1}{c}{} & 15  & 0.8  \cmark &  1.6 \cmark &  3.4 \cmark &  5.2 \cmark &  21.1 \cmark  &  39.8 \cmark &  65.5 \cmark &  112.3 \cmark &  205.6 \cmark \\
    \multicolumn{1}{c}{} & 20 &  1.1 \cmark & 3.5  \cmark &  6.9 \cmark &  17.2\cmark & 88.2  \cmark &  170.5 \cmark & 567.4  \cmark &  1289.4 \cmark &  2521.5 \cmark \\
    \multicolumn{1}{c}{} & 25 &  1.5 \cmark & 6.4  \cmark &  25.4 \cmark &  75.1 \cmark &  22.2 \cmark  &  672.1 \cmark  & 2835.6  \cmark &  
    3416.2$^\ast$\cmark & 5861.4$^\ast$\cmark     \\
     \multicolumn{1}{c}{} & 30 &  1.9 \cmark & 11.7  \cmark &  35.1 \cmark  & 206.6  \cmark & 1161.1  \cmark &  4195.9 \cmark &  
	3945.9$^\ast$\cmark     
      &  $-$ &  $-$   \\
    \cline{1-11}
   \hline
  \cline{1-11}
\end{tabular}
\end{table}

The results are presented in Table~\ref{fig:scalabilityresultsFisher}.
The rows  show the time (in seconds) required by the model-checking procedure with different bounds $k$.
Symbol ``$-$"
indicates a timeout.
The columns contain the results obtained by considering an increasing number $n$ of  participants.
For properties $\livewaitone$, $\livewaittwo$, $\livewaitfour$ and $\livewaitsix$ \NAME\ always returned the correct result.
For properties $\livewaitthree$ and $\livewaitfive$,  when $k$ was equal to $10$,
the value of $n$ is too large, compared to the bound $k$, to allow \NAME\ to find a counterexample (i.e., to allow the underlying satisfiability procedure for \logic\ to detect a contradiction).
However, increasing the bound allows  \NAME\ to detect the counterexample.

When \aetwozot\ is used, the time required by \NAME\ to verify models increases as the values of $n$ and $k$ increase,
whereas the results obtained with \aetwosbvzot\ are less homogeneous.
Indeed, even if \aetwosbvzot\ is in general more efficient than \aetwozot, some tests carried out by \aetwosbvzot\ resulted in a timeout.
For example, this is the case of property $\livewaitthree$ with $k=15$ and $n=6$ which, conversely, has been successfully solved by \aetwozot.
To understand the reason of this result and, in particular, whether it was caused by the adopted Zot plugin, two different versions of Z3 were compared with each other.
A general variation of the performance of \NAME---even when using the same plugin---is evident, and it stems from the tactics that are used by Z3 to solve the satisfiability problem.
In fact, the experiments using version $4.4.1$ of Z3 show that \NAME\ is able to complete the verification for the cases that resulted in timeouts using version $4.4.1$, though it timed out in others.
This evidence proves that the choice of the Zot plugin does not determine the presence or absence of the timeouts.
In Table~\ref{fig:scalabilityresultsFisher}, the results that were obtained with version $4.4.1$ of Z3 instead of version $4.7.1$ are marked with the $^\ast$ symbol.

\vspace{0.2cm}
\emph{CSMA/CD Protocol.} 
The CSMA/CD protocol (Carrier Sense, Multiple-Access with Collision Detection) aims at assigning a  bus to one of $n$ competing stations.
When a station has data to send, it first listens to the bus. 
If no other station is transmitting (the bus is idle), the station begins the transmission. 
If another station is transmitting (the bus is busy), it waits a random amount of time and then repeats the previous steps.
The synchronization among the participants is obtained through a broadcast transition-based synchronization.

The following MITL property was tested. It is inspired by the one considered in the Uppaal benchmark~\cite{larsen1997uppaal}.
\begin{align}
	&\livecsmacd &&:=&&
	  \LTLg_{[0,\infty)} ( 
	  P_1.start\_send
        \rightarrow  (\neg collision\_after\_transm)) & \\    
&P_1.start\_send &&:=&&  (\neg P_1.send) \wedge  
            (P_1.send \LTLu_{(0,inf)}    \top )
        )\label{formula:aaaa}\\ 
& collision\_after\_transm &&:=&& 
               \LTLg_{(0,52]} (
                    P_1.send \wedge
                        ( 
                            P_1.send
						 \LTLu_{[0,inf)}                             
                            (P_1.send \wedge P_2.send)
                        )
                ) \label{formula:bbb}
\end{align} 
The property $\livecsmacd$ predicates on the occurrence of
a collision---i.e., $P_1$ and $P_2$ simultaneously sending a message. 
It specifies that a collision does not occur after
 $P_1$ is transmitting for  $52$ time units or more.
 Let us consider formula $P_1.start\_send$ (\ref{formula:aaaa}).
Formula  $\neg P_1.send$  specifies that $P_1$ is not sending at the current time $t$.
Formula $P_1.send \LTLu_{[(0,\infty)}    \top $ specifies that there exists a $t^{\prime}$
such that, for every $t^{\prime\prime}$ s.t. $t <t^{\prime\prime} < t^{\prime}$ holds, $P_1.send$ holds.
 Thus, formula $P_1.start\_send$ is true when $P_1$ starts sending a message.
 Let us now consider formula $collision\_after\_transm$ (\ref{formula:bbb}), which  specifies that  $P_1$ transmits for  $52$ time units or more, and then a collision is detected.
Operator $\LTLg_{(0,52]}$ forces formula $   (P_1.send) \wedge
                        ( 
                            P_1.send
						 \LTLu_{[0,\infty)}                           
                            (P_1.send \wedge P_2.send)
                        )
                    )$ to hold continuously from the current time instant, until $52$ time units from now, included.
 Since the formula must also hold at time instant $52$ from now, it forces a collision to be detected at a time instant that is after $52$ time units from now.
Furthermore, since the formula must hold in interval $(0,52]$, $P_1$ must keep sending a message within this interval.
The results of the experiments are presented in Table~\ref{fig:csmacd}.

\begin{table}[t]
\small
\center
\caption {Time (s) required to check the property of the CSMA/CD Protocol.
The symbol \cmark\ indicates that the property is satisfied, i.e., the \logic\	formula is unsatisfiable.
The symbol \xmark\ indicates that the property is not satisfied, i.e., the \logic\	formula is satisfiable.}
\label{fig:csmacd}
\begin{tabular}{ 
r  r  r  
r  r  r  
r  r  r 
r  r  
}
\toprule
  \multicolumn{2}{c}{} &   \multicolumn{9}{c}{\NAME\ ae2zot}  \\  
  \midrule
  \multicolumn{2}{c}{}  &    \multicolumn{9}{c}{$n$}  \\
  \midrule
\multicolumn{1}{c}{}  & k &    \multicolumn{1}{c}{\textbf{2}} & \multicolumn{1}{c}{\textbf{3}} & \multicolumn{1}{c}{\textbf{4}} & \multicolumn{1}{c}{\textbf{5}} & \multicolumn{1}{c}{\textbf{6}} & \multicolumn{1}{c}{\textbf{7}} & \multicolumn{1}{c}{\textbf{8}} & \multicolumn{1}{c}{\textbf{9}} & \multicolumn{1}{c}{\textbf{10}}  \\
\toprule
\multirow{5}{*}{\rotatebox[origin=c]{90}{$\livecsmacd$}}  & 10  & 2.6  \cmark &  5.4  \cmark &   5.8  \cmark &  7.1  \cmark & 9.9  \cmark &  7.6  \cmark &  11.3  \cmark &  12.3  \cmark  & 16.0  \cmark  \\
    \multicolumn{1}{c}{} & 15 & 10.5  \cmark & 19.5 \cmark  &  23.1  \cmark  &  45.8  \cmark & 64.0  \cmark &  135.4  \cmark &  123.9  \cmark &   221.5  \cmark & 453.3  \cmark \\
    \multicolumn{1}{c}{} & 20   &  20.8  \cmark & 46.5  \cmark &  97.3  \cmark &  140.8  \cmark & 409.9  \cmark &  663.5  \cmark &  1146.6  \cmark &  1102.9  \cmark & 1299.4  \cmark \\
    \multicolumn{1}{c}{} & 25 &  81.2  \cmark &  125.4  \cmark &  220.3  \cmark &  387.9  \cmark &   1278.6  \cmark &  1959.1  \cmark &  4742.7  \cmark &  2820.2  \cmark &  7184.4  \cmark   \\
     \multicolumn{1}{c}{} & 30  &  98.8  \cmark &  389.4   \cmark &  868.0  \cmark &  1500.9  \cmark &   2195.1  \cmark & - & - & - & -   \\
  \toprule
  \multicolumn{2}{c}{}  &   \multicolumn{9}{c}{\NAME\ ae2sbvzot}  \\  
    \midrule
  \multicolumn{2}{c}{}  &    \multicolumn{9}{c}{$n$}  \\
  \midrule
\multicolumn{1}{c}{}  & k &    \multicolumn{1}{c}{\textbf{2}} & \multicolumn{1}{c}{\textbf{3}} & \multicolumn{1}{c}{\textbf{4}} & \multicolumn{1}{c}{\textbf{5}} & \multicolumn{1}{c}{\textbf{6}} & \multicolumn{1}{c}{\textbf{7}} & \multicolumn{1}{c}{\textbf{8}} & \multicolumn{1}{c}{\textbf{9}} & \multicolumn{1}{c}{\textbf{10}}  \\
   \toprule
      \multirow{5}{*}{\rotatebox[origin=c]{90}{$\livecsmacd$}}  & 10  &  1.6 \cmark &  1.7 \cmark & 2.1  \cmark &  2.2 \cmark &  2.9 \cmark &   2.8 \cmark &  3.3 \cmark &  3.8 \cmark & 3.9   \cmark   \\
    \multicolumn{1}{c}{}  &  15 &   4.4 \cmark & 6.9  \cmark &  8.8 \cmark &  7.6 \cmark &   8.4 \cmark & 24.3  \cmark &  32.9 \cmark &  24.2 \cmark & 21.1  \cmark     \\
    \multicolumn{1}{c}{}  & 20   & 9.0  \cmark & 15.5 \cmark & 12.0  \cmark & 26.2  \cmark & 32.3   \cmark &   35.2 \cmark &  49.4 \cmark &  75.6 \cmark &  65.0 \cmark   \\
    \multicolumn{1}{c}{}  & 25 &   19.2 \cmark &  21.7 \cmark &  45.3 \cmark &   68.9 \cmark &   107.4 \cmark & 143.6  \cmark &  178.5  \cmark & 267.2  \cmark &  245.7 \cmark     \\
     \multicolumn{1}{c}{}  & 30  &   70.3  \cmark & 68.7  \cmark &  151.6  \cmark &  130.0 \cmark &  438.9 \cmark & 328.9   \cmark & 4441.8  \cmark & 854.2  \cmark & 6649.4  \cmark    \\
   \bottomrule
  \cline{1-11}
\end{tabular}
\end{table}

\vspace{0.2cm}
\emph{Token Ring.} The token ring benchmark considers $n$ symmetric stations that are organized in a ring, plus one process that models the ring.
The ring moves the token on a given direction among the $n$ processes.
The processes  may hand back the token in a synchronous (high-speed) or an asynchronous (low priority) fashion.
The synchronization among the participants is obtained through a channel transition-based synchronization.

The following MITL property was tested. It is inspired by---though it is not the same as---the one considered in the Uppaal benchmark~\cite{larsen1997uppaal}.
\begin{align}
&\livetoken &&:=&& 
 \LTLg_{(0,\infty)} 
 \left(
        \neg \left( \left( ST_1.zsync \lor ST_1.zasync \lor ST_1.ysync \lor  ST_1.yasync \right)
            \wedge  \right. \right. &\nonumber\\
            && & & &
     \left. \left.   \left( ST_2.zsync \lor ST_2.zasync \lor ST_2.ysync  \lor ST_2.yasync \right) \right)
    \right)
	 &\nonumber
\end{align}
Property $livetoken$ specifies that two stations $ST_1$ and $ST_2$ can not simultaneously sync---i.e., while one of them is in a synchronize state the other must be idle.
The results are presented in Table~\ref{fig:token}.

The results presented in Tables~\ref{fig:scalabilityresultsFisher},~\ref{fig:csmacd} and~\ref{fig:token} 
show that in all the cases the verification time is reasonable for practical adoptions of the proposed verification technique.
Furthermore,  the proposed technique easily allows   considering different semantics---e.g., different synchronization mechanisms---without directly changing  the verification algorithm.

\begin{table*}[t]
\small
\center
\caption {Time (s) required to check the property of the Token Ring.
The symbol \cmark\ indicates that the property is satisfied, i.e., the \logic\	formula is unsatisfiable.
The symbol \xmark\ indicates that the property is not satisfied, i.e., the \logic\	formula is satisfiable.}
\label{fig:token}
\begin{tabular}{ 
r 
r  r  r  
r  r  r  
r  r  r 
r  r  
}

\toprule
  \multicolumn{2}{c}{} &   \multicolumn{9}{c}{\NAME\ ae2zot}  \\  
  \midrule
  \multicolumn{2}{c}{}  &    \multicolumn{9}{c}{$n$}  \\
  \midrule
\multicolumn{1}{c}{}  & k &    \multicolumn{1}{c}{\textbf{2}} & \multicolumn{1}{c}{\textbf{3}} & \multicolumn{1}{c}{\textbf{4}} & \multicolumn{1}{c}{\textbf{5}} & \multicolumn{1}{c}{\textbf{6}} & \multicolumn{1}{c}{\textbf{7}} & \multicolumn{1}{c}{\textbf{8}} & \multicolumn{1}{c}{\textbf{9}} & \multicolumn{1}{c}{\textbf{10}}  \\
\toprule
\multirow{5}{*}{\rotatebox[origin=c]{90}{$\livetoken$}}  & 10 & 0.9 \cmark & 1.1 \cmark & 1.3 \cmark & 2.1 \cmark & 1.9 \cmark & 2.1 \cmark & 2.1 \cmark  & 2.3 \cmark & 2.2 \cmark    \\
    \multicolumn{1}{c}{} & 15  & 1.5 \cmark & 1.5  \cmark & 1.8 \cmark & 2.2 \cmark & 3.9 \cmark & 4.8 \cmark & 3.7 \cmark & 3.2 \cmark & 9.0 \cmark  \\
    \multicolumn{1}{c}{} & 20 & 2.2 \cmark & 2.2 \cmark & 4.8 \cmark & 3.1 \cmark & 5.0  \cmark & 10.6 \cmark & 7.1 \cmark & 18.9 \cmark & 10.4 \cmark \\
    \multicolumn{1}{c}{} & 25 & 2.7 \cmark & 5.0  \cmark & 3.7 \cmark & 5.8 \cmark  & 5.7 \cmark & 24.3 \cmark & 25.6 \cmark  & 19.6 \cmark & 58.2 \cmark   \\
     \multicolumn{1}{c}{} & 30 & 6.0 \cmark & 9.9  \cmark & 6.9 \cmark  & 17.6 \cmark  & 27.3 \cmark & 36.3 \cmark & 43.8 \cmark  & 21.3 \cmark  & 36.0 \cmark   \\
  \toprule
  \multicolumn{2}{c}{}  &   \multicolumn{9}{c}{\NAME\ ae2sbvzot}  \\  
    \midrule
  \multicolumn{2}{c}{}  &    \multicolumn{9}{c}{$n$}  \\
  \midrule
\multicolumn{1}{c}{}  & k &    \multicolumn{1}{c}{\textbf{2}} & \multicolumn{1}{c}{\textbf{3}} & \multicolumn{1}{c}{\textbf{4}} & \multicolumn{1}{c}{\textbf{5}} & \multicolumn{1}{c}{\textbf{6}} & \multicolumn{1}{c}{\textbf{7}} & \multicolumn{1}{c}{\textbf{8}} & \multicolumn{1}{c}{\textbf{9}} & \multicolumn{1}{c}{\textbf{10}}  \\
   \toprule
      \multirow{5}{*}{\rotatebox[origin=c]{90}{$\livetoken$}}  & 10 & 0.9 \cmark & 0.9 \cmark & 0.9 \cmark  & 1.0 \cmark  & 1.1 \cmark  & 1.2 \cmark  & 1.3 \cmark  & 1.5 \cmark  &  1.6 \cmark    \\
    \multicolumn{1}{c}{} & 15 & 1.2 \cmark  & 1.1  \cmark & 1.1 \cmark  & 1.1 \cmark  & 1.2 \cmark  & 1.4 \cmark  & 1.5 \cmark  & 1.7 \cmark  & 1.7  \cmark  \\
    \multicolumn{1}{c}{} & 20 & 2.1 \cmark & 1.9 \cmark & 1.8 \cmark  & 1.6 \cmark  & 1.8 \cmark  & 1.8 \cmark  & 1.8 \cmark  & 2.1 \cmark  & 2.2 \cmark \\
    \multicolumn{1}{c}{} & 25 & 2.5 \cmark & 3.7 \cmark & 3.5 \cmark  & 3.1 \cmark  & 2.3 \cmark  & 2.4 \cmark  & 2.5 \cmark  & 2.9 \cmark  & 2.9  \cmark  \\
     \multicolumn{1}{c}{} & 30 &  3.6 \cmark & 5.6 \cmark & 4.8 \cmark  & 5.3 \cmark  & 4.1 \cmark  & 3.5 \cmark  & 3.0 \cmark  & 3.2 \cmark  & 3.9 \cmark   \\
    \cline{1-11}
   \hline
  \cline{1-11}
\end{tabular}
\end{table*}

\section{Conclusion}
\label{sec:conclusion}
This paper presented a flexible approach for checking networks of TA against properties expressed in MITL.
The technique relies on an intermediate  artifact---i.e., a \logic\  formula---in which both the model and the property are encoded.
The intermediate artifact is then evaluated using suitable satisfiability checkers.
The proposed technique addresses three main challenges: 
(i) it allows considering a signal-based semantics;
(ii) it allows verifying properties expressed using the \emph{full} MITL;
(iii) it allows easily adding new TA constructs and changing their semantics (e.g., synchronization mechanisms, liveness conditions and  edge constraints).

The technique has been implemented in an open source tool called \NAME\ (Timed Automata ChecKer), which is publicly available at \url{http://github.com/claudiomenghi/TACK}.
Evaluation is performed by assessing:
(i) the efficiency  of \NAME\ in verifying MITL properties of TA;
(ii) the possibility of considering different synchronization constructs and semantic constraints. 
The intermediate artifact is evaluated 
through a bounded model checking technique that relies on 
two different  solvers
available in the \Zot\ formal verification tool~\cite{BPR16}. 
Evaluation relies on three different benchmarks that have been used to evaluate similar artifacts (e.g.,~\cite{UPPAALWEB}), namely the Fischer mutual exclusion protocol \cite{abadi1994old}, the CSMA/CD protocol \cite{CSMACD} and the Token Ring protocol \cite{jain1994fddi}.
The results 
show that the verification time is reasonable for practical adoptions of the proposed verification technique and prove that the proposed technique easily allows   considering different semantics---e.g., different synchronization mechanisms---by simply adding and removing formulae in the intermediate  \logic\ encoding.

\bibliographystyle{splncs03}
\bibliography{submission} 

\newpage
\newpage
\section{Appendix}
\label{sec:appendix}

The encoding used in the experiments is simpler than the general one presented in Fig.~\ref{tab:automaton} of Sec.~\ref{sec:TA2CLTLoc} because it is tailored only to signals whose intervals are all left-open and right-closed.
In such a case, the encoding can be simplified as the distinction between the kind of transitions is no longer needed.

\subsection{Encoding of traces for left-open right-closed signals.}
The following Fig.~\ref{tab:automaton-lorc} shows the simplified encoding of the network traces which, however, still retains the structure of the general one.
In particular, the atom $\edge$ becomes irrelevant, and then it can be removed, and the formulae $\varphi_4$ and $\varphi_5$ are modified.

\begin{figure*}[h!]
	\centering
	\small
	\begin{tabular}{| l | l | l | l | l | l |}
		\hline
		\multicolumn{2}{| c |}{
			$\varphi_{1}\coloneqq \underset{ \autindex \in [1,K]}{\bigwedge} (\loc[\autindex]=0) $
		} 
		&
		\multicolumn{2}{| c |}{
			$\varphi_{2}  \coloneqq  \underset{\variable \in \variables}{\bigwedge} \var=\variablevaluationfunction^0(\variable)$ 
		}
		&
		\multicolumn{2}{| c |}{
			$\varphi_{3} \coloneqq 
			\underset{ 
				\begin{subarray}{l}
				\autindex \in [1,K]
				\end{subarray}
			}{\bigwedge}  Inv(\loc[\autindex])  $} \\
		\hline
		\multicolumn{6}{|c|}{
			$\varphi_{4} \coloneqq 
			\underset{ 
				\begin{subarray}{l}
				\autindex \in [1,K]\\
				q \in Q_\autindex
				\end{subarray}
			}{\bigwedge} \left(  (\loc[\autindex]=q \land \tr[\autindex] = \notr) \rightarrow \LTLx( r_1(Inv(q)) \right)  $}
		\\
		\hline
		\multicolumn{6}{| c |}{
			\begin{tabular}{c}
				$
				\varphi_{5} \coloneqq  
				\underset{ 
					\begin{subarray}{l}
					\autindex \in [1,K],t \in T_k
					\end{subarray}
				}{\bigwedge}
				\tr[k] = {t}
				\LTLimplication
				\left( \loc[k] = t^- \land  \phi_{\variableconstraint}  \wedge \LTLx (\loc[k]={t^+}   \wedge   \phi_{\clockconstraint} \wedge \phi_{\varassignement} \wedge   \phi_{\resettedclocks} \land \phi_{\action^{](}}(t^-,t^+,k) \right)$ \\ 
				$\phi_{\action^{](}}(a,b,i) \coloneqq Inv(a) \land r_2(Inv_w(b))$\\ \\
			\end{tabular}
		}
		\\
		\hline
		\multicolumn{6}{|c|}{
			$\varphi_{6} \coloneqq 
			\underset{ 
				\autindex \in [1,K], 
				q,q' \in Q_\autindex \mid
				q \neq q'
			}{\bigwedge}  \left( ((\loc[\autindex]=q) \wedge \LTLx (\loc[\autindex]=q')) \rightarrow 
			\underset{
				t \in T_\autindex,
				t^- = q,
				t^+ = q'}{\bigvee} (\tr[k] = t) \right)$}\\
		\hline
		\multicolumn{3}{| c |}{$\varphi_{7} \coloneqq \underset{\clock \in \clocks}{\bigwedge} \left( \LTLx(\clock_0 = 0 \vee \clock_1 = 0) \rightarrow 
			\underset{
				\begin{subarray}{c}
				\autindex \in [1,\numberOfTA]\\
				t \in T_\autindex \mid \clock \in t_\resettedclocks
				\end{subarray}
			}{\bigvee} \tr[k] = t \right)$}
		&
		\multicolumn{3}{| c |}{$\varphi_8 \coloneqq \underset{\variable \in \variables}{\bigwedge} \left( (\neg (\variable =\LTLx (\variable))) \rightarrow \underset{
				\begin{subarray}{c}
				\autindex \in [1,\numberOfTA]\\
				t \in T_\autindex \mid \variable \in U(t)
				\end{subarray}}{\bigvee} \tr[k] = t\right)$} \\
		\hline
	\end{tabular}
	\caption{Encoding of the automaton.}
	\label{tab:automaton-lorc}
\end{figure*}

\subsection{Encoding of left-open right-closed signals.}
The following Fig.~\ref{fig:bindingCLTLoc-rightclosed} shows the simplified encoding of the signals.
The formulae $\varphi_1$ and $\varphi_2$ are the same as those in the general encoding in Fig.~\ref{tab:bindingCLTLoc}.
The definition of the signal in every interval determined by the trace is simpler because it does not depend anymore on the kind of transition performed by the automata.
For instance, if automaton $\autindex$ is in $\loc[\autindex]$ at position $h$ then the over the interval $I_h$ and in its right end-point the atomic propositions  in the signal include those associated with $\loc[\autindex]$.
A similar argument holds for the value of integer values. 

\begin{figure*}[h!]
	\centering
	\begin{tabular}{| c | c | c | c |}
		\hline
		\multicolumn{2}{| c |}{$\mu_1 \coloneqq \LTLg \underset{a \in AP}{\bigwedge} \rest{a} \leftrightarrow  \underset{
				\begin{subarray}{l}
				\autindex \in (0,K], 
				q \in Q_k,
				a \in L(q)
				\end{subarray}
			}{\bigvee} \left( \loc[\autindex]=q \right)$} & 
		\multicolumn{2}{| c |}{$\mu_2   \coloneqq  \LTLg \underset{(n \sim d) \in AF}{\bigwedge} (\rest{n \sim d} \leftrightarrow n \sim d )$}
		\\
		\hline
		\multicolumn{2}{| c |}{ $\mu_3 \coloneqq \underset{
				\begin{subarray}{l} 
				\autindex \in (0,K],\\
				a \in AP
				\end{subarray}}{\bigwedge}  \first{a} \leftrightarrow   
			\underset{
				a \in L(q_{0,k})
			}{\bigvee} l[k]= q_{0,k}$}  &
		\multicolumn{2}{| c |}{$\mu_4 \coloneqq \underset{(n \sim d) \in AF}{\bigwedge} \left( \first{n \sim d} \leftrightarrow n \sim d \right)$} \\
		\hline
		\multicolumn{2}{| c |}{ $\mu_5 \coloneqq \underset{
				\begin{subarray}{l} 
				\autindex \in (0,K],\\
				a \in AP
				\end{subarray}}{\bigwedge}  \rest{a} \rightarrow  \LTLx(\first{a})$}  &	
		\multicolumn{2}{| c |}{ $\mu_6 \coloneqq \underset{
				\begin{subarray}{l} 
				\autindex \in (0,K],\\
				(n\sim d \in AF)
				\end{subarray}}{\bigwedge}  \rest{n\sim d} \rightarrow  \LTLx(\first{n \sim d})$}\\
		\hline
	\end{tabular}
	\caption{Encoding of left-open right-closed signals.}
	\label{fig:bindingCLTLoc-rightclosed}
\end{figure*}

\end{document}